\documentclass[12pt]{article}
\usepackage{amsmath,amsthm,amssymb}
\usepackage{bm}
\usepackage{graphicx,psfrag,epsf}
\usepackage{enumerate}
\usepackage{natbib}
\usepackage{url} 
\usepackage{float}
\usepackage{xr}
\usepackage{multirow}
\usepackage{numprint}

\newtheorem{assumption}{Assumption}
\newtheorem{theorem}{Theorem}
\newtheorem{lemma}{Lemma}
\newtheorem{proposition}{Proposition}

\begin{document}
	
	\def\spacingset#1{\renewcommand{\baselinestretch}%
		{#1}\small\normalsize} \spacingset{1}

\title{ \bf A Bayesian Approach to Multiple-Output Quantile Regression}
\author{Michael Guggisberg\thanks{
		The author gratefully acknowledges \textit{the School of Social Sciences at the University of California, Irvine and the Institute for Defense Analyses for funding this research. The author would also like to thank Dale Poirier, Ivan Jeliazkov, David Brownstone, Daniel Gillen, Karthik Sriram, and Brian Bucks for their helpful comments.}}\hspace{.2cm}\\
	Institute for Defense Analyses}
  \maketitle

\bigskip
\begin{abstract}
This paper presents a Bayesian approach to multiple-output quantile 
regression.  The unconditional model is proven to be consistent and asymptotically correct frequentist confidence intervals can be obtained. The prior for the unconditional model can be elicited as the ex-ante 
knowledge of the distance of the $\tau$-Tukey depth contour to the Tukey median, the first prior of its kind. A proposal for conditional regression is also presented. The model is applied to 
the 
Tennessee Project Steps to Achieving Resilience (STAR) 
experiment and it finds a joint increase in \emph{$\tau$-quantile 
subpopulations} for mathematics and reading scores given a decrease in the 
number 
of 
students 
per teacher.  This result is consistent with, and much stronger than, the result 
one 
would find with multiple-output linear regression.  Multiple-output linear 
regression finds the \emph{average} mathematics and reading scores increase 
given a decrease in the number of students per teacher. However, there could 
still be subpopulations where the score           
declines. The multiple-output quantile regression approach confirms
there are no quantile subpopulations (of the inspected subpopulations) where the score declines.  This is truly a 
statement of `no child left behind' opposed to `no average child left behind.'
\end{abstract}

\noindent%
\vfill

\newpage
\spacingset{1.45} 

\section{Introduction}
Single-output (i.e.\ univariate) quantile regression, originally proposed by \cite{koenker78}, 
is a popular method of inference among empirical researchers, see \cite{yu03} for a survey.  \cite{yu01} formulated 
quantile regression into a Bayesian framework. This advance opened 
the 
doors for Bayesian inference and generated a series of applied and 
methodological research.\footnote{For example, see 
	\cite{alhamzawi12,benoit12,BayesQR,feng15, 
	kottas09,kozumi11,lancaster10, rahman16,  sriram16, taddy12,  thompson10}.}

Multiple-output (i.e.\ multivariate) medians have been developing slowly since the early 1900s \citep{small90}.  A multiple-output quantile can be defined in many different ways and there has 
been 
little consensus on which is the most appropriate \citep{serfling02}. The 
literature for Bayesian multiple-output quantiles is sparse, only two papers exist and neither use a commonly 
accepted definition for a multiple-output quantile \citep{drovandi11,waldmann14}.\footnote{\cite{drovandi11} 
	uses 
	a 
	copula approach and
	\cite{waldmann14} uses a multiple-output asymmetric Laplace 
	likelihood approach.}   

This paper presents a Bayesian framework for multiple-output 
quantiles defined in \cite{hallin10}.  Their `directional' quantiles   
have theoretic and computational 
properties not enjoyed by many other definitions. These quantiles are unconditional and the quantile objective functions are averaged over the covariate space. See \cite{mckeague11} and \cite{zscheischler14} for frequentist applications of multiple-output quantiles. This paper also presents a Bayesian framework for conditional multiple-output quantiles defined in \cite{hallin15}. These approaches use an idea 
similar to \cite{chernozhukov03} which uses 
a 
likelihood that is not necessarily representative of the Data Generating 
Process (DGP).  However, the resulting posterior converges almost surely 
to the 
true value.\footnote{This is proven for the unconditional model and checked via simulation for the conditional model. Posterior convergence means that as sample size increases 
	the probability mass for the posterior is concentrated in 
	smaller neighborhoods around the true value.  Eventually converging to a 
	point mass at the true value.}  By performing inference in this 
framework one gains many advantages 
of a Bayesian analysis. The Bayesian machinery provides a principled way 
of combining prior knowledge with data to arrive at conclusions. This machinery 
can be used in a data-rich world, where data is continuously collected (i.e. online learning), to 
make inferences and update them in real time. The proposed approach can take more computational time than the frequentist approach since the proposed posterior sampling algorithm recommends initializing the Markov Chain Monte Carlo (MCMC) sequence at the frequentist estimate. Thus if the researcher does not desire to provide prior information or perform online learning, the frequentist approach may be more desirable than the proposed approach.

The prior is a required component in Bayesian analysis where the researcher 
elicits their pre-analysis beliefs for the population parameters.  The prior in 
unconditional model is closely related to the Tukey depth of a distribution 
\citep{Tukey75}.  Tukey depth is a notion of multiple-output centrality of a data 
point.  This is the first Bayesian prior for Tukey depth. The prior can be elicited as the Euclidean distance of the Tukey median from a (spherical) $\tau$-Tukey depth contour.  Once a prior is 
chosen, estimates can be computed using MCMC draws from the posterior.  If the 
researcher is willing to accept prior joint normality of the model parameters 
then a Gibbs MCMC sampler can be used. Gibbs samplers have many computational 
advantages over other 
MCMC algorithms such as easy implementation, efficient convergence to the stationary distribution and little to no parameter tuning.  Consistency of the posterior and a 
Bernstein-Von Mises result are verified via a 
small simulation study.  

The models are applied to 
the Tennessee Project Steps to Achieving Resilience (STAR) experiment \citep{finn90}.  The goal of the 
experiment 
was to 
determine if classroom size has an effect on learning 
outcomes.\footnote{Students were 
	randomly 
	selected to be in a small or large classroom for four years in their early 
	elementary education.  Every year the students were given standardized mathematics 
	and 
	reading tests.}  The effect of classroom size on test scores is shown comparing  $\tau$-quantile contours for mathematics and reading test scores for first grade 
	students in small and large classrooms. The model finds that $\tau$-quantile subpopulations
of mathematics and reading scores improve for both central and extreme students in smaller classrooms compared to larger classrooms.  
This result is consistent with, and much stronger than, the result one 
would find with multiple-output linear regression.  An analysis by multiple-output 
linear 
regression finds mathematics and reading scores improve \emph{on 
	average}, however there could still be subpopulations where the 
	score           
declines.\footnote{A plausible narrative is a poor performing student in a 
	larger classroom 
	might have 
	more free time due to the teacher being busy with preparing, organization 
	and 
	grading.  During 
	this free time the student might read more than they would have in a small 
	classroom and might perform better on the reading test than they would have 
	otherwise.}  The multiple-output quantile regression approach confirms 
there are no quantile subpopulations where the score declines (of the inspected subpopulations).  This is truly a 
statement of `no child left behind' opposed to `no average child left behind.'

\section{Bayesian multiple-output quantile regression}
This section presents the unconditional and conditional Bayesian approaches to quantile regression. Notation common to both approaches is first presented followed by the unconditional model and a theorem of consistency for the Bayesian estimator is presented (section \ref{sec:unc}). Then a method to construct asymptotic confidence intervals is shown (section \ref{sec:uncondconf}).  The prior for the unconditional model is then discussed (section \ref{sec:uncprior}). Last a proposal for conditional regression is presented (section \ref{sec:cond}). Expectations and probabilities in sections \ref{sec:unc}, \ref{sec:uncondconf} and \ref{sec:cond} are conditional on parameters. Expectations in section \ref{sec:uncprior} are with respect to prior parameters. Appendix \ref{app:quantreview} reviews frequentist single and multiple-output quantiles and Bayesian single-output quantiles.

Let 
$[Y_1,Y_2,...,Y_k]' = \mathbf{Y}$ be a $k$-dimension random vector.  The 
direction and magnitude 
of the 
directional quantile is defined 
by {\boldmath$\tau$} $\in\mathcal{B}^k = \{\mathbf{v}\in\Re^k: 
0<||\mathbf{v}||_2<1\}$.  Where $\mathcal{B}^k$ is a $k$-dimension 
unit ball centered at $\mathbf{0}$ (with center removed). Define $||\cdot||_2$ 
to be the $l_2$ norm.  The vector 
{\boldmath$\tau$}$=\tau\mathbf{u}$ can be broken down into
direction, 
$[u_1,u_2,...,u_k]'=\mathbf{u}\in\mathcal{S}^{k-1}=\{\mathbf{v}\in\Re^k: 
||\mathbf{v}||_2=1\}$  and magnitude, $\tau\in(0,1)$.   

Let 
$\mathbf{\Gamma_u}$ be a 
$k\times(k-1)$ matrix such that $[\mathbf{u}\,\vdots\,\mathbf{\Gamma_u}]$ is an 
orthonormal basis of $\Re^k$. Define $\mathbf{Y_u} = \mathbf{u'Y}$ and 
$\mathbf{Y_u^\perp} =\mathbf{\Gamma_u'Y}$. Let $\mathbf{X}\in\Re^p$ to be random covariates. Define the $i$th observation of the $j$th component of $\mathbf{Y}$ to be 
$\mathbf{Y}_{ij}$ and the $i$th observation of the $l$th covariate of 
$\mathbf{X}$ to be 
$\mathbf{X}_{il}$ where $i\in\{1,2,...,n\}$ and $l\in\{1,2,...,p\}$.

\subsection{Unconditional regression}\label{sec:unc}

Define 
$\Psi^u(a,\mathbf{b})=E[\rho_\tau(\mathbf{Y_u} - 
\mathbf{b_y'Y_u^\perp} - \mathbf{b_x' X} -a)]$ to be the objective function of interest.   The 
{\boldmath$\tau$}th unconditional quantile regression of $\mathbf{Y}$ on $\mathbf{X}$ (and an 
intercept) is $\lambda_{\bm{\tau}} = \{\mathbf{y}\in\Re^k: \mathbf{u'y} = 
\mathbf{\beta_{{\bm{\tau}} y}'\Gamma_u'y}+\beta_{{\bm{\tau}} 
	\mathbf{x}}'\mathbf{X}+\alpha_{\bm{\tau}}\}$ 
where
\begin{equation}\label{eq:multdefreg}
(\alpha_{\bm{\tau}},\mathbf{\beta_{\bm{\tau}}})=(\alpha_{\bm{\tau}},
\mathbf{\beta_{{\bm{\tau}}\mathbf{y}}},\mathbf{\beta_{{\bm{\tau}} 
		\mathbf{x}}})\in \underset{a,\mathbf{b_y},\mathbf{b_x}}{argmin} \,
\Psi^u(a,\mathbf{b}).
\end{equation}

The definition of the location case is embedded in definition 
(\ref{eq:multdefreg}) where $\mathbf{b_x}$ and $\mathbf{X}$ are of null 
dimension. Note that $\beta_{\bm{\tau}\mathbf{y}}$ is a function of 
$\mathbf{\Gamma}_\mathbf{u}$.  This relationship is of little importance, 
the uniqueness of $\beta_{\bm{\tau}\mathbf{y}}'\mathbf{\Gamma_u'}$ is of greater 
interest; 
which is unique under Assumption \ref{ass:cont} presented in the next section. Thus the choice of $\mathbf{\Gamma}_\mathbf{u}$ is unimportant as long as $[\mathbf{u}\,\vdots\,\mathbf{\Gamma_u}]$ is orthonormal.\footnote{However, the choice of $\mathbf{\Gamma_u}$ could possibly effect the efficiency of MCMC sampling and convergence speed of the MCMC algorithm to the stationary distribution.}

The population parameters satisfy two subgradient conditions

\begin{equation}\label{eq:subgrad1}
\left.\frac{\partial \Psi^u(a,\mathbf{b})}{\partial 
	a}\right|_{\alpha_{{\bm{\tau}} 
		},
	\beta_{{\bm{\tau}} }} = 
Pr(\mathbf{Y_u}-\beta_{{\bm{\tau}}
	\mathbf{y}}'\mathbf{Y_u^\perp}-\beta_{{\bm{\tau}} 
	\mathbf{x}}'\mathbf{X}-\alpha_{{\bm{\tau}}}\leq 0)-\tau=0
\end{equation}
and 
\begin{equation}\label{eq:subgrad2}
\left.\frac{\partial \Psi^u(a,\mathbf{b})}{\partial \mathbf{b}}\right|
_{\alpha_{{\bm{\tau}}
	},\beta_{{\bm{\tau}} }} = 
E[[\mathbf{Y_u^\perp}',\mathbf{X}']'1_{(\mathbf{Y_u}-\beta_{{\bm{\tau}}
	\mathbf{y}}' 
\mathbf{Y_u^\perp}
-\beta_{{\bm{\tau}} \mathbf{x}}'\mathbf{X}-\alpha_{{\bm{\tau}}}\leq 
0)}] - 
\tau 
E[[\mathbf{Y_u^\perp}',\mathbf{X}']'] =\mathbf{0}_{k+p-1}.
\end{equation}
The expectations need not exist if observations are in general position \citep{hallin10}. 

Interpretations of the subgradient conditions are presented in the Appendix A, one of which is new to the literature and will be restated here. The second subgradient condition can be rewritten as
\begin{align*}
E[\mathbf{Y}_{\mathbf{u}i}^\perp|\mathbf{Y_u}-\beta_{{\bm{\tau}}
	\mathbf{y}}' 
\mathbf{Y_u^\perp}
-\beta_{{\bm{\tau}} \mathbf{x}}'\mathbf{X}-\alpha_{{\bm{\tau}}}\leq 
0] &= E[\mathbf{Y}_{\mathbf{u}i}^\perp]\text{ for 
	all } i\in\{1,...,k-1\}\\
E[\mathbf{X}_{i}|\mathbf{Y_u}-\beta_{{\bm{\tau}}
	\mathbf{y}}' 
\mathbf{Y_u^\perp}
-\beta_{{\bm{\tau}} \mathbf{x}}'\mathbf{X}-\alpha_{{\bm{\tau}}}\leq 
0] &= E[\mathbf{X}_{i}]\text{ for 
	all } i\in\{1,...,p\}
\end{align*}	
This shows the probability mass center in the lower 
halfspace for the orthogonal response is equal to that of the
probability 
mass center in the entire orthogonal response space.   Likewise for the covariates, the probability mass center of being in the 
lower 
halfspace is equal to the probability 
mass center in the entire covariate space. Appendix A provides more background on multiple-output quantiles defined in \cite{hallin10}.

The Bayesian approach assumes \[\mathbf{Y_u}|\mathbf{Y_u^\perp}, 
\mathbf{X},\alpha_{\bm{\tau}},\mathbf{\beta_{\bm{\tau}}} \sim 
ALD(\alpha_{\bm{\tau}} + \mathbf{\beta_{{\bm{\tau}} y}'Y_u^\perp} + 
\mathbf{\beta_{{\bm{\tau}} x}' 
	X} 
,\sigma_{\bm{\tau}},\tau)\] whose density is
\[f_{\bm{\tau}}(\mathbf{Y}|\mathbf{X},\alpha_{\bm{\tau}},\beta_{\bm{\tau}}, 
\sigma_{\bm{\tau}}) = \frac{\tau(1-\tau)}{\sigma_{\bm{\tau}}}
exp(-\frac{1}{\sigma_{\bm{\tau}}}\rho_\tau(\mathbf{Y} - \alpha_{\bm{\tau}} - 
\mathbf{\beta_{{\bm{\tau}} y}'Y_u^\perp} - 
\mathbf{\beta_{{\bm{\tau}} x}' 
	X})). \]

The nuisance scale parameter, $\sigma_{\bm{\tau}}$, is fixed at 1.\footnote{The nuisance parameter is sometimes taken to be a free parameter in single-output Bayesian quantile regression \citep{kozumi11}. The posterior has been shown to still be consistent with a free nuisance scale parameter in the single-output model \citep{sriram13}. This paper will not attempt to prove consistency with a free nuisance scale parameter. Future research could follow the outline proposed in the single-output model and extend it to multiple-output model \citep{sriram13}.} The likelihood is
\begin{equation}\label{eq:unclk}
L_{\bm{\tau}}(\alpha_{\bm{\tau}},\beta_{\bm{\tau}}) = \prod_{i=1}^{n} f_{\bm{\tau}}(\mathbf{Y}_i|\mathbf{X}_i,\alpha_{\bm{\tau}},\beta_{\bm{\tau}},1).
\end{equation}

The ALD distributional assumption likely does not represent the DGP and is thus a misspecified distribution.  However, as more observations 
are obtained the posterior probability mass concentrates around neighborhoods 
of 
$(\alpha_{\bm{\tau}0},\beta_{\bm{\tau}0})$, where  $(\alpha_{\bm{\tau}0},
\beta_{\bm{\tau}0})$ satisfies (\ref{eq:subgrad1}) and (\ref{eq:subgrad2}). 
Theorem \ref{thm:consist} 
shows this 
posterior consistency.

The assumptions for Theorem \ref{thm:consist} are below.
\begin{assumption}\label{ass:ind}
	The observations $(\mathbf{Y}_i,\mathbf{X}_i)$ are independent and identically distributed (i.i.d.) with true measure 
	$\mathbf{P}_0$ for $i\in\{1,2,...,n,...\}$.
\end{assumption} 
The density of $\mathbf{P}_0$ is denoted $p_0$. Assumption \ref{ass:ind} states 
the observations are independent.  This 
still allows for dependence among the components within a given observation (e.g.\ heteroskedasticity that is a function of $\mathbf{X}_i$). The i.i.d.\ assumption is required for the subgradient conditions to be well defined.

The next assumption causes the subgradient conditions to exist and be unique ensuring the population 
parameters,$(\alpha_{{\bm{\tau}}0},\beta_{{\bm{\tau}}0})$, are well defined.\footnote{This assumption can be weakened \citep{serflingzuo10}.}

\begin{assumption}\label{ass:cont}
	The measure of $(\mathbf{Y}_i,\mathbf{X}_i)$ is continuous with respect to 
	Lebesgue measure, has connected support and admits finite first moments, 
	for 
	all $i\in\{1,2,...,n,...\}$.
\end{assumption}
The next assumption describes the prior.
\begin{assumption}\label{ass:prior}The prior, $\Pi_{\bm{\tau}}(\cdot)$, has 
	positive 
	measure for every open neighborhood of $(\alpha_{{\bm{\tau}}0},
	\beta_{{\bm{\tau}}0})$ and is
	
	a) proper, or
	
	b) improper but admits a proper posterior.
\end{assumption}

Case b includes the Lebesgue measure on $\Re^{k+p}$ (i.e.\ flat prior) as a 
special 
case \citep{yu01}. Assumption 
\ref{ass:prior} is 
satisfied using the joint normal 
prior suggested in section \ref{sec:uncprior}. 

The next assumption bounds the covariates and response variables.
\begin{assumption}\label{ass:finite}
	There exists a $c_x>0$ such that $|\mathbf{X}_{i,l}|<c_x$ for all 
	$l\in\{1,2,...,p\}$ and all $i\in\{1,2,....,n,...\}$.
	There exists a $c_y>0$ such that $|\mathbf{Y}_{i,j}|<c_y$ for all 
	$j\in\{1,2,...,k\}$ and all $i\in\{1,2,....,n,...\}$.
	There exists a $c_\Gamma>0$ such that 
	$\underset{i,j}{\sup}\left|[\mathbf{\Gamma_u}]_{i,j}\right|<c_\Gamma$.
\end{assumption}
The restriction on $\mathbf{X}$ is fairly mild in application, any given 
dataset will satisfy these restrictions.  Further $\mathbf{X}$ can be 
controlled by the researcher in some situations (e.g.\ experimental 
environments). The restriction on $\mathbf{Y}$ is more contentious.  However, like $\mathbf{X}$, any given dataset 
will satisfy this restriction.  The assumption on $\mathbf{\Gamma_u}$ is 
innocuous since $\mathbf{\Gamma_u}$ is chosen by the researcher, 
it is easy to choose such that all components are finite.

The next assumption ensures the Kullback Leibler minimizer is well defined.
\begin{assumption}\label{ass:finitekl}
	$	
	E\log\left(\frac{p_{0}(\mathbf{Y}_i,\mathbf{X}_i)}{f_{\bm{\tau}}(\mathbf{Y}
		_i|
		X_i,\alpha,\beta,1)}\right)< \infty$ for all $i\in\{1,2,...,n,...\}$.
\end{assumption}

The next assumption is to ensure the orthogonal response and covariate vectors 
are not 
degenerate.
\begin{assumption}\label{ass:nondegen}
	There exist vectors $\epsilon_Y>\mathbf{0}_{k-1}$ and 
	$\epsilon_X>\mathbf{0}_p$ 
	such that 
	\[Pr(\mathbf{Y}^\perp_{\mathbf{u}ij}>\epsilon_{Yj},\mathbf{X}_{il}>
	\epsilon_{Xl}, 
	\forall 
	j\in\{1,...,k-1\}, \forall l\in\{1,...,p\})=c_p\not\in\{0,1\}.\]
\end{assumption}
This assumption can always be satisfied with a simple location shift as long as 
each 
variable takes on at least two different values with positive joint probability. Let 
$U\subseteq\Theta$, 
define the posterior probability of $U$ to be

\[\Pi_{\bm{\tau}}(U|(\mathbf{Y}_1,\mathbf{X}_1),(\mathbf{Y}_2,\mathbf{X}_2),...,
(\mathbf{Y}_n,\mathbf{X}_n)) = 
\frac{\int_U \prod_{i=1}^{n} 
	\frac{f_{\bm{\tau}}(\mathbf{Y}_i|\mathbf{X}_i,\alpha_{\bm{\tau}},
		\beta_{\bm{\tau}}, 
		\sigma_{\bm{\tau}})}{f_{\bm{\tau}}(\mathbf{Y}_i|\mathbf{X}_i,
		\alpha_{\bm{\tau}0},
		\beta_{\bm{\tau}0}, 
		\sigma_{\bm{\tau}0})} d\Pi_{\bm{\tau}}(\alpha_{\bm{\tau}},
	\beta_{\bm{\tau}})}{\int_{\Theta}\prod_{i=1}^{n} 
	\frac{f_{\bm{\tau}}(\mathbf{Y}_i|\mathbf{X}_i\alpha_{\bm{\tau}},
		\beta_{\bm{\tau}}, 
		\sigma_{\bm{\tau}})}{f_{\bm{\tau}}(\mathbf{Y}_i| \mathbf{X}_i, 
		\alpha_{\bm{\tau}0},
		\beta_{\bm{\tau}0}, 
		\sigma_{\bm{\tau}0})} d\Pi_{\bm{\tau}}(\alpha_{\bm{\tau}},	
		\beta_{\bm{\tau}})}.\]

The main theorem of the paper can now be stated.

\begin{theorem}\label{thm:consist}
	Suppose assumptions \ref{ass:ind}, \ref{ass:cont}, \ref{ass:prior}a, 
	\ref{ass:finite} and \ref{ass:nondegen} hold or assumptions \ref{ass:ind}, 
	\ref{ass:cont}, \ref{ass:prior}b, 
	\ref{ass:finite}, \ref{ass:finitekl} and \ref{ass:nondegen}. Let        
	$U=\{(\alpha_{\bm{\tau}},\beta_{\bm{\tau}}): 
	|\alpha_{\bm{\tau}}-\alpha_{\bm{\tau}0}|<\Delta,|\beta_{\bm{\tau}}-
	\beta_{\bm{\tau}0}|< \Delta	\mathbf{1}_{k-1}\}$. Then
	$\lim\limits_{n\rightarrow\infty}\Pi_{\bm{\tau}}(U^c|(\mathbf{Y}_1,
	\mathbf{X}_1),..., (\mathbf{Y}_n,\mathbf{X}_n)) = 
	0$ $a.s.$ $[\mathbf{P}_0]$.
\end{theorem}

The proof is presented in Appendix B. The 
strategy of the proof follows very closely to the 
strategy used in the conditional
single-output model \citep{sriram13}. 
First construct an open set $U_n$ containing 
$(\alpha_{{\bm{\tau}}0},\beta_{{\bm{\tau}}0})$ for all $n$ that converges to 
$(\alpha_{{\bm{\tau}}0},\beta_{{\bm{\tau}}0})$, the population parameters. Define 
$B_n = 
\Pi_{\bm{\tau}}(U_n^c|(\mathbf{Y}_1,
\mathbf{X}_1),..., (\mathbf{Y}_n,\mathbf{X}_n)) $. 
To show convergence of $B_n$ to $B=0$ almost surely, it is sufficient to show 
$\lim\limits_{n\rightarrow\infty}\sum_{i=1}^{n}E[|B_n-B|^d]<\infty$ for some 
$d>0$,  using the Markov inequality and Borel-Cantelli lemma.  The Markov 
inequality states if $B_n-B\geq0$ then for any $d>0$
\[Pr(|B_n-B|>\epsilon)\leq \frac{E[|B_n-B|^d]}{\epsilon^d}\]
for any $\epsilon>0$.  The Borel-Cantelli lemma states
\[\text{if }\lim\limits_{n\rightarrow\infty}\sum_{i=1}^{n}Pr(|B_n-B|>\epsilon)<
\infty \text{ 
	then } 
Pr(\underset{n\rightarrow\infty}{\limsup} \, |B_n-B|>\epsilon)=0.\]
Thus by Markov inequality
\[\sum_{i=1}^{n}Pr(|B_n-B|>\epsilon)\leq 
\sum_{i=1}^{n}\frac{E[|B_n-B|^d]}{\epsilon^d}.\]
Since $\lim\limits_{n\rightarrow\infty}\sum_{i=1}^{n}E[|B_n-B|^d]<\infty$ 
then 
$\lim\limits_{n\rightarrow\infty}\sum_{i=1}^{n}Pr(|B_n-B|>\epsilon)<\infty$.  
By Borel-Cantelli
\[Pr(\underset{n\rightarrow\infty}{\limsup}\, |B_n-B|>\epsilon)=0.\]
To show $\lim\limits_{n\rightarrow\infty}\sum_{i=1}^{n}E[|B_n-B|^d]<\infty$, a set $G_n$ is created where $(\alpha_{\tau0},\beta_{\tau 0})\not\in G_n$.  Within 
this the expectation of the posterior numerator is less 
than 
$e^{-2n\delta}$ and 
the expectation of the posterior denominator is greater than $e^{-n 
	\delta}$ for some 
$\delta>0$. Then the expected value of the posterior is less than $e^{-n\delta 
}$, which is summable. 

\subsection{Confidence Intervals}\label{sec:uncondconf}
Asymptotic confidence intervals for the unconditional location case can be obtained using Theorem 4 from \cite{chernozhukov03} and asymptotic results from \cite{hallin10}.\footnote{A rigorous treatment would require verification of the assumptions of Theorem 4 from \cite{chernozhukov03}. \cite{yangwanghe15,sriram15} provide asymptotic standard errors for the single-output model.} Let $V_{\bm{\tau}} = V^{mcmc}_{\bm{\tau}}J_{\mathbf{u}}' V_{\bm{\tau}}^cJ_{\mathbf{u}}V^{mcmc}_{\bm{\tau}}$ where $J_{\mathbf{u}}$ is a $k$ by $k+1$ block diagonal matrix with blocks $1$ and $\Gamma_{\mathbf{u}}$, 
\[V_{\bm{\tau}}^c = \begin{bmatrix}
\tau(1-\tau) & \tau(1-\tau)E[\mathbf{Y}'] \\
\tau(1-\tau)E[\mathbf{Y}]  & Var[(\tau - 1_{(\mathbf{Y}\in H_{\bm{\tau}}^{-})})\mathbf{Y}]
\end{bmatrix},\]
and 
$V^{mcmc}_{\bm{\tau}}$ is the covariance matrix of MCMC draws times $n$. The values of $E[\mathbf{Y}]$ and $Var[(\tau - 1_{(\mathbf{Y}\in H_{\bm{\tau}}^{-})})\mathbf{Y}]$ are estimated with standard moment estimators where the parameters of $H_{\bm{\tau}}^{-}$ are estimated with the Bayesian estimate plugged in.  Then $\hat{\theta}_{\bm{\tau}i} \pm \Phi^{-1}(1-\alpha/2)\sqrt{V_{\bm{\tau}ii}/n}$ has a $1-\alpha$ coverage probability, where $\Phi^{-1}$ is the inverse standard normal CDF. Section \ref{sec:sim} verifies this in simulation.

\subsection{Choice of prior}\label{sec:uncprior}
A new model is estimated for each unique 
$\bm{\tau}$ and thus a prior is needed for each one.  This might seem like 
there is an overwhelming amount of ex-ante elicitation required if one wants to estimate many models. For example,  to estimate $\tau$-quantile (regression) contours (see Appendix A).\footnote{Section \ref{sec:compund} discusses how to estimate many models simultaneously.}  However, 
simplifications can be made to make elicitation easier.

Let $\bm{\mu}$ be the Tukey median of $\mathbf{Y}$, where the Tukey median is the point with maximal Tukey depth. See Appendix A for a discussion of Tukey depth and Tukey median. Define $\mathbf{Z}= \mathbf{Y}-\bm{\mu}$ to be the Tukey median centered transformation of $\mathbf{Y}$. Let $\alpha_{\bm{\tau}}$ and $\beta_{\bm{\tau}}$ be the parameters of the $\lambda_{\bm{\tau}}$ hyperplane for $\mathbf{Z}$.  If the prior is 
centered over $H_0:\alpha_{\bm{\tau}}=\alpha_\tau,\; \beta_{\bm{\tau}\mathbf{z}}=\mathbf{0}_{k-1} \text{ and } \beta_{\bm{\tau}\mathbf{x}} = \beta_{\tau\mathbf{x}}$ for 
all 
$\bm{\tau}$ (e.g.\ $E[\alpha_{\bm{\tau}}]=\alpha_\tau,\; E[\beta_{\bm{\tau}\mathbf{z}}]=\mathbf{0}_{k-1} \text{ and } E[\beta_{\bm{\tau}\mathbf{x}}] = \beta_{\tau\mathbf{x}}$) then the implied ex-ante belief is $\mathbf{Y}$ has spherical Tukey 
contours.\footnote{
	The null hypothesis $H_0:\alpha_{\bm{\tau}}=\alpha_\tau,\; \beta_{\bm{\tau}\mathbf{z}}=\mathbf{0}_{k-1} \text{ and } \beta_{\bm{\tau}\mathbf{x}} = \beta_{\tau\mathbf{x}}$ for 
	all 
	$\bm{\tau}$ is a sufficient condition for spherical Tukey depth contours. It may or may not be necessary.
	
	A 
	sufficient condition for a density to have spherical Tukey depth contours is for 
	the 
	PDF to have spherical density contours and that the PDF, with a 
	multivariate argument $\mathbf{Y}$, can be written as a monotonically 
	decreasing function of $\mathbf{Y}'\mathbf{Y}$
	\citep{dutta11}.  This condition is satisfied for the location family for 
	the 
	standard multivariate Normal, T and Cauchy. The distance of the Tukey 
	median from
	the $\tau$-Tukey depth contour for the multivariate standard normal 
	is 	$\Phi^{-1}(1-\tau)$.  Another distribution with spherical Tukey 
	contours 
	is the uniform 
	hyperball.  The distance of the Tukey median from 
	the $\tau$-Tukey depth contour for the uniform 
	hyperball is the value $r$ such that $arcsin(r) + r\sqrt{1-r^2}=\pi(0.5-\tau)$. 
	This 
	function is invertible for $r\in(0,1)$ and $\tau\in(0,.5)$ and can be 
	computed 
	using numerical approximations \citep{rousseeuw99}.}   Under the 
belief 
$H_0$, $|\alpha_{\bm{\tau}}+\beta_{\bm{\tau}\mathbf{x}}\mathbf{X}|$ 
is the Euclidean distance of the $\tau$-Tukey depth contour from the Tukey 
median.  Since the 
contours are spherical, the distance is the same for all $\mathbf{u}$. This result is obtained using Theorem \ref{proof:prior} (presented below) and the fact that the boundary of the intersection of upper quantile halfspaces corresponds to $\tau$-Tukey depth contours, see equation (\ref{eq:quantregion}) and the following text in Appendix A.  The proof for Theorem 2 is presented in Appendix C. A notable corollary is if $\beta_{\bm{\tau}\mathbf{x}} = \mathbf{0}_p$ or $\mathbf{X}$ has null dimension then the radius of the spherical $\tau$-Tukey depth contour is $|\alpha_{\tau}|$. Note if $\mathbf{X}$ has null dimension, $p=2$, and $\mathbf{Z}$ has a zero vector Tukey median then for any $\mathbf{u}\in\mathcal{S}^{k-1}$ the population $\alpha_{\bm{\tau}0}$ is negative for $\tau<0.5$ and the population $\alpha_{\bm{\tau}0}$ is positive for $\tau>0.5$.

A prior for $(\alpha_{\bm{\tau}},
\beta_{\bm{\tau}})$ centered over $H_0$ expresses the researcher's confidence in the hypothesis 
of 
spherical Tukey depth contours.  
A large prior variance allows for large departures from $H_0$. If $\mathbf{X}$ is of null dimension then the prior variance of $\alpha_{\tau}$ represents the uncertainty of the distance of the $\tau$-Tukey depth contour from the Tukey median. Further if the parameter space for $\alpha_{\bm{\tau}}$ is restricted to $\alpha_{\bm{\tau}}=\alpha_{\tau}$ for fixed $\tau$ then the prior variance of $\alpha_{\tau}$ represents the uncertainty of the distance of the spherical $\tau$-Tukey depth contour from the Tukey median. 

\begin{theorem}\label{proof:prior}
	Suppose i) $\alpha_{\bm{\tau}}=\alpha_\tau,\; \beta_{\bm{\tau}\mathbf{z}}=\mathbf{0}_{k-1} \text{ and } \beta_{\bm{\tau}\mathbf{x}} = \beta_{\tau\mathbf{x}}$ for all $\bm{\tau}$ with $\tau$ fixed and ii) $\mathbf{Z}$ has spherical Tukey depth contours (possibly traveling through $\mathbf{X}$) denoted by  $T_\tau$ with Tukey median at $\mathbf{0}_k$. Then 1) the radius of the $\tau$-Tukey depth contour is $d_\tau = |\alpha_{\tau} + \beta_{\tau \mathbf{x}}\mathbf{X}|$, 2) for any point $\tilde{\mathbf{Z}}$ on the $\tau$-Tukey depth contour the hyperplane $\lambda_{\tilde{\bm{\tau}}}$ with $\tilde{\mathbf{u}} = \tilde{\mathbf{Z}}/\sqrt{\tilde{\mathbf{Z}}'\tilde{\mathbf{Z}}}$ and $\tilde{\bm{\tau}} = \tau \tilde{\mathbf{u}}$ is tangent to the contour at $\tilde{\mathbf{Z}}$ and 3) the hyperplane $\lambda_{\bm{\tau}}$ for any $\mathbf{u}$ is tangent to the $\tau$-Tukey depth contour.
\end{theorem} 

Arbitrary priors not centered over 0 require a more detailed discussion. Consider the 2 
dimensional case ($k=2$).   There are two ways to think of appropriate 
priors for $(\alpha_{\bm{\tau}},\beta_{\bm{\tau}})$.  The 
first approach is a direct 
approach thinking of $(\alpha_{\bm{\tau}},\beta_{\bm{\tau}})$ as the intercept and slope of 
$\mathbf{Y}_\mathbf{u}$ 
against $\mathbf{Y}_{\mathbf{u}}^\perp$ and $\mathbf{X}$.\footnote{The value of 
	$\mathbf{Y}_\mathbf{u}$ is the scalar projection of $\mathbf{Y}$ in 
	direction $
	\mathbf{u}$ and $\mathbf{Y}_\mathbf{u}^\perp$ is the scalar projection of $
	\mathbf{Y}$ in the direction of the other (orthogonal) basis vectors.}  The 
second approach is thinking of the implied prior of 
$\phi_{\bm{\tau}} = 
\phi_{\bm{\tau}}(\alpha_{\bm{\tau}},\beta_{\bm{\tau}})$ as the intercept and slope of $Y_2$ against $Y_1$ and 
$\mathbf{X}$.  The second approach is presented in Appendix D.

In the direct approach the parameters relate directly to the subgradient 
conditions (\ref{eq:subgrad1}) and (\ref{eq:subgrad2}) and their effect in $\mathbf{Y}$ space.  A $\delta$ unit increase in $\alpha_{\bm{\tau}}$ results in a 
parallel 
shift 
in the hyperplane $\lambda_{\bm{\tau}}$ by $\frac{\delta}{u_2 - 
\beta_{{\bm{\tau}}\mathbf{y}}
u_2^\perp}$ units. A $\delta$ unit increase in $\beta_{\bm{\tau}\mathbf{x}l}$ results in a 
parallel 
shift 
in the hyperplane $\lambda_{\bm{\tau}}$ by $\frac{\delta\mathbf{X}_l}{u_2 - 
\beta_{{\bm{\tau}}\mathbf{y}}
u_2^\perp}$ units. 
When $\beta_{\bm{\tau}}=\mathbf{0}_{2+p-1}$ $
\lambda_{\bm{\tau}}$ is orthogonal to $\mathbf{u}$ (and thus                  
$\lambda_{\bm{\tau}}$ is 
parallel to $\Gamma_\mathbf{u}$). As $\beta_{\bm{\tau}\mathbf{y}}$ increases or decreases monotonically such that $|\beta_{\bm{\tau}\mathbf{y}}|
\rightarrow 
\infty$, $
\lambda_{\bm{\tau}}$ converges to $\mathbf{u}$ 
monotonically.\footnote{Monotonic 
	meaning either the outer or inner angular distance between $\lambda_{\bm{\tau}}$ and 
	$\mathbf{u}$  is 
	always decreasing for strictly increasing or decreasing $\beta_{\bm{\tau}
		\mathbf{y}}$.}  A $\delta$ unit increase in $
\beta_{\bm{\tau} \mathbf{y}}$ tilts the $\lambda_{\bm{\tau}}$ 
hyperplane.\footnote{Define 
	$slope(\delta)$ to be the slope of the hyperplane when $\beta$ is increased 
	by $
	\delta$.  The slope of the new hyperplane is $slope(\delta)=(u_2-(\beta+
	\delta)u_2^\perp)^{-1}(\delta u_1^\perp+(u_2-\beta u_2^\perp)slope(0)$}  
The 
direction of the tilt is 
determined by the vectors $\mathbf{u}$ and $\Gamma_\mathbf{u}$ and the sign of 
$\delta$.  The vectors $\mathbf{u}$ 
and 
$\Gamma_\mathbf{u}$ always form a $90^\circ$ and $270^\circ$ angle.  
For positive $\delta$, the hyperplane travels monotonically through the 
triangle formed by $\mathbf{u}$ and $\Gamma_{\mathbf{u}}$.  For negative $\delta$ the hyperplane 
travels monotonically in the opposite direction.

\begin{figure}[H]
	\centering
	\includegraphics[width=.9\linewidth]{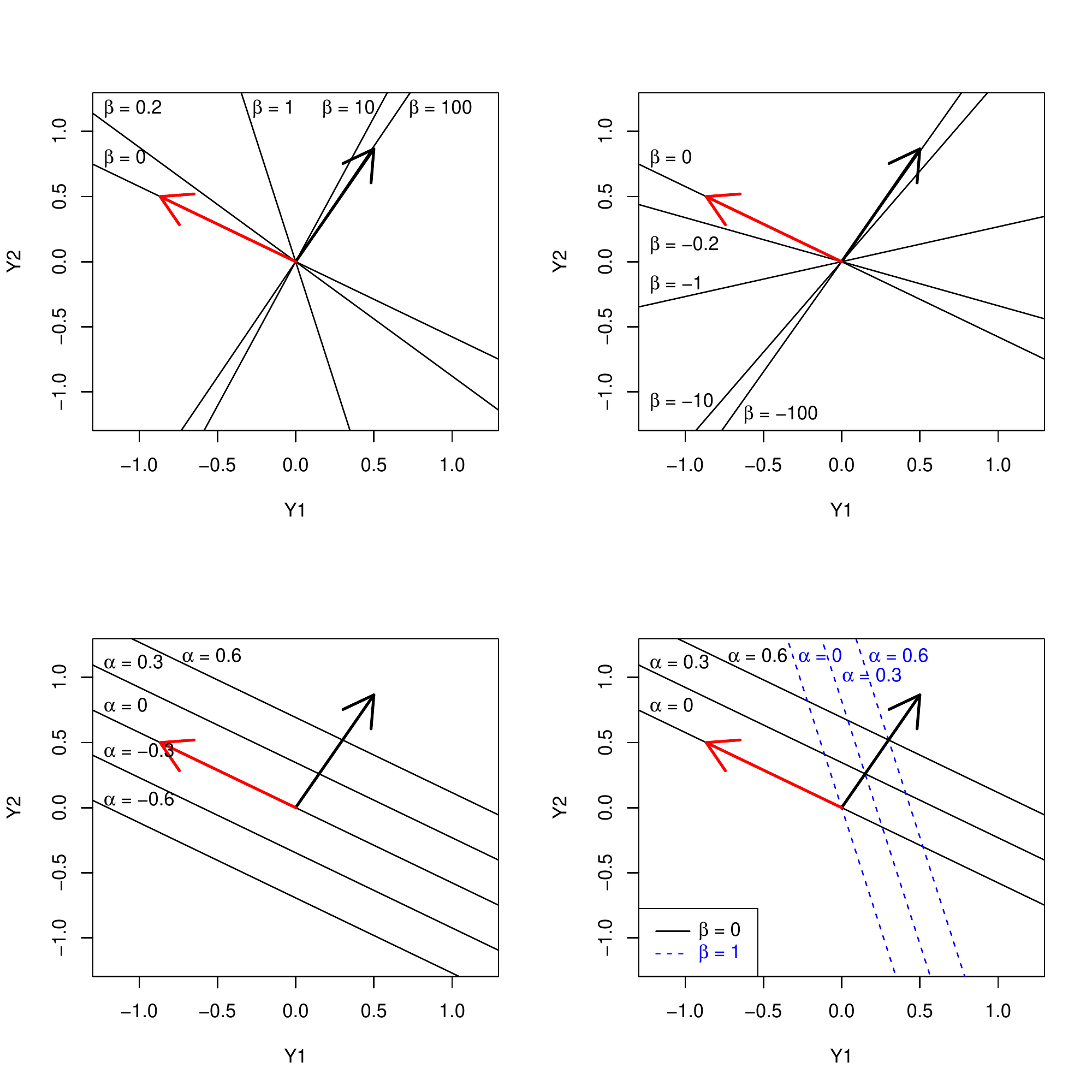}
	\caption[Implied $\lambda_{\bm{\tau}}$ from various hyperparameters (${\bm{\tau}}$ subscript 
	omitted)]{Implied $\lambda_{\bm{\tau}}$ from various hyperparameters (${\bm{\tau}}$ subscript 
		omitted). Top left, positive increasing $\beta$. Top right, negative decreasing $\beta$. 
		Bottom left, 
		different $\alpha$s. Bottom right, different $\alpha$s and $\beta$s.}
	\label{fig:priorpar}
\end{figure}

Figure \ref{fig:priorpar} shows prior $\lambda_{\bm{\tau}}$ implied from the center of the prior with various 
hyperparameters.  For all four plots $k=2$, the directional vector is 
$\mathbf{u}=(\frac{1}{\sqrt{2}},\frac{1}{\sqrt{2}})$ (black arrow) and 
$\Gamma_\mathbf{u}=(-\frac{1}{\sqrt{2}},\frac{1}{\sqrt{2}})$ (red arrow).  The top left 
plot shows $\lambda_{\bm{\tau}}$ for $\beta_{\bm{\tau}}$ increasing from $0$ to $100$ 
for fixed $\alpha_{\bm{\tau}}=0$.  At $\beta_{\bm{\tau}}=0$ the hyperplane is perpendicular 
to $\mathbf{u}$, as $\beta_{\bm{\tau}}$ increases $\lambda_{\bm{\tau}}$ travels counterclockwise until 
it becomes parallel to $\mathbf{u}$. The top right plot shows the $\lambda_{\bm{\tau}}$ 
for $\beta_{\bm{\tau}}$ decreasing from $0$ to $-100$ for fixed $\alpha_{\bm{\tau}}=0$.  At 
$\beta_{\bm{\tau}}=0$ $\lambda_{\bm{\tau}}$ is perpendicular to $\mathbf{u}$, as $\beta_{\bm{\tau}}$
decreases $\lambda_{\bm{\tau}}$ travels clockwise until it becomes parallel to 
$\mathbf{u}$.  The bottom left plot shows $\lambda_{\bm{\tau}}$ with 
$\alpha_{\bm{\tau}}$ ranging from $-0.6$ to $0.6$. The Tukey median can be 
thought of the point $(0,0)$, then $|\alpha_{\bm{\tau}}|$ is the distance of 
the intersection of $\mathbf{u}$ and $\lambda_{\bm{\tau}}$ from the Tukey 
median.\footnote{The Tukey median does not exist in these plots since there is 
no data. If there was data, the point where $\mathbf{u}$ and 
$\Gamma_\mathbf{u}$ intersect would be the Tukey median.}  For positive 
$\alpha_{\bm{\tau}}$ $\lambda_{\bm{\tau}}$ is moving in the direction $\mathbf{u}$ 
and for negative $\alpha_{\bm{\tau}}$ $\lambda_{\bm{\tau}}$ is moving in the 
direction $-\mathbf{u}$.  The bottom right plot shows $\lambda_{\bm{\tau}}$ for 
various $\alpha_{\bm{\tau}}$ and $\beta_{\bm{\tau}}$.  The solid black 
$\lambda_{\bm{\tau}}$ are for $\beta_{\bm{\tau}}=0$ and the dashed blue $\lambda_{\bm{\tau}}$ are 
for $\beta_{\bm{\tau}}=1$ and $\alpha_{\bm{\tau}}$ takes on values 0, 0.3 and 
0.6 for both values of $\beta_{\bm{\tau}}$.  This plot confirms changes in 
$\alpha_{\bm{\tau}}$ result in parallel shifts of $\lambda_{\bm{\tau}}$ while 
$\beta_{\bm{\tau}}$ tilts $\lambda_{\bm{\tau}}$.

If one 
is willing 
to accept joint normality of $(\alpha_{\bm{\tau}},
\beta_{\bm{\tau}})$ then a Gibbs sampler can be used.  The sampler is presented 
in Section \ref{sec:compund}. Further, if data is being collected and analyzed in 
real time, then the prior of the current analysis can be centered over the 
estimates from the previous analysis and the variance of the prior is the 
willingness the researcher is to allow for departures from the previous 
analysis.

\subsection{Conditional Quantiles}\label{sec:cond}
Quantile regression defined so far is unconditional on covariates. Thus the quantiles are averaged over the covariate space. A conditional quantile provides local quantile estimates conditional on covariates. A  Bayesian multiple-output conditional quantile can be defined from \cite{hallin15}.\footnote{This section omits theoretical discussion of multiple-output conditional quantiles. See \cite{hallin15} for a rigorous exploration of the properties (including contours) for multiple-output conditional quantiles.} A possible Bayesian approach is outlined but no proof of consistency is provided. Consistency is checked via simulation in Section \ref{sec:sim}. The $\lambda_{\bm{\tau}}$ hyperplanes are separately estimated for each conditioning value, thus the approach can be computationally expensive. Define $\mathcal{M}_{\mathbf{u}}  = \{(a,\mathbf{d}): a\in\Re, \; \mathbf{d}\in \Re^k \text{ subject to } \mathbf{d}'\mathbf{u}=1\}$ the population parameters are

\[
(\alpha_{\bm{\tau};\mathbf{x}_0},\delta_{\bm{\tau};\mathbf{x}_0}) = \underset{(a,\mathbf{d})\in\mathcal{M}_{\mathbf{u}}}{argmin} \; E[\rho_\tau (\mathbf{d}'\mathbf{Y} - a) | \mathbf{X} = \mathbf{x}_0].
\]

If $\mathbf{d}=\mathbf{u} - \mathbf{b}\Gamma_{\mathbf{u}}'$ the population objective function can be rewritten as $\Psi^c(a,\mathbf{b})=E[\rho_\tau (\mathbf{Y}_{\mathbf{u}} - \mathbf{b}' \mathbf{Y}_{\mathbf{u}}^\perp - a) | \mathbf{X} = \mathbf{x}_0]$. The population parameters are

\begin{equation}\label{eq:popobj2}
(\alpha_{\bm{\tau};\mathbf{x}_0},\beta_{\bm{\tau};\mathbf{x}_0}) = \underset{(a,\mathbf{b})\in\Re^k}{argmin} \; \Psi^c(a,\mathbf{b}).
\end{equation}

The subgradient conditions are

\begin{equation}\label{eq:subgradcond1}
\left.\frac{\partial \Psi^c(a,\mathbf{b})}{\partial 
	a}\right|_{\alpha_{\bm{\tau};\mathbf{x}_0},\beta_{\bm{\tau};\mathbf{x}_0}} = 
Pr(\mathbf{Y_u}-\beta_{{\bm{\tau}};\mathbf{x}_0}'\mathbf{Y_u^\perp}-\alpha_{{\bm{\tau}};\mathbf{x}_0}\leq 0 | \mathbf{X} = \mathbf{x}_0)-\tau=0
\end{equation}
and 
\begin{equation}\label{eq:subgradcond2}
\left.\frac{\partial \Psi^c(a,\mathbf{b})}{\partial \mathbf{b}}\right|
_{\alpha_{\bm{\tau};\mathbf{x}_0},\beta_{\bm{\tau};\mathbf{x}_0}} = 
E[\mathbf{Y_u^\perp}1_{(\mathbf{Y_u}-\beta_{{\bm{\tau}};\mathbf{x}_0}' 
	\mathbf{Y_u^\perp}-\alpha_{{\bm{\tau}};\mathbf{x}_0}\leq 
	0)} | \mathbf{X} = \mathbf{x}_0] - 
\tau 
E[\mathbf{Y_u^\perp} | \mathbf{X} = \mathbf{x}_0] =\mathbf{0}_{k-1}.
\end{equation}

Assuming the distribution of $\mathbf{X}$ is continuous then the conditioning set has probability 0. \cite{hallin15} creates an empirical (frequentist) estimator using weights providing larger weight to observations near $\mathbf{x}_0$. The estimator is

\begin{equation}\label{eq:condest}
\hat{\theta}_{\bm{\tau};\mathbf{x}_0} = \underset{\mathbf{b}}{argmin} \sum_{i=1}^{n} \;K_h(\mathbf{X}_i - \mathbf{x}_0) \rho_\tau (\mathbf{Y}_{\mathbf{u}i} - b\mathcal{X}_{\mathbf{u}i}^r) \,\text{ for } \, r = c,l.
\end{equation}

The function $K_h$ is a kernel function whose corresponding distribution has zero first moment and positive definite second moment (e.g.\ uniform, Epanechnikov or Gaussian). The parameter $h$ determines bandwidth.\footnote{To guarantee consistency of the frequentist estimator $h$ must satisfy $\lim\limits_{n\rightarrow\infty} h = 0$ and $\lim\limits_{n\rightarrow\infty}nh_n^{p-1} = \infty$. \cite{hallin15} provides guidance for choosing $h$.} If $r=c$ then $\mathcal{X}_{\mathbf{u}i}^c = [1,\mathbf{Y}_{\mathbf{u}i}^{\perp \prime} ]'$ and the estimator is called a local constant estimator. If $r = l$ then $\mathcal{X}_{\mathbf{u}i}^l = [1,\mathbf{Y}_{\mathbf{u}i}^{\perp\prime}]'\otimes  [1,(\mathbf{X}_{i}-\mathbf{x}_0)']'$ and the estimator is called a local bilinear estimator. The space that $\mathbf{b}$ is minimized over is the real numbers of dimension equal to the length of $\mathcal{X}_{\mathbf{u}i}^r$. For either value of $r$ the minimization can be expressed as maximization of an asymmetric Laplace likelihood with a known (heteroskedastic) scale parameter.

The Bayesian approach assumes \[\mathbf{Y_u}|\mathcal{X}^r,\theta_{\bm{\tau};\mathbf{x}_0} \sim 
ALD(\mathbf{\theta}_{\bm{\tau}}'\mathcal{X}^r,K_h(\mathbf{X} - \mathbf{x}_0)^{-1},\tau)\] whose density is
\begin{align*}
f_{\bm{\tau}}(\mathbf{Y}|\mathcal{X}^r,\theta_{\bm{\tau};\mathbf{x}_0}, 
K_h(\mathbf{X} - \mathbf{x}_0)^{-1}) &= \tau(1-\tau)K_h(\mathbf{X} - \mathbf{x}_0)
exp(-K_h(\mathbf{X} - \mathbf{x}_0)\rho_\tau(\mathbf{Y} - 
\theta_{\bm{\tau};\mathbf{x}_0}'\mathcal{X}^r))\\
&\propto exp(-K_h(\mathbf{X} - \mathbf{x}_0)\rho_\tau(\mathbf{Y} - 
\theta_{\bm{\tau};\mathbf{x}_0}'\mathcal{X}^r))
\end{align*}

If the researcher assumes the prior distribution for $\theta_{\bm{\tau}}$ is normal then the parameters can be estimated with a Gibbs sampler, which is presented in Section \ref{sec:compcond}.  
\section{MCMC simulation}
In this section a Gibbs sampler to obtain draws from the posterior distribution is presented for unconditional regression quantiles (Section \ref{sec:compund}) and conditional regression quantiles (Section \ref{sec:compcond}). 
\subsection{Unconditional Quantiles}\label{sec:compund}
Assuming joint normality of the prior distribution for the 
parameters estimation 
can be 
performed using draws from the posterior distribution obtained from a Gibbs sampler developed in 
\cite{kozumi11}.  The approach assumes $\mathbf{Y}_{\mathbf{u}i} = 
\mathbf{\beta_{{\bm{\tau}} y}'}\mathbf{Y}_{\mathbf{u}i}^\perp + 
\mathbf{\beta_{{\bm{\tau}} 
		x}' 
}\mathbf{X}_i + \alpha_{\bm{\tau}} + 
\epsilon_i$ where $\epsilon_i\overset{iid}{\sim} ALD(0,1)$. The random 
component, $\epsilon_i$, 
can 
be written as a mixture of a normal and an exponential, $\epsilon_i = \eta W_i + 
\gamma \sqrt{W_i} U_i$ where $\eta = \frac{1-2\tau}{\tau(1-\tau)}$, $\gamma = 
\sqrt{\frac{2}{\tau(1-\tau)}}$, $W_i \overset{iid}{\sim}  exp(1)$ and $U_i 
\overset{iid}{\sim}  N(0,1)$ are mutually 
independent \citep{kotz01}. This mixture representation allows for efficient simulation using data augmentation \citep{tannerwong87}.  It follows 
$\mathbf{Y}_{\mathbf{u}i}|\mathbf{Y}_{\mathbf{u}i}^\perp, 
\mathbf{X}_i,W_i,\beta_{\bm{\tau}},\alpha_{\bm{\tau}}$
is normally 
distributed. 
Further, if the prior is $\theta_{\bm{\tau}} = (\alpha_{\bm{\tau}},\beta_{\bm{\tau}}) 
\sim N(\mu_{\theta_{\bm{\tau}}},\Sigma_{\theta_{\bm{\tau}}})$ then $\theta_{\bm{\tau}}|\mathbf{Y}_{\mathbf{u}},\mathbf{Y}_{\mathbf{u}}^\perp,\mathbf{X}, W$ is normally distributed.  
Thus the $m+1$th 
MCMC draw is given by the following algorithm 
\begin{enumerate}
	\item Draw $W_i^{(m+1)}\sim 
	W|\mathbf{Y}_{\mathbf{u}i},\mathbf{Y}_{\mathbf{u}i}^\perp, 
	\mathbf{X}_i,\theta_{\bm{\tau}}^{(m)} \sim 
	GIG(\frac{1}{2},\hat{\delta}_i,\hat{\phi}_i) $  for $i \in \{1,...,n\}$
	
	\item Draw $\theta^{(m+1)}_{\bm{\tau}}\sim 
	\theta_{\bm{\tau}}|\vec{\mathbf{Y}}_{\mathbf{u}}, 
	\vec{\mathbf{Y}}_{\mathbf{u}}^\perp,
	\vec{\mathbf{X}}, \vec{W}^{(m+1)} \sim 
	N(\hat{\theta}_{\bm{\tau}},\hat{B}_{\bm{\tau}})$.
\end{enumerate}
where 
\begin{align*}
	\hat{\delta}_i &=\frac{1}{\gamma^2} (\mathbf{Y}_{\mathbf{u}i}  - 
	\mathbf{\beta'}^{(m)}_{{\bm{\tau}} \mathbf{y}} 
	\mathbf{Y}_{\mathbf{u}i}^\perp- 
	\mathbf{\beta'}^{(m)}_{{\bm{\tau}} \mathbf{x}} \mathbf{X}_i - 
	\alpha_{\bm{\tau}}^{(m)})^2\\
	\hat{\phi}_i &= 2 + \frac{\eta^2}{\gamma^2}\\
	\hat{B}^{-1}_{\bm{\tau}} &=  B_{{\bm{\tau}} 
		0}^{-1}+
	\sum_{i=1}^{n}\frac{[\mathbf{Y}_{\mathbf{u}i}^{\perp\prime},\mathbf{X}_i']
		[\mathbf{Y}_{\mathbf{u}i}^{\perp\prime},\mathbf{X}_i']'}{\gamma^2W_i^{(m+1)}}\\ 
	\hat{\beta}_{\bm{\tau}} &= \hat{B}_{\bm{\tau}} 
	\left(B_{{\bm{\tau}} 0}^{-1}\beta_{{\bm{\tau}} 0} + 
	\sum_{i=1}^{n}\frac{[\mathbf{Y}_{\mathbf{u}i}^{\perp\prime},\mathbf{X}_i']'
		(\mathbf{Y}_{\mathbf{u}i} - 
		\eta 
		W_i^{(m+1)})}{\gamma^2W_i^{(m+1)}} \right) 
\end{align*} 
and $GIG(\nu,a,b)$ is the 
Generalized 
Inverse Gamma distribution whose 
density is 
\[f(x|\nu,a,b) = 
\frac{(b/a)^\nu}{2K_\nu(ab)}x^{\nu-1}exp(-\frac{1}{2}(a^2x^{-1}+b^2x)), x>0, 
-\infty<\nu<\infty, a,b\geq 0\] and $K_\nu(\cdot)$ is the modified Bessel 
function of the third kind.\footnote{An efficient sampler of the Generalized Inverse 
	Gamma distribution was 
	developed in \cite{dagpunar89}. Implementations of the Gibbs 
sampler with a free $\sigma$ parameter for R are provided 
in the package `bayesQR' and `AdjBQR'  \citep{BayesQR, AdjBQR}. However, the results presented in this paper use a fixed $\sigma=1$ parameter.} To speed convergence the MCMC sequence can be initialized with the frequentist estimate.\footnote{The R package `quantreg' can provide such estimates \citep{quantreg}.} The Gibbs sampler is geometrically 
ergodic and thus the MCMC standard error is finite and the MCMC central limit 
theorem applies \citep{khare12}. This guarantees that after a long 
enough burn-in draws 
from this sampler are equivalent to random draws from the 
posterior.

Numerous other algorithms can be used if the prior is non-normal. \cite{kozumi11} provides a Gibbs sampler for when the prior is double exponential. \cite{li10} and \cite{alhamzawi12} provide algorithms for when regularization is desired. General purpose sampling schemes can also be used such as the Metropolis-Hastings, slice sampling or other algorithms \citep{hastings70,neal03,liu08}. 

The Metropolis-Hastings algorithm can be implemented as follows. Define the likelihood to be $L_{\bm{\tau}}(\theta_{\bm{\tau}}) = \prod_{i=1}^{n} f_{\bm{\tau}}(\mathbf{Y}_i|\mathbf{X}_i,\alpha_{\bm{\tau}},\beta_{\bm{\tau}}, 
1)$. Let the prior for $\theta_{{\bm{\tau}}}$ have the density $\pi_{\bm{\tau}}(\theta_{\bm{\tau}})$. Define $g(\theta^\dagger|\theta)$ to be a proposal density. The $m+1$th MCMC draw is given by the following algorithm

\begin{enumerate}
	\item Draw $\theta_{\bm{\tau}}^\dagger$ from $g(\theta^\dagger_{\bm{\tau}}|\theta^{(m)}_{\bm{\tau}})$
	\item Compute $A(\theta_{\bm{\tau}}^\dagger,\theta_{{\bm{\tau}}}^{(m)}) = min\left(1, \frac{L(\theta_{\bm{\tau}}^\dagger)\pi_{\bm{\tau}}(\theta_{\bm{\tau}}^\dagger)g(\theta^{(m)}_{\bm{\tau}}|\theta_{\bm{\tau}}^\dagger)}{L(\theta_{\bm{\tau}}^{(m)})\pi_{\bm{\tau}}(\theta_{\bm{\tau}}^{(m)})g(\theta_{\bm{\tau}}^\dagger|\theta_{\bm{\tau}}^{(m)})}\right)$
	\item Draw $u$ from $Uniform(0,1)$
	\item If $u\leq A(\theta_{\bm{\tau}}^\dagger,\theta_{{\bm{\tau}}}^{(m)})$ set $\theta_{\bm{\tau}}^{(m+1)} = \theta_{\bm{\tau}}^\dagger$, else set $\theta_{\bm{\tau}}^{(m+1)} = \theta_{\bm{\tau}}^{(m)}$
\end{enumerate}

Estimation of $\tau$-quantile contours (see Appendix \ref{app:quantreview}) requires the simultaneous estimation of several different $\lambda_{\bm{\tau}}$. Simultaneous estimation of multiple $\lambda_{\bm{\tau}_m}$ ($m\in\{1,2,...,M\}$) can be performed by creating an aggregate likelihood. The aggregate likelihood is the product of the likelihoods for each $m$, $L_{\bm{\tau}_1,\bm{\tau}_2,...,\bm{\tau}_M}(\alpha_{\bm{\tau}_1},\beta_{\bm{\tau}_1},\alpha_{\bm{\tau}_2},\beta_{\bm{\tau}_2},...,\alpha_{\bm{\tau}_M},\beta_{\bm{\tau}_M})=\prod_{m=1}^{M} L_{\bm{\tau}_m}(\alpha_{\bm{\tau}_m},\beta_{\bm{\tau}_m})$. The prior is then defined for the vector  $(\alpha_{\bm{\tau}_1},\beta_{\bm{\tau}_1},\alpha_{\bm{\tau}_2},\beta_{\bm{\tau}_2},...,\alpha_{\bm{\tau}_M},\beta_{\bm{\tau}_M})$. The Gibbs algorithm can easily be modified for fixed $\tau$ to accommodate simultaneous estimation. To estimate the parameters from various $\tau$, the values of $\eta$ and $\gamma$ need to be adjusted appropriately.

\subsection{Conditional Quantiles}\label{sec:compcond}
Sampling from the conditional quantile posterior is similar to that of unconditional quantiles except the likelihood is heteroskedastic with known heteroskedasticity. The approach assumes $\mathbf{Y}_{\mathbf{u}i} =  
\mathbf{\theta_{{\bm{\tau}}}' 
}\mathcal{X}_i^r + K_h(\mathcal{X}_i^r - \mathbf{x}_0)^{-1} \epsilon_i$ where $\epsilon_i\overset{iid}{\sim} ALD(0,1)$. The random 
component, $K_h(\mathcal{X}_i^r - \mathbf{x}_0)^{-1} \epsilon_i$, 
can 
be written as a mixture of a normal and an exponential, $K_h(\mathcal{X}_i^r - \mathbf{x}_0)^{-1} \epsilon_i = \eta V_i + 
\gamma \sqrt{K_h(\mathcal{X}_i^r - \mathbf{x}_0)^{-1}V_i} U_i$ where $V_i =K_h(\mathcal{X}_i^r - \mathbf{x}_0)^{-1}W_i$.
If the prior is $\theta_{\bm{\tau}} = (\alpha_{\bm{\tau}},\beta_{\bm{\tau}}) 
\sim N(\mu_{\theta_{\bm{\tau}}},\Sigma_{\theta_{\bm{\tau}}})$ then a Gibbs 
sampler 
can be used. The $m+1$th 
MCMC draw is given by the following algorithm
\begin{enumerate}
	\item Draw $V_i^{(m+1)}\sim 
	W|\mathbf{Y}_{\mathbf{u}i},\mathcal{X}_i^r,\theta_{\bm{\tau}}^{(m)} \sim 
	GIG(\frac{1}{2},\hat{\delta}_i,\hat{\phi}_i) $  for $i \in \{1,...,n\}$
	
	\item Draw $\theta^{(m+1)}_{\bm{\tau}}\sim 
	\theta_{\bm{\tau}}|\vec{\mathbf{Y}}_{\mathbf{u}}, 
	\vec{\mathbf{Y}}_{\mathbf{u}}^\perp,\vec{\mathcal{X}^r}, \vec{W}^{(m+1)} \sim 
	N(\hat{\theta}_{\bm{\tau}},\hat{B}_{\bm{\tau}})$.
\end{enumerate}
where 
\begin{align*}
\hat{\delta}_i &=\frac{K_h(\mathcal{X}_i^r - \mathbf{x}_0)}{\gamma^2} (\mathbf{Y}_{\mathbf{u}i}  - 
\mathbf{\theta'}^{(m)}_{\bm{\tau}} \mathcal{X}_i^r)^2\\
\hat{\phi}_i &= 2K_h(\mathcal{X}_i^r - \mathbf{x}_0) + \frac{\eta^2  K_h(\mathcal{X}_i^r - \mathbf{x}_0)}{\gamma^2}\\
\hat{B}^{-1}_{\bm{\tau}} &=  B_{{\bm{\tau}} 
	0}^{-1}+
\sum_{i=1}^{n}\frac{K_h(\mathcal{X}_i^r - \mathbf{x}_0)\mathcal{X}_i^r\mathcal{X}_i^{r\prime}}{\gamma^2W_i^{(m+1)}}\\ 
\hat{\beta}_{\bm{\tau}} &= \hat{B}_{\bm{\tau}} 
\left(B_{{\bm{\tau}} 0}^{-1}\beta_{{\bm{\tau}} 0} + 
\sum_{i=1}^{n}\frac{K_h(\mathcal{X}_i^r - \mathbf{x}_0)\mathcal{X}_i^r(\mathbf{Y}_{\mathbf{u}i} - 
	\eta 
	W_i^{(m+1)})}{\gamma^2W_i^{(m+1)}} \right).
\end{align*} 

The MCMC sequence can be initialized with the frequentist estimate.\footnote{The R package `quantreg' can provide such estimates using the weights option.} A Metropolis-Hastings algorithm similar to the unconditional model can be used where $L(\theta_{\bm{\tau}}) = \prod_{i=1}^{n} f_{\bm{\tau}}(\mathbf{Y}_i|\mathcal{X}_i^r,\theta_{\bm{\tau}}, 
K_h(\mathcal{X}_i^r - \mathbf{x}_0)^{-1})$. Simultaneous estimation of many $\lambda_{\bm{\tau}_m}$ is similar to the unconditional model.

\section{Simulation}\label{sec:sim}
This section verifies pointwise consistency of the unconditional and conditional models. Asymptotic coverage probability of the unconditional location model using the results from Section \ref{sec:uncondconf} is also verified.  Pointwise consistency is 
verified by checking convergence to solutions of the subgradient conditions (population parameters).   Four DGPs are considered.
\begin{enumerate}
	\item $\mathbf{Y} \sim Uniform \; Square$
	\item $\mathbf{Y} \sim Uniform \; Triangle$
	\item $\mathbf{Y} \sim N(\mu,\Sigma)$, where $\mu = \mathbf{0}_2$ and $
	\Sigma 
	=  \begin{bmatrix}
	1 & 1.5  \\
	1.5 & 9 
	\end{bmatrix}$
	\item $\mathbf{Y} = \mathbf{Z} + \begin{bmatrix}
	0   \\
	X 
	\end{bmatrix}$ where $\begin{bmatrix}
	X  \\
	\mathbf{Z}  
	\end{bmatrix} \sim 
	N\left(\begin{bmatrix}
	\mu_{X}   \\
	\mu_{\mathbf{Z}}  
	\end{bmatrix} ,
	\begin{bmatrix}
	\Sigma_{XX} & \Sigma_{X\mathbf{Z}}   \\
	\Sigma_{X\mathbf{Z}}' & \Sigma_{\mathbf{Z}\mathbf{Z}}  
	\end{bmatrix} \right)$, \\
	
	$\Sigma_{XX} =  4$,  $\Sigma_{X\mathbf{Z}} =  \begin{bmatrix}
	0 \\ 
	2  
	\end{bmatrix}$,  $\Sigma_{\mathbf{Z}\mathbf{Z}} =  \begin{bmatrix}
	1 & 1.5  \\
	1.5 & 9 
	\end{bmatrix}$, $\mu_{X} = 0$ 
	and $\mu_{\mathbf{Z}} = \mathbf{0}_2$
\end{enumerate}

The first DGP has corners at $(-\frac{1}{2},-\frac{1}{2}),(-\frac{1}{2},\frac{1}{2}),(\frac{1}{2},-\frac{1}{2}),(\frac{1}{2},\frac{1}{2})$.  The second DGP has 
corners at $(-\frac{1}{2},-\frac{1}{2\sqrt{3}}),(\frac{1}{2},-\frac{1}{2\sqrt{3}}),(0,\frac{1}{\sqrt{3}})$.  DGPs 1,2 and 3 are location 
models and 4 is a 
regression model. DGPs 1 and 2 conform to all the 
assumptions on the data generating process. DGPs 3 and 4 are cases when 
Assumption \ref{ass:finite} is violated.  In DGP 4, the unconditional 
distribution of $
\mathbf{Y}$ is  $\mathbf{Y}\sim N
\left(\begin{bmatrix}
0  \\
0 
\end{bmatrix} ,
\begin{bmatrix}
1 & 1.5  \\
1.5 & 17 
\end{bmatrix} \right)$. 

Two directions are considered, $\mathbf{u} =(\frac{1}{\sqrt{2}},\frac{1}
{\sqrt{2}})
$ and $\mathbf{u} =(0,1)$. The orthogonal directions are $\mathbf{\Gamma_{u}}=(1,0)$ and $\mathbf{\Gamma_{u}}=(1/\sqrt{2},-1/\sqrt{2})$. The first vector is a $45^o$ line between $Y_2$ 
and $Y_1$ in the positive quadrant and the second vector points vertically in the 
$Y_2$ direction. The depth is $\tau = 0.2$.  The sample sample sizes are 
$n \in \{10^2, 10^3, 10^4\}$.  The prior 
is 
$\theta_{\bm{\tau}}\sim 
N(\mu_{\theta_{\bm{\tau}}},\Sigma_{\theta_{\bm{\tau}}})$ 
where $\mu_{\theta_
	{\bm{\tau}}}=\mathbf{0}_{k+p-1}$ and $ \Sigma_{\theta_{\bm{\tau}}}=1000 
\mathbf{I}_{k+p-1}$.  The number of Monte Carlo simulations is 100 and for each Monte Carlo simulation 1,000 MCMC draws are used. The initial values are set to the frequentist estimate.

\subsection{Unconditional model pointwise consistency}
Consistency for the unconditional model is verified by checking
convergence to the solutions of the subgradient conditions (population parameters). Convergence of subgradient conditions (\ref{eq:subgrad1}) and (\ref{eq:subgrad2}) is verified in Appendix \ref{app:simsubgrad}.

The population parameters for the four DGPs are presented in Table \ref{tab:simres4}.\footnote{The population parameters are found by numerically minimizing the objective function. The expectation in the objective function is calculated with a Monte Carlo simulation sample of $10^6$.} The RMSE of the parameter estimates are presented in Tables \ref{tab:simres4a}, \ref{tab:simres4b} and \ref{tab:simres4c}. The results show the Bayesian estimator is converging to the population parameters.\footnote{Frequentist bias was also investigated and the bias showed convergence towards zero as sample size increased (no table presented).}

\begin{table}[htb]
	\centering
	\begin{tabular}{rr|cccc}
		& & \multicolumn{4}{c}{Data Generating Process}  \\ \hline
		$\mathbf{u}$ & $\theta$ & 1 & 2 & 3 & 4 \\ 
		\hline
	& $\alpha_{\bm{\tau}}$ & -0.26 & -0.20 & -1.17 & -1.16  \\ 
	$(1/\sqrt{2},1/\sqrt{2})$	&$\beta_{\bm{\tau}\mathbf{y}}$ & 0.00 & 0.44 & -1.14 & -1.17  \\ 
		&$\beta_{\bm{\tau}\mathbf{x}}$ &   &  &  &   -0.18  \\ 
		\hline
			& $\alpha_{\bm{\tau}}$ & -0.30 & -0.20 & -2.19 & -2.02  \\ 
		$(0,1)$	&$\beta_{\bm{\tau}\mathbf{y}}$ & 0.00 & 0.00 & 1.50 & 1.50 \\ 
		&$\beta_{\bm{\tau}\mathbf{x}}$ &   &  &  &   1.50  \\ 
		\hline
	\end{tabular}
	\caption{Unconditional model population parameters} 
	\label{tab:simres4}
\end{table}

\begin{table}[htb]
	\centering
	\begin{tabular}{rr|cccc}
				& & \multicolumn{4}{c}{Data Generating Process}\\ \hline
	$\theta$	& $n$ & 1 & 2 & 3 & 4 \\ 
		\hline
		& $10^2$ & 5.70e-02 & 4.41e-02 & 2.20e-01 & 1.83e-01 \\ 
		$\alpha_{\bm{\tau}}$ & $10^3$ & 1.49e-02 & 1.19e-02 & 6.80e-02 & 5.39e-02 \\ 
		& $10^4$ & 4.30e-03 & 3.66e-03 & 1.97e-02 & 1.85e-02 \\ 
		\hline
		& $10^2$ & 9.63e-02 & 2.79e-01 & 9.61e-02 & 1.08e-01 \\ 
		$\beta_{\bm{\tau}\mathbf{y}}$ & $10^3$ & 3.63e-02 & 6.58e-02 & 3.15e-02 & 3.15e-02 \\ 
		& $10^4$ & 1.19e-02 & 1.78e-02 & 1.07e-02 & 1.06e-02 \\ 
	\end{tabular}
	\caption{Unconditional model RMSE of parameter estimates $(\mathbf{u} = (1/\sqrt{2},1/\sqrt{2}))$} 
	\label{tab:simres4a}
\end{table}

\begin{table}[htb]
	\centering
	\begin{tabular}{rr|cccc}
		& & \multicolumn{4}{c}{Data Generating Process}\\ \hline
$\theta$  & $n$ & 1 & 2 & 3 & 4 \\ 
		\hline
		& $10^2$ & 3.57e-02 & 2.23e-02 & 3.47e-01 & 2.94e-01 \\ 
		$\alpha_{\bm{\tau}}$ & $10^3$ & 1.25e-02 & 5.59e-03 & 1.15e-01 & 1.13e-01 \\ 
		& $10^4$ & 4.23e-03 & 2.10e-03 & 3.27e-02 & 3.36e-02 \\ 
		\hline
		& $10^2$ & 1.16e-01 & 7.03e-02 & 3.94e-01 & 2.78e-01 \\ 
		$\beta_{\bm{\tau}\mathbf{y}}$ & $10^3$ & 3.96e-02 & 1.61e-02 & 1.18e-01 & 1.17e-01 \\ 
		& $10^4$ & 1.37e-02 & 6.73e-03 & 4.20e-02 & 3.13e-02 \\ 
	\end{tabular}
	\caption{Unconditional model RMSE of parameter estimates $(\mathbf{u} = (0,1))$} 
	\label{tab:simres4b}
\end{table}

\begin{table}[htb]
	\centering
	\begin{tabular}{rr|c}
		& & \multicolumn{1}{c}{Data Generating Process}\\ \hline
		$\mathbf{u}$  & $n$  & 4 \\ 
		\hline
		& $10^2$ & 1.58e-01 \\ 
		$(1/\sqrt{2},1/\sqrt{2})$ & $10^3$ & 4.86e-02 \\ 
		& $10^4$ & 1.48e-02  \\ 
		\hline
		& $10^2$ & 1.49e-01  \\ 
		$(0,1)$ & $10^3$ & 5.82e-02  \\ 
		& $10^4$ & 1.83e-02 \\ 
	\end{tabular}
	\caption{Unconditional model RMSE of $\beta_{\bm{\tau} \mathbf{x}}$ estimates} 
	\label{tab:simres4c}
\end{table}

\subsection{Unconditional location model coverage probability}

Coverage probabilities for the unconditional location model using the procedure in Section \ref{sec:uncondconf} are presented in Table \ref{tab:simcp}.  A correct coverage probability is 0.95. The number of Monte Carlo simulations is 300. The results show that the coverage probability tends to improve with sample size but has a slight under-coverage with sample size of $10^5$. A naive interval constructed from the $0.025$ and $0.975$ quantiles of the MCMC draws produces coverage probabilities ranging from $0.980$ to $1.000$, with a majority at $1$ (no table presented). This is clearly a strong over-coverage and thus the proposed procedure is preferred.

\begin{table}[htb]
	\centering
	\begin{tabular}{rrccc|ccc}
		& & \multicolumn{6}{c}{Data Generating Process}\\ 
		&\multicolumn{1}{c|}{}	& 1 & 2 & 3 & 1 & 2 & 3  \\ 
		\hline
		$\theta$	& \multicolumn{1}{c|}{$n$}  & \multicolumn{3}{c|}{$\mathbf{u} = (1/\sqrt{2},1/\sqrt{2})$} & \multicolumn{3}{c}{$\mathbf{u} = (0,1)$}\\ \hline
		& \multicolumn{1}{c|}{$10^2$} & .960 & .950 & .967 & 1.00  & 1.00 & .937 \\ 
		$\alpha_{\bm{\tau}}$ & \multicolumn{1}{c|}{$10^3$} & .963 & .950 & .940 & .960 & .963 & .953 \\ 
		& \multicolumn{1}{c|}{$10^4$} &  .910 & .947 & .967 & .930 & .963 & .957  \\ 
		\hline
		& \multicolumn{1}{c|}{$10^2$} & .810 & 1.00 & .953 & 1.00 & .987 & .937  \\ 
		$\beta_{\bm{\tau}\mathbf{y}}$ & \multicolumn{1}{c|}{$10^3$} & .907 & .967 & .960 & .978 & .970 & .933  \\ 
		& \multicolumn{1}{c|}{$10^4$} &  .933 & .953 & .937 & .927 & .960 & .950 \\ 
	\end{tabular}
	\caption{Unconditional location model coverage probabilities} 
	\label{tab:simcp}
\end{table}

\subsection{Conditional model pointwise consistency}

Convergence of the local constant conditional model is verified by checking convergence of the Bayesian estimator to the parameters minimizing the population objective function (\ref{eq:popobj2}). The local constant specification is presented because that is the specification used in the application. The conditional distribution of DGP 4 is $Y|X=x_0 \sim N\left(\begin{bmatrix}
0\\
x_0/2
\end{bmatrix},\begin{bmatrix}
1 & 1.5 \\
1.5 & 8
\end{bmatrix}\right)$. Thus the population objective function can be calculated using Monte Carlo integration or with quadrature methods. The population parameters with $x_0=1$ are $(\alpha_{\bm{\tau};1},\beta_{\bm{\tau};1}) = (-1.23, 1.167)$ for $\mathbf{u} = (1/\sqrt{2},1/\sqrt{2})$ and $(\alpha_{\bm{\tau};1},\beta_{\bm{\tau};1}) = (-1.53, 1.49)$ for $\mathbf{u} = (0,1)$, which are found by numerically minimizing the Monte Carlo estimated population objective function.\footnote{The relative error difference between the Monte Carlo and quadrature methods was at most $5\cdot 10^{-3}$. The Monte Carlo approach used a simulation sample of $10^6$.} The weight function is set to 
$K_h(\mathbf{X}_i - \mathbf{x}_0) = \frac{1}{\sqrt{2\pi h^2}}exp\left(-\frac{1}{2h^2}(\mathbf{X}_i - \mathbf{x}_0)^2\right)$ where $h = \sqrt{9\hat{\sigma}^2_{\mathbf{X}}n^{-1/5}}$. 

Table \ref{tab:simres4cond} shows the RMSE of the Bayesian estimator for the conditional local constant model. The estimator appears to be converging to the population parameter.

\begin{table}[htb]
	\centering
	\begin{tabular}{rr|cc}
		& & \multicolumn{2}{c}{$\mathbf{u}$}\\ \hline
		$\theta$	& $n$ & $(1/\sqrt{2},1/\sqrt{2})$ & $(0,1)$ \\ 
		\hline
		& $10^2$ & 2.90e-01 &   6.78e-01  \\ 
		$\alpha_{\bm{\tau},1}$ & $10^3$ & 7.10e-02 & 3.04e-01 \\ 
		& $10^4$ & 2.86e-02 & 1.33e-01 \\ 
		\hline
		& $10^2$ & 8.79e-02 & 3.22e-01  \\ 
		$\beta_{\bm{\tau},1}$ & $10^3$ & 3.35e-02 & 1.29e-01 \\ 
		& $10^4$ & 1.53e-02 & 5.04e-02 \\ 
	\end{tabular}
	\caption{Local constant conditional model RMSE of parameter estimates} 
	\label{tab:simres4cond}
\end{table}

\section{Application}
The unconditional and conditional models are applied to educational data collected from  the Project STAR public 
access 
database.  Project STAR was an experiment conducted on 11,600 students in 300 
classrooms from 1985-1989 with interest of determining if reduced classroom 
size 
improved academic performance.\footnote{The data is publicly available at http://fmwww.bc.edu/ec-p/data/stockwatson.}  Students and teachers were randomly selected in 
kindergarten 
to be in  small (13-17 students) or large (22-26 students)            
classrooms.\footnote{This analysis omits large classrooms that had a teaching assistant.}  The students 
then stayed in their assigned classroom size throughout the fourth grade.  The 
outcome of the treatment was scores of mathematics and reading tests given each year.  This dataset has been analyzed many times before, 
see 
\cite{finn90,folger89,krueger99,mosteller95,word90}.\footnote{ 
	\cite{folger89} and \cite{finn90}  were the first two published studies. 
	\cite{word90} was  the official report from the Tennessee State Department 
	of 
	Education. 
	\cite{mosteller95} 
	provided a review of the study and \cite{krueger99} performed a rigorous 
	econometric analysis focusing 
	on 
	validity.} The studies performed analyses on either single-output test score 
measures or on a functional of mathematics and reading scores. 
Single-output analysis ignores important 
information about the relationship the mathematics and reading test scores 
might 
have with each other. Analysis on the average of scores better 
accommodates 
joint effects but obscures the source of effected subpopulations.  Multiple-output 
quantile regression provides information on the joint relationship 
between 
scores for the entire multivariate distribution (or several specified 
quantile subpopulations). 

A student's outcome was measured using a standardized test called the Stanford 
Achievement Test (SAT) for mathematics and reading.\footnote{The 
	test scores 
	have a finite discrete support.  Computationally, this does 
	not effect the 
	Bayesian estimates, however prevents asymptotically unique estimators.  Thus each 
	score is perturbed with a uniform(0,1) random variable.} 
Section \ref{star:fixtau} inspects the $\tau$-quantile contours on the subset 
first grade 
students (sample size of $n=4,247$, after removal of missing 
data). The results for 
other grades were similar.\footnote{This application explains the concepts of Bayesian      
	multiple-output quantile regression and does not provide rigorous causal 
	econometric 
	inferences.  In the later case a thorough discussion of missing data would 
	be 
	necessary.  For the same reason first grade scores were chosen.  The first 
	grade 
	subset was best suited for pedagogy. } The treatment effect of classroom size is determined by inspecting the location $\tau$-quantile contours of the unconditional model for small and large classrooms. The treatment effect for teacher experience is determined by inspecting the $\tau$-quantile contours from the conditional model (conditional on teacher experience) pooling small and large classrooms. Section \ref{star:sensitive} shows a sensitivity analysis of the unconditional model by inspecting the posterior $\tau$-quantile contours with different prior specifications. Appendix \ref{app:star} presents fixed-$\mathbf{u}$ analysis and an additional sensitivity analysis.

Define the vector $\mathbf{u}=(u_1,u_2)$, where $u_1$ is the mathematics score 
dimension and $u_2$ is the reading score dimension. The $\mathbf{u}$ directions 
have an interpretation of relating how much relative importance the researcher 
wants to give 
to mathematics or reading. Define $\mathbf{u}^\perp = (u_1^\perp,u_2^\perp)$, where $
\mathbf{u}^\perp$ is orthogonal to $\mathbf{u}$.  The components $(u_1^
\perp,u_2^\perp)$ have no meaningful interpretation.  Define $mathematics_i$ to be the 
mathematics score of student $i$ and $reading_i$ to be the reading score of student 
$i$. 

\subsection{$\tau$-quantile (regression) contours}\label{star:fixtau}

The unconditional model is 
\begin{align}\label{eq:appunc}
\begin{split}
\mathbf{Y}_{\mathbf{u}i} &= mathematics_i u_1 + reading_i u_2  \\
\mathbf{Y}_{\mathbf{u}i}^\perp &=  mathematics_i u_1^\perp + reading_iu_2^\perp\\
\mathbf{Y}_{\mathbf{u}i} &= \alpha_{\bm{\tau}} + \beta_{\bm{\tau}} 
\mathbf{Y}
_{\mathbf{u}i}^
\perp + 
\epsilon_i\\
\epsilon_i&\overset{iid}{\sim} ALD(0,1,\tau)\\
\theta_{\bm{\tau}}=(\alpha_{\bm{\tau}},\beta_{\bm{\tau}}) & \sim 
N(\mu_{\theta_{\bm{\tau}}},
\Sigma_{\theta_{\bm{\tau}}}).
\end{split}
\end{align}
Unless otherwise noted, $\mu_{\theta_{\bm{\tau}}} = \mathbf{0}_2$ and $
\Sigma_{\theta_{\bm{\tau}}} = 1000\mathbf{I}_{2}$, meaning ex-ante knowledge is 
a weak belief that the joint distribution of mathematics and reading has spherical 
Tukey depth contours. The number of MCMC draws is 3,000 with a burn in of 1,000. The Gibbs algorithm is initialized at the frequentist estimate.

Figure \ref{fig:test3smlg1} shows the $\tau$-quantile contours for $
\tau=0.05$, $0.20$ and $0.40$. A $\tau$-quantile contour is defined as the boundary of \ref{eq:quantregion} in Appendix \ref{app:quantreview}. The data is stratified into smaller 
classrooms (blue) and larger classrooms (black) and separate models are estimated for each. The unconditional regression model was used but the effective results are conditional since separate models are estimated by classroom size. The innermost contour is the $\tau = 0.40$ 
region, the middle contour is the $\tau = 0.20$ region and the outermost 
contour 
is the $\tau = 0.05$ region. Contour regions for larger $\tau$ will always be 
contained in regions of smaller $\tau$ (if no numerical error and priors are not contradictory). All the points that lie 
on the contour have an estimated Bayesian Tukey depth of $\tau$. The contours for 
larger $\tau$ capture the effects for the more extreme students (e.g.\ students 
who perform exceptionally well on mathematics and reading or exceptionally poorly on 
mathematics but well on reading). The contours for smaller $\tau$ capture the effects 
for the more central or more `median' student (e.g.\ students who do not 
stand out 
from their peers). It can be seen 
that all the contours shift up and to the right for the smaller classroom.  
This 
shows the centrality of mathematics and reading scores improves for
smaller classrooms compared to larger classrooms.  This also 
all quantile subpopulations of scores improve for students in smaller 
classrooms.\footnote{To claim all quantile subpopulations of scores improve would require estimating the $\tau$-quantile regions for all $\tau$.} 
\begin{figure}[H]
	\centering
	\includegraphics[width=.8\linewidth]{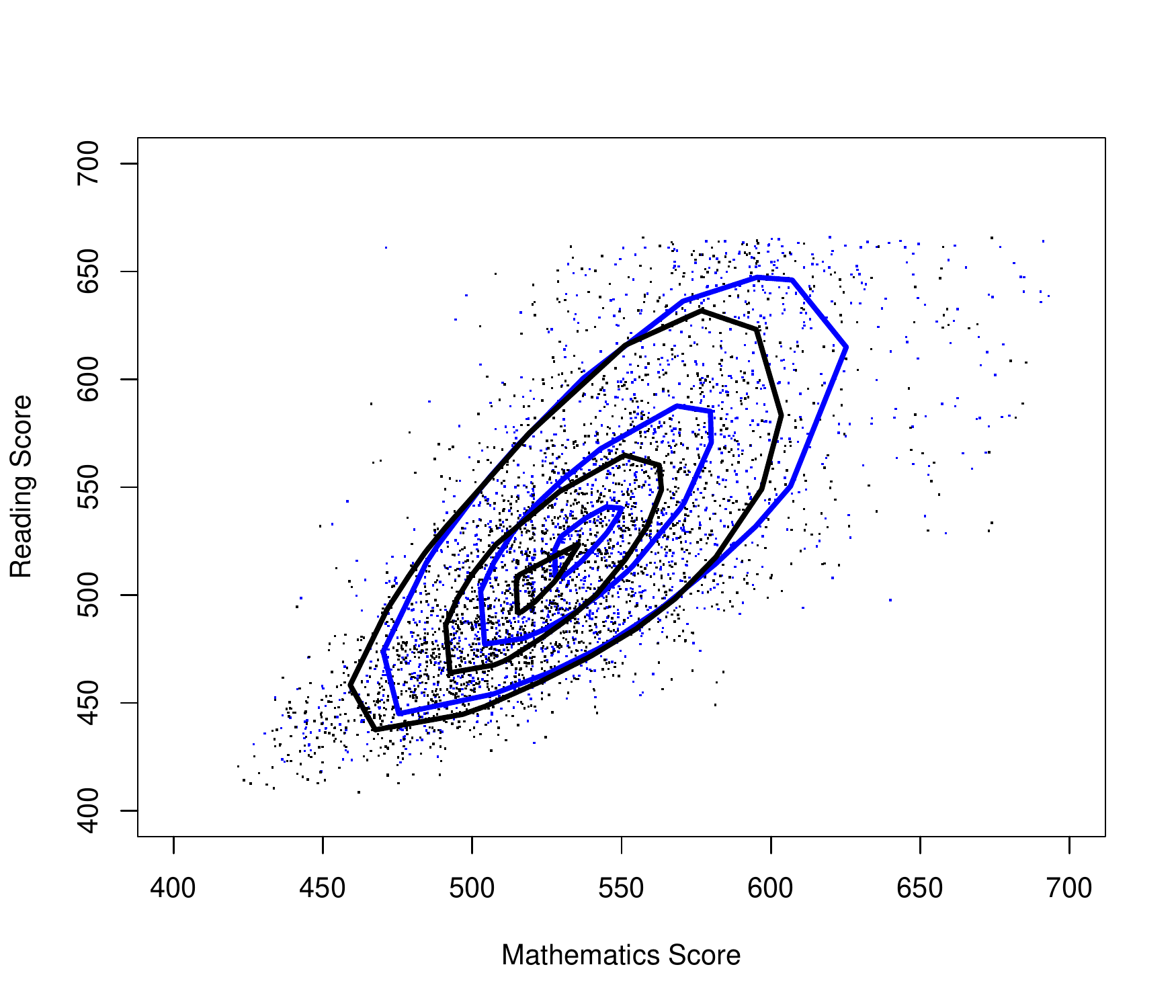}
	\caption[$\tau$-quantile contours]{$\tau$-quantile contours. Blue represents small and black represents large classrooms.}
	\label{fig:test3smlg1}
\end{figure}

Up to this point only quantile location models conditional on binary classroom size have been estimated.\footnote{The unconditional model from (\ref{eq:multdefreg}) is used but is called conditional because separate models were estimated on small and large classrooms.}  When including continuous 
covariates the $\tau$-quantile regions become `tubes' that travel through the 
covariate space. Due to random assignment of teachers, teacher experience can be treated as exogenous and the impact of 
experience on student outcomes can be estimated.  Treating teacher experience as continuous the appropriate model to use is the conditional regression model (\ref{eq:condest}). The local constant specification is preferred if the researcher wishes to inspect slices of the regression tube. The local bilinear specification is preferred if the researcher wishes to connect the slices of the regression tube to create the regression tube. This analysis only looks at slices of the tube, thus the local constant specification is used.

The conditional model is 
\begin{align}\label{eq:appcond}
\begin{split}
\mathbf{Y}_{\mathbf{u}i} &= mathematics_i u_1 + reading_iu_2\\
\mathbf{Y}_{\mathbf{u}i}^\perp &=  mathematics_i u_1^\perp + reading_iu_2^\perp\\
\mathbf{X}_{i} &= years \, of \, teacher \, experience_i\\
\mathcal{X}_{\mathbf{u}i}^l &= [1,\mathbf{X}_{i}-\mathbf{x}_0,\mathbf{Y}_{\mathbf{u}i}^\perp,(\mathbf{X}_{i}-\mathbf{x}_0)\mathbf{Y}_{\mathbf{u}i}^\perp]'\\
\hat{\sigma}^2_{\mathbf{X}} &= \frac{1}{n-1}\sum_{i=1}^{n}(\mathbf{X}_i - \bar{\mathbf{X}})^2\\
h &= \sqrt{9\hat{\sigma}^2_{\mathbf{X}}n^{-1/5}}\\
K_h(\mathbf{X}_i - \mathbf{x}_0) &= \frac{1}{\sqrt{2\pi h^2}}exp\left(-\frac{1}{2h^2}(\mathbf{X}_i - \mathbf{x}_0)^2\right)\\ 
\mathbf{Y}_{\mathbf{u}i} &= \theta_{\bm{\tau};\mathbf{x}_0}\mathcal{X}_{\mathbf{u}i} + 
\epsilon_i\\
\epsilon_i&\overset{iid}{\sim} ALD(0,K_h(\mathbf{X}_i - \mathbf{x}_0)^{-1},\tau)\\
\theta_{\bm{\tau};\mathbf{x}_0} & \sim 
N(\mu_{\theta_{\bm{\tau};\mathbf{x}_0}},
\Sigma_{\theta_{\bm{\tau};\mathbf{x}_0}}).
\end{split}
\end{align}
The prior hyperparameters are $\mu_{\theta_{\bm{\tau};\mathbf{x}_0}} = \mathbf{0}_4$ and $
\Sigma_{\theta_{\bm{\tau};\mathbf{x}_0}} = 1000\mathbf{I}_{4}$. Small and large classrooms are pooled together. Figure \ref{fig:regtube} shows the $\tau$-quantile regression regions with a 
covariate for experience.  The values $\tau$ takes on are 0.20 (left plot) and 
0.05 (right plot). The 
tubes are sliced at $\mathbf{x}_0\in\{1, 10, 20\}$ years of teaching experience.  The left plot 
shows reading scores increase with teacher experience for the more `central' 
students but there does not seem to be a change in mathematics scores.  The 
right plot shows a similar story for most of the `extreme' students.  

\begin{figure}[H]
	\centering
	\includegraphics[width=.85\linewidth]{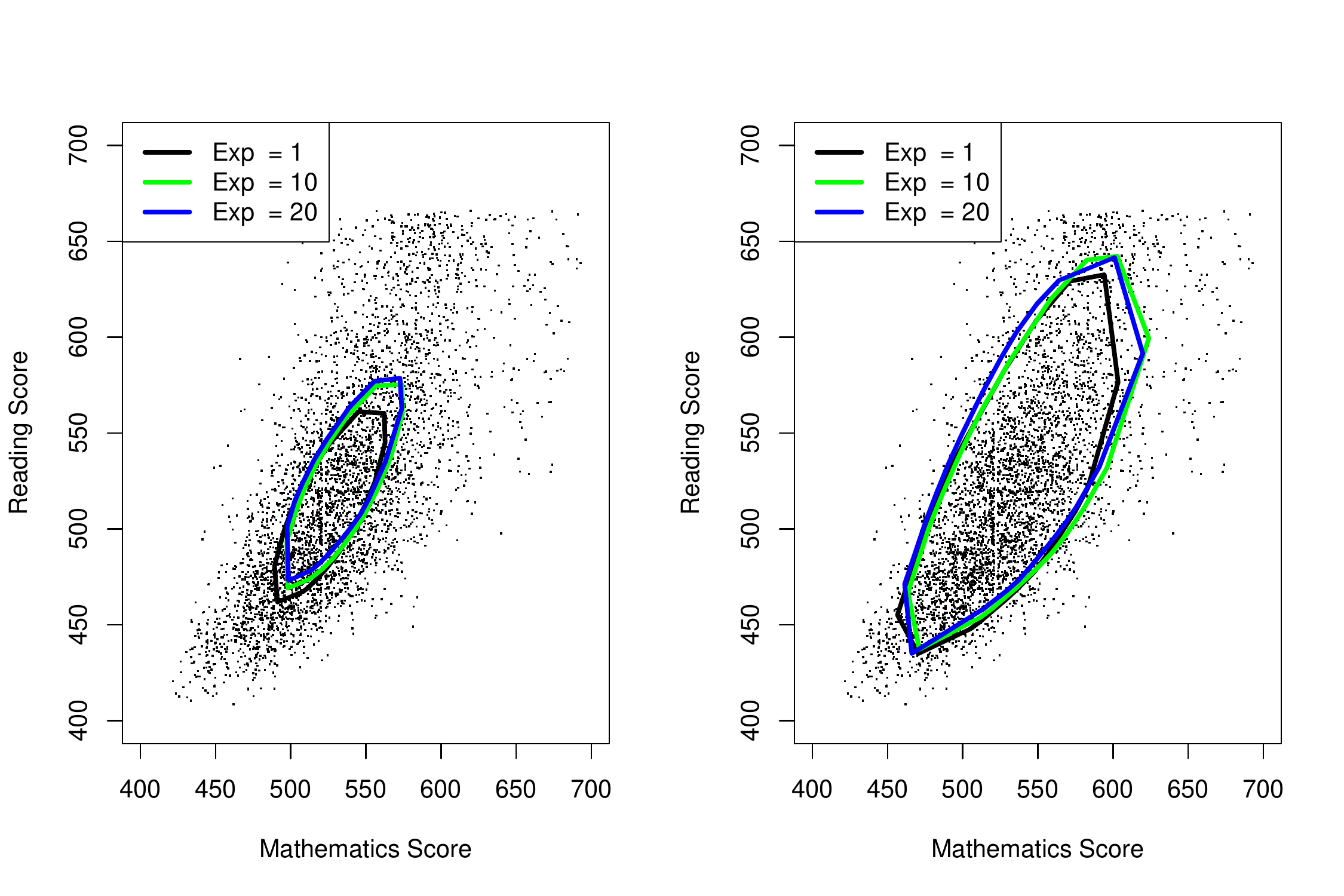}
	\caption[Regression tubes (linear)]{Regression tube slices. Left, $
		\tau=0.2$ quantile regression tube. Right, $\tau=0.05$ quantile regression tube.}
	\label{fig:regtube}
\end{figure}

A non-linear effect of experience is observed. It is clear there is a 
larger marginal impact on student outcomes going from 1 to 10 years of 
experience than from 10 to 20 years of experience. This marginal effect is more 
pronounced for the more central students ($\tau=0.2$).  Previous research has shown strong evidence that the effect of teacher 
experience on student achievement is non-linear. Specifically, the 
marginal effect of experience tends to be 
much larger for teachers that are at the beginning of their career than       
mid-career or late-career teachers \citep{rice10}. The more outlying students ($\tau=0.05$) have a heterogeneous treatment effect with respect to experience. The best performing students in mathematics and reading show increased performance when experience increases from 1 to 10 years but little change after that. All other outlying students are largely unaffected by experience.

\subsection{Sensitivity analysis}\label{star:sensitive}

Figure (\ref{fig:Test3SmBayesFreq}) shows posterior sensitivity of expected $\tau$-quantile contours for weak and strong priors of spherical Tukey depth contours for $\tau\in\{0.05,0.20,0.40\}$.\footnote{To better show the effect of the prior the dataset is reduced to the small classroom subset.}  The posterior from the weak prior is represented by the dashed red line and the posterior from the strong prior is represented by the dotted blue line. The posteriors are compared against the frequentist estimate represented by a solid black contour. The weak priors have covariance $\Sigma_{\theta_{\bm{\tau}}} = diag(1000,1000)$ for all $\bm{\tau}\in\mathcal{B}^k$. For all plots the posterior expected $\tau$-contour with a weak prior is indistinguishable from the frequentist $\tau$-contour. Appendix \ref{app:star} presents a sensitivity analysis for a single $\lambda_{\bm{\tau}}$. 

The top-left plot shows expected posterior $\tau$-contours from a prior mean $\mu_{\theta_{\bm{\tau}}} = \mathbf{0}_2$ for all  $\bm{\tau}\in\mathcal{B}^k$. The strong prior has covariance $\Sigma_{\theta_{\bm{\tau}}} = diag(1,1000)$ for all $\bm{\tau}\in\mathcal{B}^k$. The strong prior represents a strong a priori belief that all $\tau$-contours are near the Tukey median. The prior influence is strongest ex-post for $\tau = 0.05$. This is because the distance between the $\tau$-contour and the Tukey median increases with decreasing $\tau$. The strong prior in the top-right plot has covariance $\Sigma_{\theta_{\bm{\tau}}} = diag(1000,0.0001)$ for all $\bm{\tau}\in\mathcal{B}^k$, all else is the same as the priors from the top-left plot. The top-right plot shows that a strong prior information on $\beta_{{\bm{\tau}}\mathbf{y}}$ provides little information ex-post in this setup.

The bottom-left plot shows expected posterior $\tau$-contours from prior means $\mu_{\theta_{0.05\mathbf{u}}} = (-65,0)$, $\mu_{\theta_{0.20\mathbf{u}}} = (-33,0)$, and $\mu_{\theta_{0.40\mathbf{u}}} = (-10,0)$ for all $\mathbf{u}\in\mathcal{S}^{k-1}$. The strong prior has covariance  $\Sigma_{\theta_{\bm{\tau}}} = diag(1,0.0001)$. This is a strong a priori belief that the $\tau$-contours are spherical and the distance between the $\tau$-contour and the Tukey median decreases with increasing $\tau$. This plot shows the reduction of ellipticity of the posterior $\tau$-contours. The effect is strongest for $\tau=0.05$.

\begin{figure}[htb]
	\centering
	\includegraphics[width=0.7\linewidth]{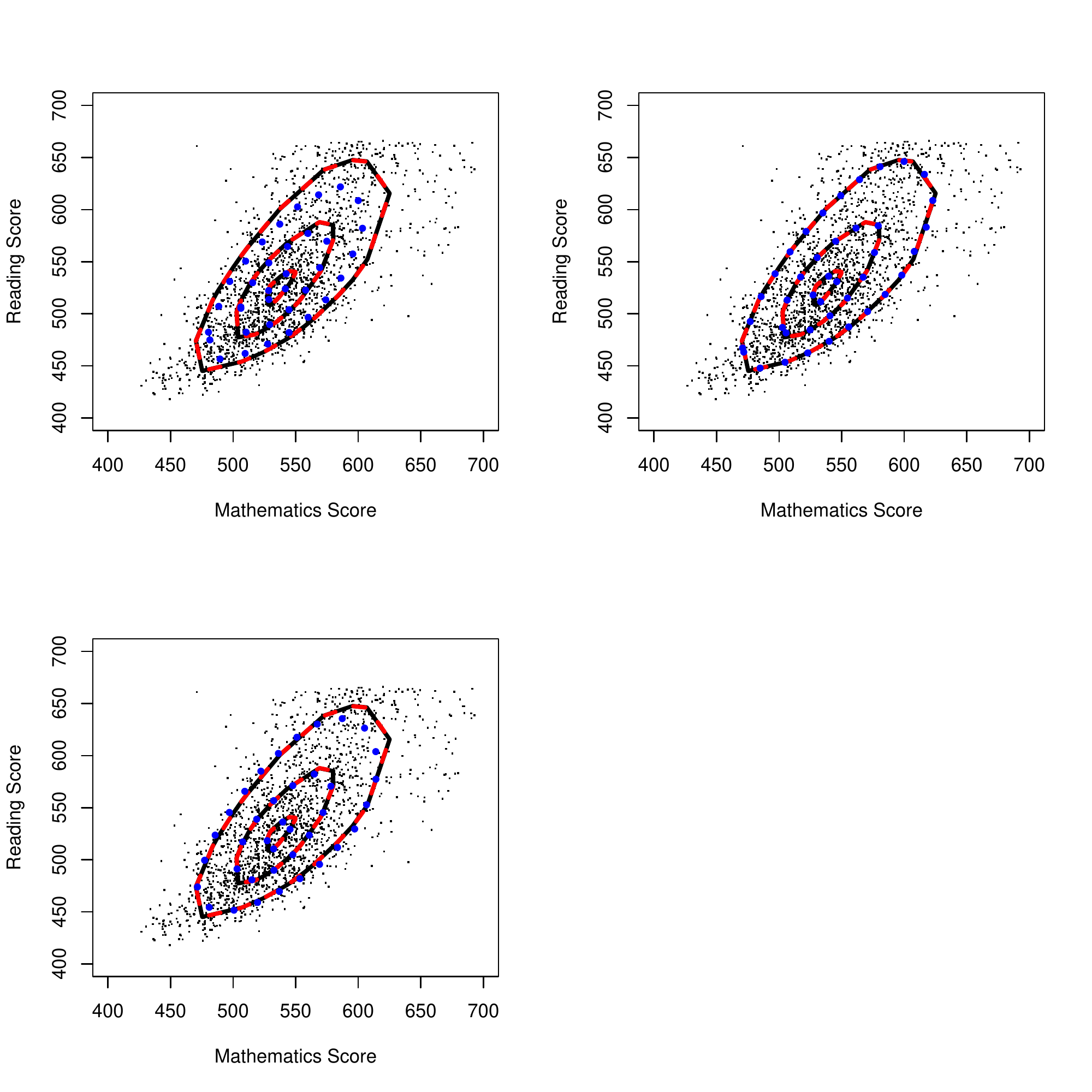}
	\caption[Prior influence ex-post]{Tukey depth contour prior influence ex-post}
	\label{fig:Test3SmBayesFreq}
\end{figure}

\section{Conclusion}
A Bayesian framework for estimation of 
multiple-output directional quantiles was presented. The 
resulting posterior is consistent for the parameters of 
interest, despite having a misspecified likelihood.  By 
performing inferences as a Bayesian one inherits many of the strengths of a 
Bayesian approach.    The model is applied to the Tennessee Project STAR 
experiment and it 
concludes that students in a smaller classroom perform better for every quantile 
subpopulation than students in a larger classroom.

A possible avenue for future work is to find a structural economic model whose 
parameters relate directly to the subgradient conditions.  This would give a 
contextual economic interpretation of the subgradient conditions.  Another 
possibility would 
be developing a formalized hypothesis test for the distribution comparison 
presented in Figure 
\ref{fig:test3smlg1}. This would be a test for the ranking of multivariate 
distributions based off the directional quantile.

\appendix
\section*{Appendix}
\setcounter{equation}{14}

\section{Review of single-output quantiles, Bayesian quantiles and  multiple-output quantiles}\label{app:quantreview}  
\subsection{Quantiles and quantile regression}

Quantiles sort and rank observations 
to describe how extreme an observation is.  In one dimension, for $\tau\in(0,1)
$, the $\tau$th 
quantile is the observation that splits the 
data into two bins: a left bin that contains $\tau\cdot100\%$ of the total 
observations that are
smaller and a right bin that contains the rest of the $(1-\tau)\cdot100\%$ 
total 
observations that are larger. The entire 
family of $\tau\in (0,1)$ 
quantiles allows one to uniquely define the distribution of 
interest. Let $Y\in\Re$ be a univariate random variable with Cumulative Density Function 
(CDF)
$F_Y(y) = Pr(Y\leq y)$ then the $\tau$th 
population single-output quantile is defined as 
\begin{equation}\label{eq:quant1dim}
Q_Y(\tau)=\inf\{y\in \Re: \tau \leq F_Y(y)\}.
\end{equation}
If $Y$ is a continuous random variable then the CDF is invertible and the 
quantile is 
$Q_Y(\tau)=F^{-1}_Y(\tau)$.  Whether or not $Y$ is continuous, $Q_Y(\tau)$ can 
be defined as the generalized inverse of $F_Y(y)$    (i.e.\ $F_Y(Q_Y(\tau))=\tau
$).\footnote{There are several ways to define the generalized inverse 
	of 
	a CDF \citep{embrechts13,feng12}.} 
The 
definition of sample quantile is the same as (\ref{eq:quant1dim}) with 
$F_Y(y)$  replaced with 
its empirical counterpart $\hat{F}_Y(y)=\frac{1}{n}\sum_{i=1}^{n}1_{(y_i\leq y)}$ where $1_{(A)}$ is an indicator function for event $A$ being true.  

\cite{fox64} showed quantiles can be computed via an optimization based 
approach. Define the 
check function to be
\begin{equation}\label{eq:chkfnc}
\rho_\tau(x) = x(\tau-1_{(x<0)}).
\end{equation}
The $\tau$th 
population quantile of $Y\in\Re$ is equivalent to
$Q_Y(\tau)=\underset{a}{argmin}\,E[\rho_\tau(Y-a)]$.  Note this definition 
requires 
$E[Y]$ and $E[Y1_{(Y-a<0)}]$ to be finite. If the moments of $Y$ are not finite, an alternative but equivalent definition 
can be used instead 
\citep{paindaveinesiman11}. The corresponding sample quantile 
estimator is 
\begin{equation}\label{eq:locsamp}
\hat{\alpha}_\tau=\underset{a}{argmin}\frac{1}{n}\sum_{i=1}^{n}\rho_
\tau(y_i-a).
\end{equation}

The commonly accepted definition of single-output linear conditional quantile regression (generally known as 
`quantile regression') was originally proposed by \cite{koenker78}.  The $\tau$th conditional population quantile regression function is 
\begin{equation}\label{eq:linqreg1dim}
Q_{Y|X}(\tau) = \inf\{y\in \Re: \tau \leq F_{Y|X}(y)\}=X'\beta_\tau
\end{equation}
which can be equivalently defined as $Q_{Y|X}(\tau)=\underset{b}
{argmin}\,E[\rho_\tau(Y-X'b)|X]$ (provided the moments 
$E[Y|X]$ 
and $E[Y1_{(Y-X'b<0)}|X]$ are finite).  The parameter $\beta_\tau$ is
estimated by solving  
\begin{equation}\label{eq:frquni}
\hat{\beta}_\tau = 
\underset{b}{argmin}\,\frac{1}{n}\sum_{i=1}^{n}\rho_\tau(y_i 
- 
x_i'b).
\end{equation}
This 
optimization 
problem can be written as a linear programming problem and solutions can 
be found using the simplex or interior 
point algorithms.  

There are two common motivations for quantile regression. First is quantile regression estimates and predictions can be robust to outliers and violations of model assumptions.\footnote{For example, the median of 
	a 
	distribution can be consistently estimated whether or not the distribution 
	has a 
	finite first moment.} Second, quantiles can be of greater 
scientific interest than means or conditional means (as one would find 
in linear regression).\footnote{For example, if one were interested in the 
	effect of police expenditure on crime,  one would expect there to be 
	larger effect for high crime areas (large $\tau$) and little to no effect 
	on low 
	crime areas (small $\tau$).} These two motivations extend to 
multiple-output quantile regression.  See \cite{koenker05} for a  survey of 
the 
field of single-output quantile 
regression.

Several 
approaches to generalizing quantiles from a single-output to a multiple-output random variable have been proposed \citep{small90,serfling02}.  
Generalization is difficult because the inverse of the 
multiple-output CDF is a one-to-many mapping, hence a 
definition based off inverse CDFs can 
lead to difficulties.   See \cite{serflingzuo10} for a discussion of desirable 
criteria one might expect a multiple-output quantiles to have and 
\cite{serfling02} 
for a survey of extending quantiles to the multiple-output case.   \cite{small90} 
surveys the special case of a median.

The proposed method uses a definition of multiple-output quantiles using `directional 
quantiles' introduced by 
\cite{laine01} and rigorously developed by \cite{hallin10}.  A 
directional quantile of $\mathbf{Y}\in\Re^k$ is a function of two objects:  a 
direction vector $\mathbf{u}$ (a point on the surface of $k$ dimension unit 
hypersphere) and 
a depth $\tau\in(0,1)$.  A directional quantile is then uniquely defined by $
\bm{\tau} = \mathbf{u}\tau$. The $\bm{\tau}$ directional 
quantile hyperplane is denoted $\lambda_{\bm{\tau}}$ which is a hyperplane 
through $
\Re^k$.  The hyperplane $\lambda_{\bm{\tau}}$ 
generates two quantile regions: a lower region of all points below $
\lambda_{\bm{\tau}}$ and an upper region of all points above $
\lambda_{\bm{\tau}}$.  The lower region contains $\tau
\cdot 100\%$ of observations and the upper region contains the remaining $(1-
\tau)
\cdot100\%$. Additionally, the vector connecting the probability mass centers 
of the two regions 
is parallel to $\mathbf{u}$.  Thus $\mathbf{u}$ orients the regression and can 
be thought of as a vertical axis.  

\subsection{Bayesian single-output conditional quantile regression}
Bayesian 
methods require a likelihood 
and hence an observational distributional assumption. Yet quantile 
regression avoids making strong distributional assumptions (a seeming contradiction). \cite{yu01} introduced a 
Bayesian 
approach by using a possibly misspecified Asymmetric Laplace 
Distribution (ALD) likelihood.\footnote{Note, the ALD 
	maximum likelihood estimator is equal to the estimator from (\ref{eq:frquni}).}   
The Probability Density Function (PDF) of the $ALD(\mu,\sigma,\tau)$ is 
\begin{equation}\label{eq:ald}
f_\tau(y|\mu,\sigma) = 
\frac{\tau(1-\tau)}{\sigma}exp(-\frac{1}{\sigma}\rho_\tau(y-\mu)).
\end{equation}
The Bayesian assumes $Y|X\sim ALD(X'\beta_\tau,\sigma,\tau)$, selects a prior 
and 
performs estimation using standard procedures. The nuisance scale parameter can be fixed (typically at $\sigma=1$) or freely estimated.\footnote{\cite{rahman16} and \cite{RahmanKarnawat19} are two examples where the scale parameter is used in an ordinal model.} \cite{sriram13} showed posterior 
consistency for this model, meaning that as sample size increases the probability mass of the 
posterior concentrates around the values of $\beta$ that satisfy 
(\ref{eq:linqreg1dim}). The result holds wether $\sigma$ is fixed at 1 or freely estimated. \cite{yangwanghe15} and \cite{sriram15} provide a procedure for constructing confidence intervals with correct frequentist coverage probability.  If 
one is willing to accept prior joint normality of $\beta_\tau$ then a Gibbs sampler 
can be used to obtain random draws from the posterior \citep{kozumi11}.  \cite{alhamzawi12} proposed using an adaptive Lasso sampler to provide regularization. Nonparametric Bayesian approaches have been proposed by \cite{kottas09} and \cite{taddy12}.

\subsection{Unconditional multiple-output quantile regression}

Any given $\lambda_{\bm{\tau}}$ quantile hyperplane separates $\mathbf{Y}$ 
into two 
halfspaces. 
An open lower quantile halfspace,
\begin{equation}\label{eq:lwhalfsp}
H^-_{{\bm{\tau}}} = H_{{\bm{\tau}}}^-(\alpha_{\bm{\tau}},
\beta_{\bm{\tau}}) = \{y\in 
\Re^k: 
\mathbf{u'y}<\mathbf{\beta_{{\bm{\tau}} y}'\Gamma_u'y + \beta_{{\bm{\tau}} 
		x}'X} 
+\alpha_{\bm{\tau}} \},
\end{equation}
and a closed upper quantile halfspace,
\begin{equation}\label{eq:uphalfsp}
H^+_{{\bm{\tau}}} = H_{\bm{\tau}}^+(\alpha_{\bm{\tau}},
\beta_{\bm{\tau}}) = \{y\in 
\Re^k: 
\mathbf{u'y}\geq\mathbf{\beta_{\bm{\tau} y}'\Gamma_u'y + \beta_{{\bm{\tau}} 
		x}'X}
+\alpha_{\bm{\tau}} \}.
\end{equation}

Under certain conditions the distribution $\mathbf{Y}$ can be fully 
characterized 
by a family of 
hyperplanes $\Lambda = \{\lambda_{\bm{\tau}} : 
\text{\boldmath$\tau$}=\tau\mathbf{u}\in \mathcal{B}^k\}$          
\citep[Theorem 5]{kong12}.\footnote{The conditions required are the directional 
	quantile envelopes of the probability distribution of $\mathbf{Y}$ with 
	contiguous support have smooth boundaries for every $\tau \in (0,0.5)$} 
There 
are two 
subfamilies of hyperplanes: a fixed-$\mathbf{u}$ subfamily, $\Lambda_\mathbf{u} = 
\{\lambda_{\bm{\tau}} : 
\text{\boldmath$\tau$}=\tau\mathbf{u}, \tau\in (0,1)\}$, and a fixed-$\tau$ 
subfamily, $\Lambda_\tau = 
\{\lambda_{\bm{\tau}} : 
\text{\boldmath$\tau$}=\tau\mathbf{u}, \mathbf{u}\in \mathcal{S}^{k-1}\}$.  The 
$\mathbf{\tau}$ subfamily is called a $\mathbf{\tau}$ quantile regression region (if no $\mathbf{X}$ is included then it is called a $\mathbf{\tau}$ quantile region).  The 
$\tau$-quantile (regression) region is defined as 
\begin{equation}\label{eq:quantregion}
R(\tau) = \bigcap\limits_{\mathbf{u}\in\mathcal{S}^{k-1}}\cap 
\{H^+_{{\bm{\tau}}}\},
\end{equation}
where $\cap 
\{H^+_{{\bm{\tau}}}\}$ is the intersection over $H^+_{{\bm{\tau}}}$ if 
(\ref{eq:multdefreg}) is not unique.  

The boundary of $R(\tau)$ is called the 
$\tau$-quantile (regression) contour.  The boundary has a strong connection to Tukey (i.e. halfspace) depth contours. Tukey depth is a multivariate notion of centrality for some 
point $\mathbf{y}\in\Re^k$.  
Consider the set of all hyperplanes in $\Re^k$ that pass through $\mathbf{y}$.  The Tukey depth of $\mathbf{y}$ is the minimum  of the
percentage of observations separated by each hyperplane passing through $
\mathbf{y}$. \cite{hallin10} showed the $\mathbf{\tau}$ quantile region is 
equivalent 
to the Tukey depth 
region.\footnote{Mathematically, the 
	Tukey (or halfspace) depth of $\mathbf{y}$ with respect to probability 
	distribution $P$ is defined as $HD(\mathbf{y},P) = \inf \{P[H]: H \text{ is 
		a 
		closed halfspace containing }\mathbf{y}\}$.  Then the Tukey halfspace 
	depth 
	region is defined as $D(\tau) = \{\mathbf{y}\in\Re^k: HD(\mathbf{y},P) \geq 
	\tau \}$.  \cite{hallin10} show $R(\tau) = D(\tau)$ for all 
	$\tau\in[0,1)$.}    
This provides a numerically efficient approach to find Tukey depth contours. 

If $\mathbf{Y}$ (and $\mathbf{X}$ for the 
regression 
case) is 
absolutely continuous with respect to Lebesgue measure, has connected support 
and finite first moments then 
$(\alpha_\tau,\mathbf{\beta_\tau})$ and $\lambda_\tau$ are unique
\citep{paindaveinesiman11}.\footnote{This is Assumption \ref{ass:cont}, stated formally in Section \ref{sec:unc}.}  Under this assumption 
the 
`subgradient conditions' required for consistency are well defined.  It follows that $
\Psi(a,\mathbf{b})$ continuously 
differentiable with respect to $a$ and $\mathbf{b}$ and convex.  The population 
parameters 
$(\alpha_{{\bm{\tau}}},\beta_{{\bm{\tau}}})$ are defined as the parameters 
that satisfy 
two subgradient conditions:  
\begin{equation}
\left.\frac{\partial \Psi(a,\mathbf{b})}{\partial 
	a}\right|_{\alpha_{{\bm{\tau}} 
	},
	\beta_{{\bm{\tau}} }} = 
Pr(\mathbf{Y_u}-\beta_{{\bm{\tau}}
	\mathbf{y}}'\mathbf{Y_u^\perp}-\beta_{{\bm{\tau}} 
	\mathbf{x}}'\mathbf{X}-\alpha_{{\bm{\tau}}}\leq 0)-\tau=0
\end{equation}
and 
\begin{equation}
\left.\frac{\partial \Psi(a,\mathbf{b})}{\partial \mathbf{b}}\right|
_{\alpha_{{\bm{\tau}}
	},\beta_{{\bm{\tau}} }} = 
E[[\mathbf{Y_u^\perp}',\mathbf{X}']'1_{(\mathbf{Y_u}-\beta_{{\bm{\tau}}
		\mathbf{y}}' 
	\mathbf{Y_u^\perp}
	-\beta_{{\bm{\tau}} \mathbf{x}}'\mathbf{X}-\alpha_{{\bm{\tau}}}\leq 
	)}] - 
\tau 
E[[\mathbf{Y_u^\perp}',\mathbf{X}']'] =\mathbf{0}_{k+p-1}.
\end{equation}
The first condition can be written as $Pr(\mathbf{Y}\in 
H^-_{{\bm{\tau}}})= \tau$ which retains the idea of a quantile 
partitioning the support into two sets, one with probability $\tau$ and one 
with probability $(1-\tau)$.  The 
second condition is equivalent to
\begin{align*}
	\tau &= \frac{E[\mathbf{Y}_{\mathbf{u}i}^\perp1_{(\mathbf{Y}\in 
			H_{{\bm{\tau}}}^-)}]}{E[\mathbf{Y}_{\mathbf{u}i}^\perp]}\text{ for 
		all } i\in\{1,...,k-1\}\\
	\tau &= \frac{E[\mathbf{X}_{i}1_{(\mathbf{Y}\in 
			H_{{\bm{\tau}}}^-)}]}{E[\mathbf{X}_{i}]}\text{ for 
		all } i\in\{1,...,p\}
\end{align*}
Note that using the law of total expectations  $E[\mathbf{Y}_{\mathbf{u}i}^\perp1_{(\mathbf{Y}\in 
	H_{{\bm{\tau}}}^-)}]=E[\mathbf{Y}_{\mathbf{u}i}^\perp|\mathbf{Y}\in 
H_{{\bm{\tau}}}^-]Pr(\mathbf{Y}\in 
H_{{\bm{\tau}}}^-) + 0 Pr(\mathbf{Y}\notin 
H_{{\bm{\tau}}}^-) = E[\mathbf{Y}_{\mathbf{u}i}^\perp|\mathbf{Y}\in 
H_{{\bm{\tau}}}^-]\tau$. Then the second condition can be rewritten as
\begin{align*}
	E[\mathbf{Y}_{\mathbf{u}i}^\perp|\mathbf{Y}\in 
	H_{{\bm{\tau}}}^-] &= E[\mathbf{Y}_{\mathbf{u}i}^\perp]\text{ for 
		all } i\in\{1,...,k-1\}\\
	E[\mathbf{X}_{i}|\mathbf{Y}\in 
	H_{{\bm{\tau}}}^-] &= E[\mathbf{X}_{i}]\text{ for 
		all } i\in\{1,...,p\}.
\end{align*}	

Thus the probability mass center in the lower 
halfspace for the orthogonal response is equal to the
probability 
mass center in the entire orthogonal response space.  Likewise for the covariates, the probability mass center in the 
lower 
halfspace is equal to the probability 
mass center in the entire covariate space.  

The first $k-1$ dimensions of the second subgradient conditions can also be rewritten as $\Gamma_{\mathbf{u}}' E[\mathbf{Y}|\mathbf{Y}\in 
H_{{\bm{\tau}}}^-] = \Gamma_{\mathbf{u}}'E[\mathbf{Y}]$ or equivalently $\Gamma_{\mathbf{u}}' (E[\mathbf{Y}|\mathbf{Y}\in 
H_{{\bm{\tau}}}^-] - E[\mathbf{Y}])=\mathbf{0}_{k-1}$, which is satisfied if $E[\mathbf{Y}|\mathbf{Y}\in 
H_{{\bm{\tau}}}^-] = E[\mathbf{Y}]$. This sufficient condition interpretation states that the probability mass center of the response in the lower halfspace is equal to the probability mass center of the response in the entire space. However, this interpretation cannot be guaranteed. 

Further note, $E[[\mathbf{Y_u^\perp}',\mathbf{X}']'] = 
E[[\mathbf{Y_u^\perp}',\mathbf{X}']'1_{(\mathbf{Y}\in 
	H_{{\bm{\tau}}}^+)}] + 
E[[\mathbf{Y_u^\perp}',\mathbf{X}']'1_{(\mathbf{Y}\in 
	H_{{\bm{\tau}}}^-)}]$. Then the second condition can be written as 
\[diag(\mathbf{\Gamma_u'},\mathbf{I}_p)\left[\frac{1}{1-\tau}E[\mathbf{[Y',X']'}
1_{(\mathbf{Y}\in H_{{\bm{\tau}}}^+)}]-\frac{1}{\tau}E[\mathbf{[Y',X']'} 
1_{(\mathbf{Y}\in H_{{\bm{\tau}}}^-)}]\right] = \mathbf{0}_{k+p-1}.\]
The first $k-1$ components,
\[\mathbf{\Gamma_u'}\left[\frac{1}{1-\tau}E[\mathbf{Y}
1_{(\mathbf{Y}\in H_{{\bm{\tau}}}^+)}]-\frac{1}{\tau}E[\mathbf{Y} 
1_{(\mathbf{Y}\in H_{{\bm{\tau}}}^-)}]\right] = \mathbf{0}_{k-1},\]
show $\frac{1}{1-\tau}E[\mathbf{Y}
1_{(\mathbf{Y}\in H_{{\bm{\tau}}}^+)}]-\frac{1}{\tau}E[\mathbf{Y} 
1_{(\mathbf{Y}\in H_{{\bm{\tau}}}^-)}]$ is orthogonal to 
$\mathbf{\Gamma_u'}$ and thus, is parallel to $\mathbf{u}$.  It follows 
the difference of the weighted probability mass centers of the two spaces ($H_{{\bm{\tau}}}^-$ and $H_{{\bm{\tau}}}^+$) is 
parallel to $\mathbf{u}$.\footnote{\cite{hallin10} provides an additional interpretation in terms of Lagrange multipliers.}

\begin{figure}[htb]
	\centering
	\includegraphics[width=0.7\linewidth]{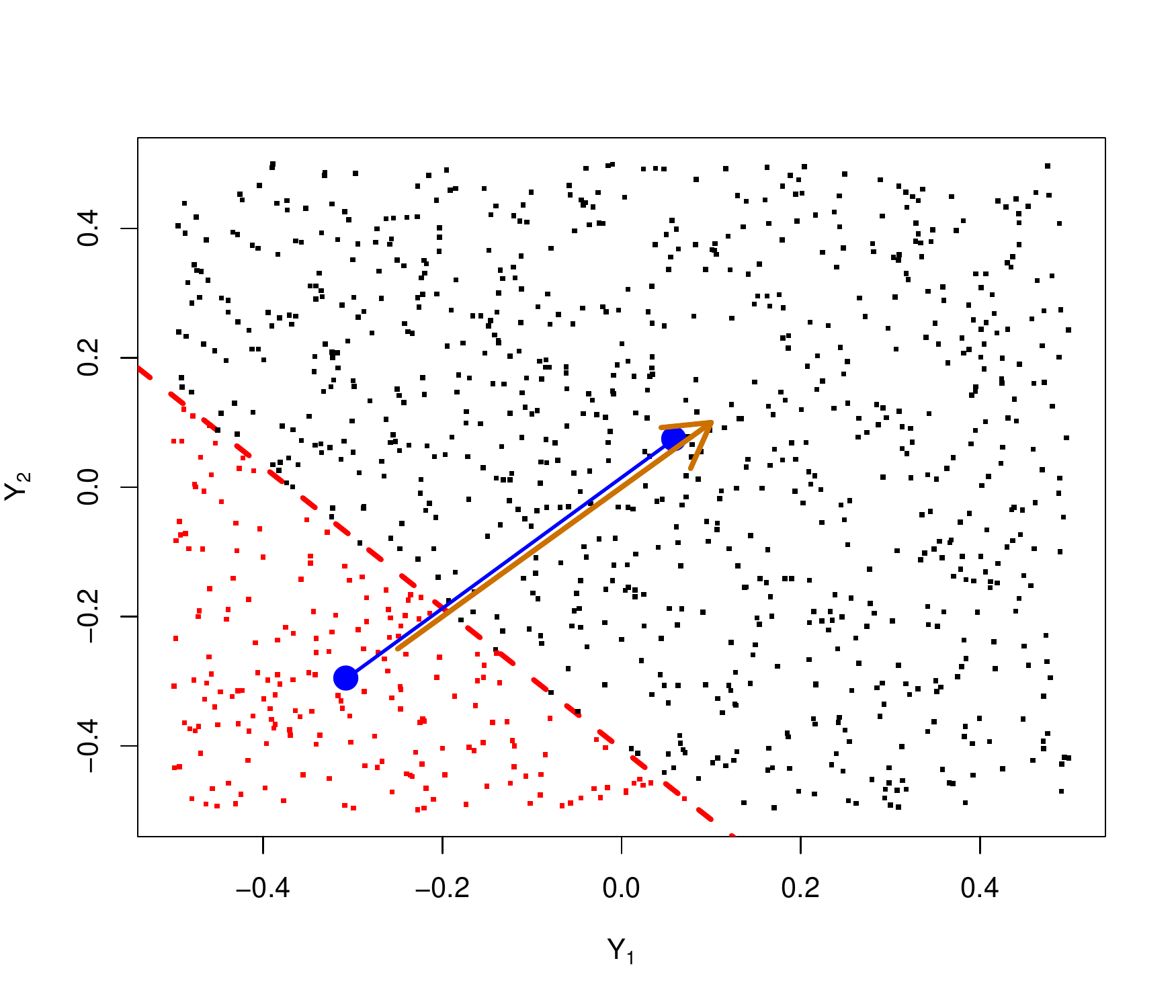}
	\caption[Example of multiple-output quantile (location)]{Lower quantile 
		halfspace for $u=(1/\sqrt{2},1/\sqrt{2})$ and $\tau = 0.2$}
	\label{fig:subgradex}
\end{figure}

Subgradient conditions \ref{eq:subgrad1} and \ref{eq:subgrad2} can be visualized in Figure \ref{fig:subgradex}. The data, $\mathbf{Y}$, are simulated 
with $1,000$ independent draws from the uniform unit square centered 
on $
(0,0)$. The directional vector is $\mathbf{u}=(1/\sqrt{2},1/\sqrt{2})$, represented by
the orange $45^\circ$ degree arrow.  The depth is $
\tau = 0.2$.  The hyperplane $\lambda_{\bm{\tau}}$  is the red dotted line.  The lower quantile region, $H^-
_{{\bm{\tau}}}$, includes the red dots lying below $\lambda_{\bm{\tau}}$.  The upper 
quantile region, $H^+_{{\bm{\tau}}}$, includes the black dots lying above $
\lambda_{\bm{\tau}}$.  The probability mass centers of the lower and upper 
quantile regions are the solid blue dots in their respective 
regions.  The first subgradient condition states that $20\%$ of all points are 
red.  The second subgradient condition states that the line joining the two 
probability mass centers is parallel to $\mathbf{u}$.

\begin{figure}[htb]
	\centering
	\includegraphics[width=1\linewidth]{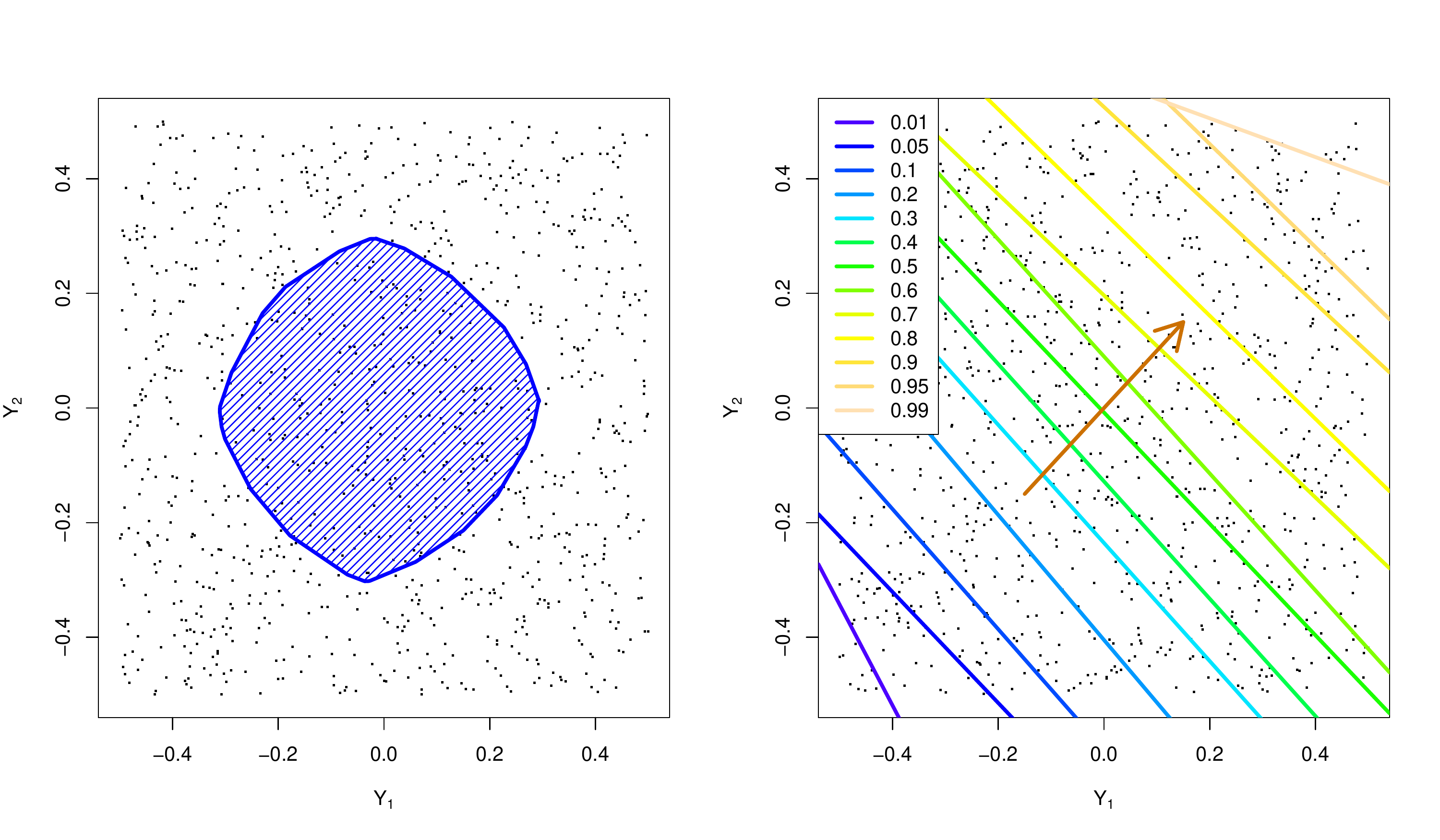}
	\caption[Example of a $\tau$-quantile region and fixed-$\mathbf{u}$ halfspaces]
	{Example of a $\tau$-quantile region and fixed-$\mathbf{u}$ halfspaces. Left, 
		fixed 
		$\tau=0.2$ quantile region. Right, fixed $u=(1/\sqrt{2},1/\sqrt{2})$ 
		quantile 
		halfspaces.}
	\label{fig:fixedutex}
\end{figure}

Figure \ref{fig:fixedutex} shows an example of a $\tau$-quantile region (left) and 
fixed-$
\mathbf{u}$ (right) halfspaces.  The 
left 
plot shows
fixed-$\tau$-quantile upper halfspace intersections of 32 equally spaced 
directions on the unit circle for $\tau=0.2$.  Any points on the boundary have Tukey depth $0.2$. All points within the 
shaded blue region have a Tukey depth greater than or equal to $0.2$ and 
all points outside the shaded blue region have Tukey depth less than $0.2$.

The right plot of Figure \ref{fig:fixedutex} plot shows 13 
quantile hyperplanes $\lambda_{\bm{\tau}}$ for a fixed $\mathbf{u} = (1/
\sqrt{2},1/
\sqrt{2})$ with various $\tau$ (provided in the legend).  The orange arrow shows 
the direction vector $
\mathbf{u}$.  
The hyperplanes split the square such that $\tau\cdot100\%$ of all points lie 
below 
the hyperplanes. The weighted probably mass centers (not shown) are 
parallel to $\mathbf{u}$.  Note the hyperplanes do not need to be orthogonal to $
\mathbf{u}$. 

\begin{figure}[htb]
	\centering
	\includegraphics[width=0.9\linewidth]{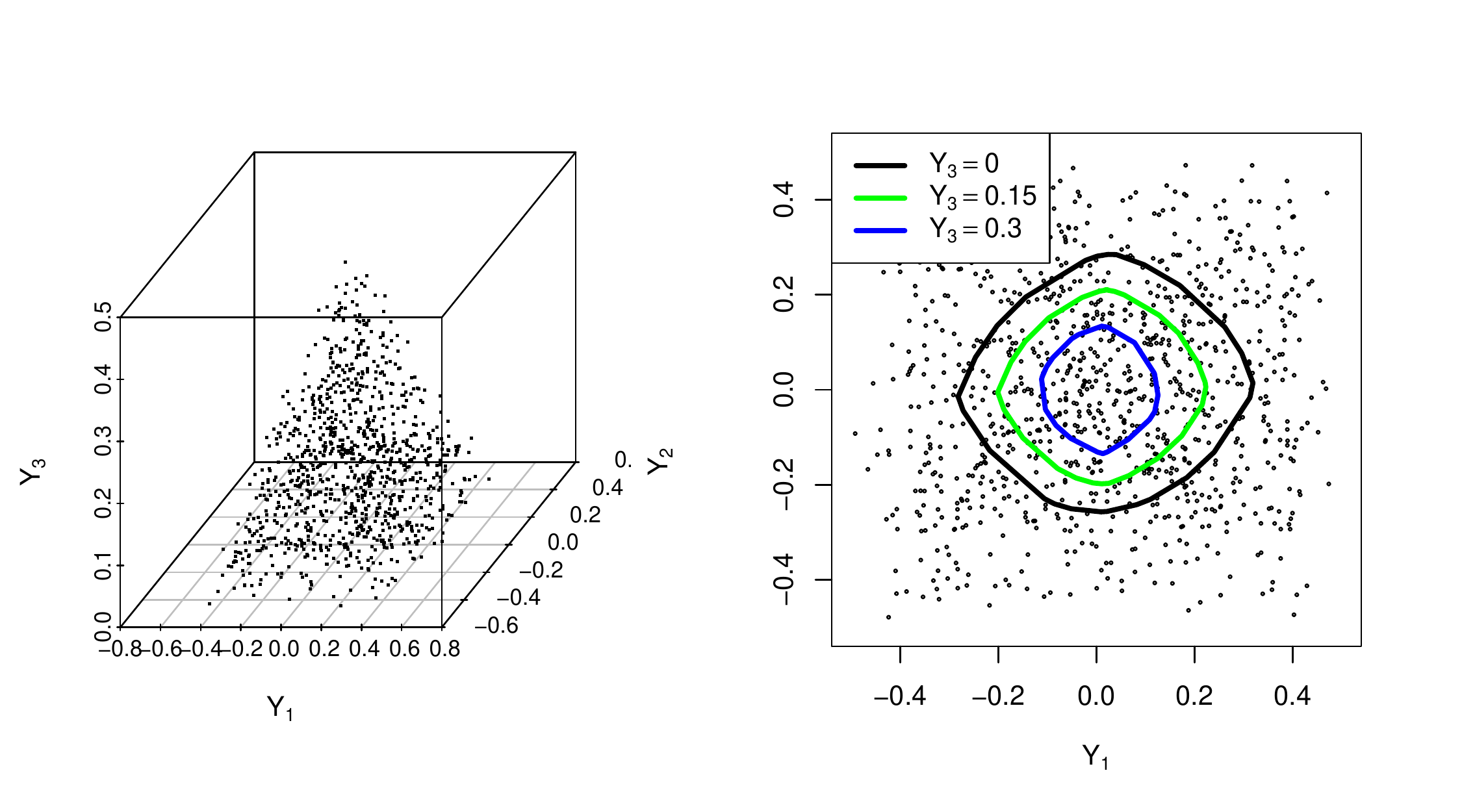}
	\caption[Example of an unconditional $\tau$-quantile regression tube through a uniform 
	pyramid]
	{Example of an unconditional $\tau$-quantile regression tube through a uniform regular square pyramid. Left, 
		draws from  
		random uniform regular square pyramid. Right, three slices of an unconditional
		$\tau=0.2$  quantile
		regression tube.}
	\label{fig:regtubeexample}
\end{figure}

Figure \ref{fig:regtubeexample} shows an example of an unconditional $\tau$-quantile regression 
tube through a random uniform regular square pyramid. The left plot is a 3 dimension scatter 
plot of the uniform regular square pyramid.\footnote{A uniform regular square pyramid is a regular right 
	pyramid with a square base where, for a fixed $\epsilon$, every $\epsilon$-ball contained 
	within 
	the pyramid has the same probability mass. The measure is normalized to 
	one.} 
The 
right plot shows the $\tau$-quantile regression 
tube of $Y_1$ and $Y_2$ regressed on $Y_3$ with cross-section cuts at $Y_3\in
\{0,0.15,0.3\}$. As $Y_3$ increases the tube travels from the base to the tip 
of 
the pyramid and the regression tube pinches. 

As in the 
single-output conditional regression model, the regression tubes are susceptible to quantile 
crossing. 
Meaning if one were to trace out the entire regression tube along $Y_3$ for a 
given $\tau$ and $\tau^\dagger>\tau$, the regression tube for $\tau^\dagger$ 
might not be contained in the one for $\tau$ for all $Y_3$.

\section{Proof of Theorem 1}\label{app:Th1Pf}
In this section consistency of the posterior for the population parameters is proven.  This proof is for the location case.  The regression 
case should come with easy modification by concatenating $\beta_{\bm{\tau y}} $ 
with $\beta_{\bm{\tau x}}$ and $\mathbf{Y}_{\mathbf{u}i}^\perp$ with $\mathbf{X}
_i$. Since $ \mathbf{Y}$ and $\mathbf{X}$ rely on the same sets of assumptions, 
and expectations and probabilities are taken over $\mathbf{Y}$ and $\mathbf{X}$, there 
should 
not 
be any issue with these results generalizing to the regression case.  For ease of readability
${\bm{\tau}}$ is omitted from $\alpha_{\bm{\tau}},\beta_{\bm{\tau}}$ and 
$\Pi_{\bm{\tau}}$. 

Define the population 
parameters 
$(\alpha_{0},\beta_{0})$ to be the parameters that satisfy 
(\ref{eq:subgrad1}) and 
(\ref{eq:subgrad2}). Note that the posterior can be written equivalently as

\begin{equation}\label{eq:poster}
\Pi(U|(\mathbf{Y}_1,\mathbf{X}_1),(\mathbf{Y}_2,\mathbf{X}
_2),...,	 	(\mathbf{Y}_n,\mathbf{X}_n)) = 
\frac{\int_U \prod_{i=1}^{n} 
	\frac{f_{\bm{\tau}}(\mathbf{Y}_i|\mathbf{X}_i,\alpha,
		\beta, 
		\sigma)}{f_{\bm{\tau}}(\mathbf{Y}_i|\mathbf{X}_i,
		\alpha_{0},
		\beta_{0}, 
		\sigma_{0})} d\Pi(\alpha,
	\beta)}{\int_{\Theta}\prod_{i=1}^{n} 
	\frac{f_{\bm{\tau}}(\mathbf{Y}_i|\mathbf{X}_i\alpha,
		\beta, 
		\sigma)}{f_{\bm{\tau}}(\mathbf{Y}_i| \mathbf{X}_i, 
		\alpha_{0},
		\beta_{0}, 
		\sigma_{0})} d\Pi(\alpha,
	\beta)}
\end{equation}

Writing the posterior in this form is for mathematical convenience. Define the posterior numerator to be
\begin{equation}\label{eq: posternumer}
I_{n}(U) = \int_U \prod_{i=1}^{n} 
\frac{f_{\bm{\tau}}(\mathbf{Y}_i|\alpha,\beta, 
	\sigma)}{f_{\bm{\tau}}(\mathbf{Y}_i|\alpha_0,\beta_0, 
	\sigma)} d\Pi(\alpha,\beta).
\end{equation}
The posterior denominator is then $I_n(\Theta)$. The next lemma (presented without proof) provides several 
inequalities that 
are useful later.

\begin{lemma}\label{lem: ineq}
	Let $b_i = (\alpha - \alpha_0) + 
	(\beta-\beta_0)'\mathbf{Y}_{\mathbf{u}i}^\perp$, $W_i = (\mathbf{u'} - 
	\beta_0'\mathbf{\Gamma_u'}) \mathbf{Y}_i - \alpha_0$, $W_i^+ = 
	max(W_i,0)$ and $W_i^- = 
	min(-W_i,0)$.  Then 
	a) $\log\left(\frac{f_{\bm{\tau}}(\mathbf{Y}_i|\alpha,\beta, 
		\sigma)}{f_{\bm{\tau}}(\mathbf{Y}_i|\alpha_0,\beta_0, 
		\sigma)}\right) =$
	
	\[ \frac{1}{\sigma}
	\begin{cases}
	-b_i(1-\tau)     & \quad \text{if } 
	(\mathbf{u'}-\beta'\mathbf{\Gamma_u'})\mathbf{Y}_i 
	- \alpha 
	\leq 0  
	\text{ and } (\mathbf{u'}-\beta_0'\mathbf{\Gamma_u'})\mathbf{Y}_i - 
	\alpha_0 
	\leq 0  \\
	-((\mathbf{u'}-\beta_0'\mathbf{\Gamma_u'}')\mathbf{Y}_i - \alpha_0) +  
	b_i\tau     & \quad 
	\text{if } 
	(\mathbf{u'}-\beta'\mathbf{\Gamma_u'})\mathbf{Y}_i - \alpha > 0  
	\text{ and } (\mathbf{u'}-\beta_0'\mathbf{\Gamma_u'})\mathbf{Y}_i - 
	\alpha_0 
	\leq 0 \\
	(\mathbf{u'}-\beta'\mathbf{\Gamma_u'})\mathbf{Y}_i - \alpha +  b_i\tau     
	& \quad 
	\text{if } 
	(\mathbf{u'}-\beta'\mathbf{\Gamma_u'})\mathbf{Y}_i - \alpha \leq 0  
	\text{ and } (\mathbf{u'}-\beta_0'\mathbf{\Gamma_u'})\mathbf{Y}_i - 
	\alpha_0 > 0 \\
	b_i\tau     & \quad \text{if } 
	(\mathbf{u'}-\beta'\mathbf{\Gamma_u'})\mathbf{Y}_i 
	- \alpha > 0  
	\text{ and } (\mathbf{u'}-\beta_0'\mathbf{\Gamma_u'})\mathbf{Y}_i - 
	\alpha_0 > 0  
	\\
	\end{cases}
	\]
	
	b) $\log\left(\frac{f_{\bm{\tau}}(\mathbf{Y}_i|\alpha,\beta, 
		\sigma)}{f_{\bm{\tau}}(\mathbf{Y}_i|\alpha_0,\beta_0, 
		\sigma)}\right) \leq \frac{1}{\sigma}|b_i|\leq 
	|\alpha-\alpha_0|+|(\beta-\beta_0)'||\mathbf{\Gamma_u'}||\mathbf{Y}_i|$
	
	c) $\log\left(\frac{f_{\bm{\tau}}(\mathbf{Y}_i|\alpha,\beta, 
		\sigma)}{f_{\bm{\tau}}(\mathbf{Y}_i|\alpha_0,\beta_0, 
		\sigma)}\right) \leq 
	\frac{1}{\sigma}|(\mathbf{u'}-\beta_0'\mathbf{\Gamma_u'})\mathbf{Y}_i - 
	\alpha_0|\leq 
	\frac{1}{\sigma}(|(\mathbf{u'}-\beta_0'\mathbf{\Gamma_u'})||\mathbf{Y}_i| + 
	|\alpha_0|)$
	
	
	d) $\log\left(\frac{f_{\bm{\tau}}(\mathbf{Y}_i|\alpha,\beta, 
		\sigma)}{f_{\bm{\tau}}(\mathbf{Y}_i|\alpha_0,\beta_0, 
		\sigma)}\right)  = \frac{1}{\sigma}
	\begin{cases}
	-b_i(1-\tau) + \min(W_i^+,b_i)     & \quad \text{if } b_i > 0  \\
	b_i\tau  + \min(W_i^-,-b_i)   & \quad \text{if } b_i \leq 0 \\
	\end{cases} $
	
	e) $\log\left(\frac{f_{\bm{\tau}}(\mathbf{Y}_i|\alpha,\beta, 
		\sigma)}{f_{\bm{\tau}}(\mathbf{Y}_i|\alpha_0,\beta_0, 
		\sigma)}\right) \geq -\frac{1}{\sigma}|b_i|\geq 
	-|\alpha-\alpha_0|-|(\beta-\beta_0)'||\mathbf{\Gamma_u'}||\mathbf{Y}_i|$
	
\end{lemma}

The next lemma provides more useful inequalities.

\begin{lemma}\label{lem: Eineq}The following inequalities hold:

	a) $E\left[\log\left(\frac{f_{\bm{\tau}}(\mathbf{Y}_i|\alpha,\beta, 
		\sigma)}{f_{\bm{\tau}}(\mathbf{Y}_i|\alpha_0,\beta_0, 
		\sigma)}\right)\right]\leq 0$
	
	b) $\sigma E\left[\log\left(\frac{f_{\bm{\tau}}(\mathbf{Y}_i|\alpha,\beta, 
		\sigma)}{f_{\bm{\tau}}(\mathbf{Y}_i|\alpha_0,\beta_0, 
		\sigma)}\right)\right] = E\left[-(W_i - b_i)1_{(b_i<W_i<0)}\right] 
	+ E\left[(W_i - b_i)1_{(0<W_i<b_i)}\right]$
	
	c) $\sigma E\left[\log\left(\frac{f_{\bm{\tau}}(\mathbf{Y}_i|\alpha,\beta, 
		\sigma)}{f_{\bm{\tau}}(\mathbf{Y}_i|\alpha_0,\beta_0, 
		\sigma)}\right) \right]\leq E\left[-(W_i - 
	b_i)\right]Pr(b_i<W_i<0) + E\left[(W_i - b_i)\right]Pr(0<W_i<b_i) 
	$
	
	d)  $\sigma E\left[\log\left(\frac{f_{\bm{\tau}}(\mathbf{Y}_i|\alpha,\beta, 
		\sigma)}{f_{\bm{\tau}}(\mathbf{Y}_i|\alpha_0,\beta_0, 
		\sigma)}\right)\right] \leq 
	-E\left[-\frac{b_i}{2}1_{(b_i<0)}\right]Pr(\frac{b_i}{2}<W_i<0)  
	-E\left[\frac{b_i}{2}1_{(0<b_i)}\right]Pr(0<W_i<\frac{b_i}{2})$
	
	e)	if Assumption 4 holds then 
	$\lim\limits_{n\rightarrow\infty}\frac{1}{n}
	\sum_{i=1}^{n}E[|W_i|]<\infty$.
\end{lemma}
\begin{proof}
	Note that 
	$E[b_i]=(\alpha-\alpha_0)+(\beta-\beta_0)'E[\mathbf{Y}_{\mathbf{u}i}^\perp]
	= 
	(\alpha-\alpha_0)+\frac{1}{\tau}(\beta-\beta_0)' 
	E[\mathbf{Y}_{\mathbf{u}i}^\perp 1_{
		((\mathbf{u'}-\beta_0'\mathbf{\Gamma}_\mathbf{u}')\mathbf{Y}_i
		- \alpha_0 \leq 0 )}] $ from subgradient condition (12). 
	Define 
	$A_i$ to be the event 
	$(\mathbf{u'}-\beta_0'\mathbf{\Gamma}_\mathbf{u}')\mathbf{Y}_i - \alpha_0 
	\leq 0$ and $A_i^c$ it's complement.  Define $B_i$ to be the event 
	$(\mathbf{u'}-\beta'\mathbf{\Gamma}_\mathbf{u}')\mathbf{Y}_i - \alpha
	\leq 0$ and $B_i^c$ it's complement. 
	\begin{align*}
		&\sigma \log\left(\frac{f_{\bm{\tau}}(\mathbf{Y}_i|\alpha,\beta, 
			\sigma)}{f_{\bm{\tau}}(\mathbf{Y}_i|\alpha_0,\beta_0, 
			\sigma)}\right)  \\
		&= b_i\tau -b_i1_{(A_i,B_i)}- 
		((\mathbf{u'}-\beta_0'\mathbf{\Gamma}_\mathbf{u}') 
		\mathbf{Y}_i-\alpha_0)1_{(A_i,B_i^c)}+ 
		((\mathbf{u'}-\beta'\mathbf{\Gamma}_\mathbf{u}')\mathbf{Y}_i-\alpha) 
		1_{(A_i^c,B_i)}\\
		&=b_i\tau - 
		b_i1_{(A_i)}+(b_i- 
		((\mathbf{u'}-\beta_0'\mathbf{\Gamma}_\mathbf{u}')\mathbf{Y}_i-
		\alpha_0)) 
		1_{(A_i,B_i^c)}+((\mathbf{u'}- 
		\beta'\mathbf{\Gamma}_\mathbf{u}')\mathbf{Y}_i-\alpha) 1_{(A_i^c,B_i)}\\
		&=b_i\tau - b_i1_{(A_i)} 
		-((\mathbf{u'}_i-\beta'\mathbf{\Gamma}_\mathbf{u}')\mathbf{Y}_i-\alpha)
		1_{(A_i,B_i^c)}
		+((\mathbf{u'}-\beta'\mathbf{\Gamma}_\mathbf{u}')\mathbf{Y}_i-\alpha) 
		1_{(A_i^c,B_i)}\\
	\end{align*}
	Since 
	$E[(\alpha-\alpha_0)1_{(A_i)}] = 
	\tau(\alpha-\alpha_0)$ then $E[b_i\tau - 
	b_i1_{(A_i)}]=0$.  Then 
	\[\sigma E\left[\log\left(\frac{f_{\bm{\tau}}(\mathbf{Y}_i|\alpha,\beta, 
		\sigma)}{f_{\bm{\tau}}(\mathbf{Y}_i|\alpha_0,\beta_0, 
		\sigma)}\right) \right]
	=E[-((\mathbf{u}'-\beta'\mathbf{\Gamma}_\mathbf{u}')\mathbf{Y}_i- 
	\alpha)1_{(A_i,B_i^c)}]+E[(\mathbf{u}'- 	\beta'\mathbf{\Gamma}_
	\mathbf{u}')\mathbf{Y}_i-\alpha)1_{(A_i^c,B_i)}]\]
	The constraint in the first term and second terms imply 
	$-((\mathbf{u}'- 
	\beta'\mathbf{\Gamma}_\mathbf{u}')\mathbf{Y}_i-\alpha)<0$ and 
	$(\mathbf{u}'- 
	\beta'\mathbf{\Gamma}_\mathbf{u}')\mathbf{Y}_i-\alpha\leq0$ 
	over their respective support regions.  It follows
	\[\sigma E\left[\log\left(\frac{f_{\bm{\tau}}(\mathbf{Y}_i|\alpha,\beta, 
		\sigma)}{f_{\bm{\tau}}(\mathbf{Y}_i|\alpha_0,\beta_0, 
		\sigma)}\right) \right]= E\left[-(W_i - b_i)1_{(b_i<W_i<0)}\right] + 
	E\left[(W_i 
	- b_i)1_{(0<W_i<b_i)}\right]. \]
	Note that $(W_i-b_i)1_{(0<W_i<b_i)}\leq (W_i-b_i)1_{(0<W_i<\frac{b_i}{2})} 
	< 
	-\frac{b_i}{2}1_{(0<W_i<\frac{b_i}{2})}$.  Likewise,  
	$-(W_i-b_i)1_{(b_i<W_i<0)} < 
	\frac{b_i}{2}1_{(\frac{b_i}{2}<W_i<0)}$. Thus,
	\[\sigma E\left[\log\left(\frac{f_{\bm{\tau}}(\mathbf{Y}_i|\alpha,\beta, 
		\sigma)}{f_{\bm{\tau}}(\mathbf{Y}_i|\alpha_0,\beta_0, 
		\sigma)}\right) \right] \leq 
	E \left[\frac{b_i}{2}1_{(\frac{b_i}{2}<W_i<0)}\right] +                      
	E \left[-\frac{b_i}{2} 1_{(\frac{b_i}{2}>W_i>0)}\right]. \]
	H\"{o}lders inequality with $p=1$ and $q=\infty$ implies $\sigma 
	E\left[\log\left(\frac{f_{\bm{\tau}}(\mathbf{Y}_i|\alpha,\beta, 
		\sigma)}{f_{\bm{\tau}}(\mathbf{Y}_i|\alpha_0,\beta_0, 
		\sigma)}\right) \right] \leq 
	-E\left[ -\frac{b_i}{2}1_{(b_i < 0)} \right] Pr(\frac{b_i}{2}<W_i<0)  
	-E\left[\frac{b_i}{2}1_{(0 < b_i)} \right] Pr(0<W_i<\frac{b_i}{2})$.
\end{proof}

The next proposition shows that the KL minimizer is the parameter vector that 
satisfies the subgradient conditions.

\begin{proposition}\label{prop: klmin}
	Suppose Assumptions 2 and 5 hold.  Then
	\[\underset{(\alpha,\beta)\in\Theta}{\inf}E\left[\log
	\left(\frac{p_0(\mathbf{Y}_i)}{f_{\bm{\tau}}
		(\mathbf{Y}_i|\alpha,\beta,
		1)}\right)\right]\geq 
	E\left[\log\left(\frac{p_0(\mathbf{Y}_i)}{f_{\bm{\tau}}(\mathbf{Y}_i|
		\alpha_0,
		\beta_0,
		1)}\right)\right]\]
	with equality if $(\alpha,\beta)=(\alpha_0,\beta_0)$ where $(\alpha_0,
	\beta_0)$ 
	are defined in (11) and (12).
\end{proposition}\begin{proof}
	This follows from the previous lemma and the fact that 
	\[E\left[\log\left(\frac{p_0(\mathbf{Y}_i)}{f_{\bm{\tau}}
		(\mathbf{Y}_i|\alpha,\beta,
		1)}\right)\right] = 
	E\left[\log\left(\frac{p_0(\mathbf{Y}_i)}{f_{\bm{\tau}}
		(\mathbf{Y}_i|\alpha_0,\beta_0,
		1)}\right)\right] +
	E\left[\log\left(\frac{f_{\bm{\tau}}(\mathbf{Y}_i|\alpha_0,\beta_0,
		1)}{f_{\bm{\tau}}(\mathbf{Y}_i|\alpha,\beta,
		1)}\right)\right]\]
\end{proof}

The next lemma creates an upper bound to approximate $E[I_n(B)^d]$.

\begin{lemma}\label{lem:hypercube}
	Suppose Assumptions 3a or 3b hold and 
	4 holds. Let 
	$B\subset\Theta\subset \Re^k$.  
	For $\delta>0$ and $d\in(0,1)$, let $\{A_j: 
	1 \leq j \leq J(\delta)\}$ be hypercubes of volume 
	$\left(\frac{\delta\frac{1}{k}}{1+c_\Gamma c_y}\right)^k$ 
	required to cover $B$.  Then for $(\alpha^{(j)},\beta^{(j)})\in A_j$, the 
	following 
	inequality holds
	\[E\left[\left(\int_{B}\prod_{i=1}^{n}\frac{f_{\bm{\tau}}(\mathbf{Y}_i|
		\alpha,
		\beta,
		1)}{f_{\bm{\tau}}(\mathbf{Y}_i|\alpha_0,\beta_0, 
		1)}d\Pi(\alpha,\beta)\right)^d\right]\leq 
	\sum_{j=1}^{J(\delta)}\left[E\left[\left(\prod_{i=1}^{n}\frac{f_{\bm{\tau}}
		(\mathbf{Y}_i
		|\alpha_j,\beta_j,
		1)}{f_{\bm{\tau}}(\mathbf{Y}_i|\alpha_0,\beta_0, 
		1)}\right)^d\right]e^{nd\delta}\Pi(A_j)^d\right]\]
\end{lemma}
\begin{proof}
	For all $(\alpha,\beta)\in A_j$, 
	$|\alpha-\alpha^{(j)}|\leq\frac{\delta\frac{1}{k}}{1+c_\Gamma c_y}$ and 
	$|\beta- \beta^{(j)}|\leq\frac{\delta\frac{1}{k}}{1+c_\Gamma 
		c_y}\mathbf{1}_{k-1}$ compenentwise.  Then 
	$|\alpha-\alpha^{(j)}| + |\beta-\beta^{(j)}|'\mathbf{1}_{k-1}c_\Gamma c_y 
	\leq 
	\delta$.  Using 
	lemma \ref{lem: ineq}b 
	\begin{align*}
		\log\left(\frac{f_{\bm{\tau}}(\mathbf{Y}_i|\alpha,\beta, 
			1)}{f_{\bm{\tau}}(\mathbf{Y}_i|\alpha^{(j)},\beta^{(j)}, 
			1)}\right) &\leq 
		|\alpha-\alpha^{(j)}|+|\beta-\beta^{(j)}|'|\Gamma_u'||\mathbf{Y}_i| \\
		&\leq 
		|\alpha-\alpha^{(j)}|+|\beta-\beta^{(j)}|'\mathbf{1}_{k-1}c_\Gamma c_y\\
		&\leq \frac{\delta}{1+c_\Gamma c_y}\\
		&<\delta 
	\end{align*} 
	Then $\int_{A_j}\prod_{i=1}^{n} 
	\frac{f_{\bm{\tau}}(\mathbf{Y}_i|\alpha,\beta, 
		1)}{f_{\bm{\tau}}(\mathbf{Y}_i|\alpha_0,\beta_0, 
		1)} d\Pi(\alpha,\beta)=$
	\begin{align*}
		&
		\prod_{i=1}^{n}\frac{f_{\bm{\tau}}(\mathbf{Y}_i|\alpha^{(j)},
			\beta^{(j)},1)}{f_{\bm{\tau}}(\mathbf{Y}_i|\alpha_0,\beta_0, 
			1)}\int_{A_j}\prod_{i=1}^{n} 
		\frac{f_{\bm{\tau}}(\mathbf{Y}_i|\alpha,
			\beta, 
			1)}{f_{\bm{\tau}}(\mathbf{Y}_i|\alpha^{(j)},\beta^{(j)}, 
			1)}d\Pi(\alpha,\beta)\\
		&\leq \prod_{i=1}^{n}\frac{f_{\bm{\tau}}(\mathbf{Y}_i|\alpha^{(j)},
			\beta^{(j)}, 
			1)}{f_{\bm{\tau}}(\mathbf{Y}_i|\alpha_0,\beta_0, 
			1)} e^{n\delta}\Pi(A_j)
	\end{align*}
	Then $E\left[\left(\int_{B}\prod_{i=1}^{n} 
	\frac{f_{\bm{\tau}}(\mathbf{Y}_i|\alpha,\beta, 
		1)}{f_{\bm{\tau}}(\mathbf{Y}_i|\alpha_0,\beta_0, 
		1)} d\Pi(\alpha,\beta)\right)^d \right]\leq$
	\begin{align*}
		&		E\left[\left(\sum_{j=1}^{J(\delta)}\left(\prod_{i=1}^{n} 
		\frac{f_{\bm{\tau}}(\mathbf{Y}_i|\alpha^{(j)},\beta^{(j)}, 
			1)}{f_{\bm{\tau}}(\mathbf{Y}_i|\alpha_0,\beta_0, 
			1)} d\Pi(\alpha,\beta)\right) e^{n\delta} 
		\Pi(A_j)\right)^d\right]\\ 
		&\leq\sum_{j=1}^{J(\delta)}	E\left[\left(\prod_{i=1}^{n} 
		\frac{f_{\bm{\tau}}(\mathbf{Y}_i|\alpha^{(j)},\beta^{(j)}, 
			1)}{f_{\bm{\tau}}(\mathbf{Y}_i|\alpha_0,\beta_0, 
			1)} d\Pi(\alpha,\beta)\right)^d e^{nd\delta} (\Pi(A_j))^d\right].
	\end{align*}
	The last inequality holds because $(\sum_{i}x_i)^d\leq \sum_{i}x_i^d$ for 
	$d\in(0,1)$ and $x_i>0$.
\end{proof}

Let 
$U^c_n\subset \Theta$ such that $(\alpha_0,\beta_0)\not\in U_n^c$.  The next 
lemma creates an upper bound for the expected value of the likelihood within 
$U_n^c$.    Break 
$U_n^c$ into a 
sequence of halfspaces, $\{V_{ln}\}_{l=1}^{L(k)}$, such that 
$\bigcup\limits_{l=1}^{L(k)} V_{ln} 
= 
U_n^c$, where 
\begin{align*}
	V_{1n} &= \{(\alpha,\beta):\alpha - \alpha_0 
	\geq \Delta_n,\beta_1 - \beta_{01}\geq0,...,\beta_{k} - \beta_{0k}\geq0\}\\
	V_{2n} &= \{(\alpha,\beta):\alpha - \alpha_0 
	\geq 0,\beta_1 - \beta_{01}\geq \Delta_n,...,\beta_{k} - \beta_{0k}\geq0\}\\
	&\vdots\\
	V_{L(k)n} &= \{(\alpha,\beta):\alpha - \alpha_0 
	<  0,\beta_1 - \beta_{01}<0,...,\beta_{k} - \beta_{0k}\leq  -\Delta_n\}
\end{align*}
for some $\Delta_n>0$.  This sequence makes explicit that the distance of at least one 
component of the vector $(\alpha,\beta)$ is larger than it's corresponding 
component of $(\alpha_0,\beta_0)$ by at least $|\Delta_n|$. 
How the sequence is indexed exactly is not important.  The rest of the proof 
will focus on 
$V_{1n}$, the arguments for the other sets are similar.  Define
$B_{in} = -E \left[ \log\left(\frac{f_{\bm{\tau}}(\mathbf{Y}_i|\alpha,\beta, 
	1)}{f_{\bm{\tau}}(\mathbf{Y}_i|\alpha_0,\beta_0, 
	1)}\right)\right]$.\footnote{I would like the thank Karthik Sriram for help with the proof of the next lemma.}

\begin{lemma}\label{lem:existdelta}
	Let $G\in \Theta$ be compact.  Suppose Assumption 4 holds 
	and 
	$(\alpha,\beta)\in G\cap 
	V_{1n}$.  Then there 
	exists a 
	$d\in(0,1)$ such that 
	\[E\left[\prod_{i=1}^{n} 
	\left(\frac{f_{\bm{\tau}}(\mathbf{Y}_i|\alpha,\beta, 
		1)}{f_{\bm{\tau}}(\mathbf{Y}_i|\alpha_0,\beta_0, 
		1)} \right)^d \right] \leq 
	e^{-d\sum_{i=1}^{n}B_{in} } \]
\end{lemma}
\begin{proof}
	Define $h_d(\alpha,\beta) = \frac{1- 
		E\left[\left(\frac{f_{\bm{\tau}}(\mathbf{Y}_i|\alpha,\beta, 
			1)}{f_{\bm{\tau}}(\mathbf{Y}_i|\alpha_0,\beta_0, 
			1)}\right)^d\right]}{d} - 
	E\left[log\left(\frac{f_{\bm{\tau}}(\mathbf{Y}_i|\alpha,\beta, 
		1)}{f_{\bm{\tau}}(\mathbf{Y}_i|\alpha_0,\beta_0, 
		1)}\right)\right]$.  From the proof of Lemma 6.3 in \cite{kleijn06}, 
	$\lim\limits_{d\rightarrow 0}h_d(\alpha,\beta)=0$ and $h_d(\alpha,\beta)$ 
	is a 
	decreasing function of $d$ for all $(\alpha,\beta)$.  Note that 
	$h_d(\alpha,\beta)$ is 
	continuous in $(\alpha,\beta)$.  Then by Dini's theorem $h_d(\alpha,\beta)$ 
	converges to 
	$h_d(0,\mathbf{0}_{k-1})$ uniformly in $(\alpha,\beta)$ as $d$ converges to 
	zero.  Define $\delta = 
	\underset{(\alpha,\beta)\in 
		G}{\inf}\log\left(\frac{f_{\bm{\tau}}(\mathbf{Y}_i|\alpha,\beta, 
		1)}{f_{\bm{\tau}}(\mathbf{Y}_i|\alpha_0,\beta_0, 
		1)}\right)$ then there exists a $d_0$ such that 
	$0-h_{d_0}(\alpha,\beta)\leq 
	\frac{\delta}{2}$.  From lemma \ref{lem: Eineq}a $E \left[ 
	log\left(\frac{f_{\bm{\tau}}(\mathbf{Y}_i|\alpha,\beta, 
		1)}{f_{\bm{\tau}}(\mathbf{Y}_i|\alpha_0,\beta_0, 
		1)}\right)\right]<0$. Then
	\begin{align*}
		E\left[\left(\frac{f_{\bm{\tau}}(\mathbf{Y}_i|\alpha,\beta, 
			1)}{f_{\bm{\tau}}(\mathbf{Y}_i|\alpha_0,\beta_0, 
			1)}\right)^{d_0}\right] &\leq 1 + d_0E \left[ 
		log\left(\frac{f_{\bm{\tau}}(\mathbf{Y}_i|\alpha,\beta, 
			1)}{f_{\bm{\tau}}(\mathbf{Y}_i|\alpha_0,\beta_0, 
			1)}\right)\right]  + d_0\frac{\delta}{2}\\
		&\leq 1 + \frac{d_0}{2}E \left[ 
		log\left(\frac{f_{\bm{\tau}}(\mathbf{Y}_i|\alpha,\beta, 
			1)}{f_{\bm{\tau}}(\mathbf{Y}_i|\alpha_0,\beta_0, 
			1)}\right)\right] \\
		&\leq e^{\frac{d_0}{2}E \left[ 
			log\left(\frac{f_{\bm{\tau}}\mathbf{Y}_i|\alpha,\beta, 
				1)}{f_{\bm{\tau}}(\mathbf{Y}_i|\alpha_0,\beta_0, 
				1)}\right)\right]}
	\end{align*}
	The last inequality holds because $1+t\leq e^t$ for any $t\in\Re$.
\end{proof}
The 
next lemma 
is used to show the numerator of the posterior, $I_n(U^c_n)$, converges to zero 
for 
sets $U^c_n$ not containing $(\alpha_0,\beta_0)$. 

\begin{lemma}\label{lem:existu}
	Suppose Assumptions 3a, 4 and 
	6  
	hold. Then there exists a 
	$u_j>0$ such that for any compact $G_j \subset \Theta$, 
	\[\int_{G_j^c\cap V_{jn}}e^{\sum_{i=1}^{n}	
		\log\left(\frac{f_{\bm{\tau}}(\mathbf{Y}_i|\alpha,\beta, 
			1)}{f_{\bm{\tau}}(\mathbf{Y}_i|\alpha_0,\beta_0, 
			1)}\right)}d\Pi(\alpha,\beta)\leq e^{-nu_j}\]
	for sufficiently large $n$.
\end{lemma}
\begin{proof}
	Let
	\[C_0 = 
	\frac{4\lim\limits_{n\rightarrow\infty}\frac{1}{m} 
		\sum_{i=1}^{m}E[|W_i|]}{(1-\tau)c_p},\]
	$\epsilon=\min(\epsilon_{Z})$ and $A=kB\epsilon=2C_0$,  where $c_p$ and 
	$\epsilon_z$ are from Assumption 6. This limit exists by 
	Lemma 
	\ref{lem: Eineq}e.  Define 
	\[G_1 = \{(\alpha,\beta): 
	(\alpha-\alpha_0,\beta_1-\beta_{01},...,\beta_k-\beta_{0k})\in[0,A] 
	\times[0,B]\times... \times [0,B]\}.\]  If $(\alpha,\beta)\in G_1^c \cap 
	W_1$ 
	then $(\alpha-\alpha_0)>A$ or $(\beta-\beta_0)_j>B$ for some $j$. If 
	$\mathbf{Y}^\perp_{\mathbf{u}i}>\epsilon$ then $b_i = (\alpha - \alpha_0)  
	+ 
	(\beta-\beta_0)'\mathbf{Y}^\perp_{\mathbf{u}i}>2C_0$.   
	Split the likelihood ratio as
	\begin{align*}
		&\sum_{i=1}^{n}	
		\log\left(\frac{f_{\bm{\tau}}(\mathbf{Y}_i|\alpha,\beta, 
			1)}{f_{\bm{\tau}}(\mathbf{Y}_i|\alpha_0,\beta_0, 
			1)}\right) =\\
		&\sum_{i=1}^{n}	
		\log\left(\frac{f_{\bm{\tau}}(\mathbf{Y}_i|\alpha,\beta, 
			1)}{f_{\bm{\tau}}(\mathbf{Y}_i|\alpha_0,\beta_0, 
			1)}\right)1_{(\mathbf{Y}^\perp_{\mathbf{u}ij}>\epsilon_{Zj}, 
			\forall 
			j)} + \sum_{i=1}^{n}	
		\log\left(\frac{f_{\bm{\tau}}(\mathbf{Y}_i|\alpha,\beta, 
			1)}{f_{\bm{\tau}}(\mathbf{Y}_i|\alpha_0,\beta_0, 
			1)}\right)(1-1_{(\mathbf{Y}^\perp_{\mathbf{u}ij}>\epsilon_{Zj}, 
			\forall 
			j)}).
	\end{align*}
	Since $\min(W_i^+,b_i)\leq W_i^+ \leq |W_i|$ and using lemma \ref{lem: ineq}
	d,
	\begin{align*}
		\sum_{i=1}^{n}	
		log\left(\frac{f_{\bm{\tau}}(\mathbf{Y}_i|\alpha,\beta, 
			1)}{f_{\bm{\tau}}(\mathbf{Y}_i|\alpha_0,\beta_0, 
			1)}\right)1(\mathbf{Y}^\perp_{\mathbf{u}ij}>\epsilon_{Zj}, \forall 
		j) &= \sum_{i=1}^{n} (-b_i(1-\tau) + 
		min(W_i^+,b_i))1_{(\mathbf{Y}^\perp_{\mathbf{u}ij}>\epsilon_{Zj}, 
			\forall 
			j)}\\
		&\leq \sum_{i=1}^{n} (-2C_0(1-\tau) + 
		|W_i|)1_{(\mathbf{Y}^\perp_{\mathbf{u}ij}>\epsilon_{Zj}, \forall 
			j)}.
	\end{align*}
	From lemma \ref{lem: ineq}b and for large enough $n$ then
	\begin{align*}
		\sum_{i=1}^{n}	
		log\left(\frac{f_{\bm{\tau}}(\mathbf{Y}_i|\alpha,\beta, 
			1)}{f_{\bm{\tau}}(\mathbf{Y}_i|\alpha_0,\beta_0, 
			1)}\right)1_{(\mathbf{Y}^\perp_{\mathbf{u}ij}>\epsilon_{Zj}, 
			\forall 
			j)} &\leq \sum_{i=1}^{n} 
		|W_i|(1-1_{(\mathbf{Y}^\perp_{\mathbf{u}ij}>\epsilon_{Zj}, \forall 
			j)}).
	\end{align*}
	Then for large enough $n$
	\begin{align*}
		\sum_{i=1}^{n}	
		log\left(\frac{f_{\bm{\tau}}(\mathbf{Y}_i|\alpha,\beta, 
			1)}{f_{\bm{\tau}}(\mathbf{Y}_i|\alpha_0,\beta_0, 
			1)}\right) &\leq 
		-nC_0(1-\tau)Pr(\mathbf{Y}^\perp_{\mathbf{u}ij}>\epsilon_{Zj}, \forall 
		j) + 2n\lim\limits_{n\rightarrow\infty}\frac{1}{m} 
		\sum_{i=1}^{m}E[|W_i|]\\
		&=-2n\lim\limits_{n\rightarrow\infty}\frac{1}{m} 
		\sum_{i=1}^{m}E[|W_i|]\\
		&=-\frac{1}{2}nC_0(1-\tau)Pr(\mathbf{Y}^\perp_{\mathbf{u}ij}>
		\epsilon_{Zj}, 
		\forall 
		j)
	\end{align*}
	Thus the result holds when $u_i = 
	\frac{1}{2}C_0(1-\tau)Pr(\mathbf{Y}^\perp_{\mathbf{u}ij}>\epsilon_{Zj}, 
	\forall 
	j)$.
\end{proof}

The next lemma shows the marginal likelihood, $I_n(\Theta)$, goes to infinity 
at 
the same rate as the numerator in the previous lemma.

\begin{lemma}\label{lem:lrgrthn}
	Suppose Assumptions 3a and 4 holds, then
	\[\int_{\Theta}e^{\sum_{i=1}^{n}	
		\log\left(\frac{f_{\bm{\tau}}(\mathbf{Y}_i|\alpha,\beta, 
			1)}{f_{\bm{\tau}}(\mathbf{Y}_i|\alpha_0,\beta_0, 
			1)}\right)}d\Pi(\alpha,\beta)\geq e^{-n\epsilon}.\]
\end{lemma}
\begin{proof}
	From Lemma \ref{lem: ineq}e 
	$\log\left(\frac{f_{\bm{\tau}}(\mathbf{Y}_i|\alpha,\beta, 
		1)}{f_{\bm{\tau}}(\mathbf{Y}_i|\alpha_0,\beta_0, 
		1)}\right) \geq -|b_i|\geq -|\alpha-\alpha_0|- 
	|\beta-\beta_0|'|\Gamma_u||\mathbf{Y}_i|$.  Define 
	\[D_\epsilon = \left\{(\alpha,\beta):|\alpha-\alpha_0| < 
	\frac{\frac{1}{k}\epsilon}{1+c_\Gamma c_y}, |\beta-\beta_{0}| < 
	\frac{\frac{1}{k}\epsilon}{1+c_\Gamma c_y}\mathbf{1}_{k-1} \text{ 
		componentwise}\right\} .\]
	
	Then for $(\alpha,\beta)\in V_\epsilon$ 
	\begin{align*}
		\log\left(\frac{f_{\bm{\tau}}(\mathbf{Y}_i|\alpha,\beta, 
			1)}{f_{\bm{\tau}}(\mathbf{Y}_i|\alpha_0,\beta_0, 
			1)}\right) &\geq -|\alpha-\alpha_0|- 
		|\beta-\beta_0|'|\Gamma_u||\mathbf{Y}_i|\\
		& \geq  -|\alpha-\alpha_0|- 
		|\beta-\beta_0|'\mathbf{1}_{k-1} c_\Gamma c_y\\
		&\geq -\frac{\epsilon}{1+c_\Gamma c_y} \\
		&> -\epsilon
	\end{align*}
	Then 
	$\sum_{i=1}^{n}\log\left(\frac{f_{\bm{\tau}}(\mathbf{Y}_i|\alpha,\beta, 
		1)}{f_{\bm{\tau}}(\mathbf{Y}_i|\alpha_0,\beta_0, 
		1)}\right) \geq -n\epsilon$.  If $\Pi(\cdot)$ is proper, then 
	$\Pi(D_\epsilon)\leq 1$.
	
\end{proof}

The previous two lemmas imply the posterior is converging to zero in a 
restricted parameter space.

\begin{lemma}\label{lem:goto0}
	Suppose Assumptions 4,  and 
	6 hold. Then for each $l\in\{1,2,...,L(k)\}$, there exists 
	a compact $G_l$ such that 
	\[\lim\limits_{n\rightarrow\infty}\Pi(V_{ln}\cap 
	G_l^c|\mathbf{Y}_1,...,\mathbf{Y}_n)=0.\]
\end{lemma}
\begin{proof}
	Let $\epsilon$ from Lemma \ref{lem:lrgrthn} equal $\frac{u_i}{4}$ from 
	Lemma \ref{lem:existu}. Then
	\begin{align*}
		\int_{\Theta} 
		e^{\sum_{i=1}^{n}log\left(\frac{f_{\bm{\tau}}(\mathbf{Y}_i|\alpha,\beta,	
				1)}{f_{\bm{\tau}}(\mathbf{Y}_i|\alpha_0,\beta_0, 				
				1)}
			\right)}d\Pi(\alpha,\beta) &\geq \int_{D_\epsilon} 
		e^{\sum_{i=1}^{n}log\left(\frac{f_{\bm{\tau}}(\mathbf{Y}_i|\alpha, 
				\beta,	1)}{f_{\bm{\tau}}(\mathbf{Y}_i|\alpha_0,\beta_0, 
				1)}\right)}d\Pi(\alpha,\beta) \\
		&\geq e^{-n\epsilon}d\Pi(D_\epsilon)
	\end{align*}
	
	Then 
	$\lim\limits_{n\rightarrow\infty}\int_{\Theta} 
	e^{\sum_{i=1}^{n}log\left(\frac{f_{\bm{\tau}}(\mathbf{Y}_i|\alpha,\beta, 
			1)}{f_{\bm{\tau}}(\mathbf{Y}_i|\alpha_0,\beta_0, 
			1)}\right)}d\Pi(\alpha,\beta)e^{nu_j/2}=\infty$ and 
	
	$\lim\limits_{n\rightarrow\infty}\int_{V_{jn}\cap G_j^c} 
	e^{\sum_{i=1}^{n}log\left(\frac{f_{\bm{\tau}}(\mathbf{Y}_i|\alpha,\beta, 
			1)}{f_{\bm{\tau}}(\mathbf{Y}_i|\alpha_0,\beta_0, 
			1)}\right)}d\Pi(\alpha,\beta)e^{nu_j/2}=0$.
\end{proof}
The next proposition bounds the expected value of the numerator, 
$E[I_n(V_{1n}\cap G)^d]$, and the denominator, $I_n(\Theta)$, of the 
posterior. Define $B_{in} = -E \left[ 
\log\left(\frac{f_{\bm{\tau}}(\mathbf{Y}_i|\alpha,\beta, 
	1)}{f_{\bm{\tau}}(\mathbf{Y}_i|\alpha_0,\beta_0, 
	1)}\right)\right]$.

\begin{lemma}\label{lem:2ineq}
	Suppose Assumptions 3a and 4 hold.  Define \\
	$D_{\delta_n} =  \left\{(\alpha,\beta):|\alpha-\alpha_0| < 
	\frac{\frac{1}{k}\delta_n}{1+c_\Gamma c_y}, |\beta-\beta_{0}| < 
	\frac{\frac{1}{k}\delta_n}{1+c_\Gamma c_y}\mathbf{1}_{k-1} \text{ 
		componentwise}\right\}$.  Then for 
	$(\alpha,\beta)\in 
	D_{\delta_n}$
	
	1. There exists a $\delta_n\in(0,1)$ and fixed $R>0$ such that
	\[ E\left[\left(\int_{V_{1n}\cap G} \prod_{i=1}^{n} 
	\frac{f_{\bm{\tau}}(\mathbf{Y}_i|\alpha,\beta, 
		1)}{f_{\bm{\tau}}(\mathbf{Y}_i|\alpha_0,\beta_0, 
		1)} d\Pi(\alpha,\beta) \right)^d\right] \leq 
	e^{d\sum_{i=1}^{n}B_{in}}e^{nd\delta_n}R^2/\delta_n^2 \]

	2.
	\[\int_{\Theta} \prod_{i=1}^{n} 
	\frac{f_{\bm{\tau}}(\mathbf{Y}_i|\alpha,\beta, 
		1)}{f_{\bm{\tau}}(\mathbf{Y}_i|\alpha_0,\beta_0, 
		1)} d\Pi(\alpha,\beta) \geq e^{-n\delta_n}\Pi(D_{\delta_n}) \]
\end{lemma}
\begin{proof}
	
	From Lemma \ref{lem:hypercube} and \ref{lem:existdelta} $	
	E\left[\left(\int_{W_{1n}\cap 
		G}\prod_{i=1}^{n}\frac{f_{\bm{\tau}}(\mathbf{Y}_i|\alpha,\beta, 
		1)}{f_{\bm{\tau}}(\mathbf{Y}_i|\alpha_0,\beta_0, 
		1)}d\Pi(\alpha,\beta)\right)^d\right]$
	
	\begin{align*}
		&\leq 
		\sum_{j=1}^{J(\delta_n)}\left[E\left[\left(\prod_{i=1}^{n}\frac{f_{u,
				\tau}(\mathbf{Y}_i
			|\alpha_j,\beta_j,
			1)}{f_{\bm{\tau}}(\mathbf{Y}_i|\alpha_0,\beta_0, 
			1)}\right)^d\right]e^{nd\delta_n}\Pi(A_j)^d\right]\\
		&\leq \sum_{j=1}^{J(\delta_n)}\left[e^{-d\sum_{i=1}^{n}B_{in} 
		}e^{nd\delta_n}\Pi(A_j)^d\right]\\
		&\leq e^{-d\sum_{i=1}^{n}B_{in}}e^{nd\delta_n}J(\delta_n)
	\end{align*}
	
	Since $G$ is compact, $R$ can be chosen large enough so that 
	$J(\delta_n)\leq R^2/\delta_n^2$.  Line 2.\ is from Lemma \ref{lem:goto0}. 
\end{proof}

The proof of Theorem 1 is below.\footnote{I would like to thank Karthik Sriram for help with the proof
	improper prior case.}

\begin{proof}
	Suppose $\Pi$ is proper. Lemma \ref{lem:existu} shows we can focus on the 
	case $W_{1n}\cap 
	G$. Set $\Delta_n = \Delta$ and $\delta_n=\delta$. Then from Lemma 
	\ref{lem:2ineq},  
	there exists a $d\in(0,1)$ such that for sufficiently large $n$
	\begin{align*}
		E\left[(\Pi(V_{1n}\cap G|\mathbf{Y}_1,...,\mathbf{Y}_n))^d\right] &\leq 
		\frac{R^2}{\delta^2(\Pi(V_\delta))^d} 
		e^{-d\sum_{i=1}^{n}B_{in}}e^{2nd\delta} \\
		&\leq 
		\frac{R^2}{\delta^2(\Pi(V_\delta))^d} 
		e^{-\frac{1}{2}dn\lim\limits_{m\rightarrow\infty}\frac{1}{m}\sum_{i=1}
			^{m}B_{im}}e^{2nd\delta}
	\end{align*}
	Chose $\delta = 
	\frac{1}{8}\lim\limits_{m\rightarrow\infty}\frac{1}{m}\sum_{i=1}^{m}B_{im}$ 
	and 
	note that $C' = \frac{R^2}{\delta^2(\Pi(V_\delta))^d}$ is a fixed 
	constant.  
	Then $E\left[(\Pi(V_{1n}\cap G|\mathbf{Y}_1,...,\mathbf{Y}_n))^d\right] 
	\leq 
	C'e^{-nd\delta/4}$.  
	Since 
	$\lim\limits_{n\rightarrow\infty}\sum_{n=1}^{\infty}C'e^{-nd\delta/4}<\infty
	$ 
	then the Markov inequality and Borel Cantelli imply posterior consistency 
	a.s..
	
	Now suppose the prior is improper but admits a proper posterior. Consider 
	the posterior from the first observation $\Pi(\cdot|\mathbf{Y}_1)$. Under 
	Assumption 3b, $\Pi(\cdot|\mathbf{Y}_1)$ is proper. 
	Assumption 5 ensures 
	that 
	$f_{\bm{\tau}}(\mathbf{Y}_i|\alpha_0,\beta_0, 1)$ dominates $p_0$. Thus the 
	formal 
	posterior exists on a set of $\mathbf{P}$ measure 1. Further, 
	$\Pi(U|\mathbf{Y}
	_1)>0$ for some open $U$ containing $(\alpha_0,\beta_0)$. Thus $\Pi(\cdot|
	\mathbf{Y}_1)$ can be used as a proper prior on the likelihood containing $
	\mathbf{Y}_2,...,\mathbf{Y}_n$ which produces a posterior equivalent to the 
	original $\Pi(\cdot|\mathbf{Y}_1,...,\mathbf{Y}_n)$ and thus the same 
	argument 
	above using a proper prior can be applied to the posterior 
	$\Pi(\cdot|\mathbf{Y}
	_2,...,\mathbf{Y}_n)$ using $\Pi(\cdot|\mathbf{Y}_1)$ as a proper prior.
\end{proof}

\section{Proof of Theorem 2}

Let $\mathbf{Z}'\mathbf{Z} = r_\tau^2$ represent the spherical $\tau$-Tukey depth contour $T_\tau$ and $\mathbf{u}'\mathbf{Z}=d_{\tau}$ represent the $\lambda_{\bm{\tau}}$ hyperplane where $d_{\tau} = \alpha_{\tau} + \beta_{\tau \mathbf{x}}\mathbf{X}$.

1) Let $\hat{\mathbf{Z}}$ represent the point of tangency between $T_\tau$ and $\lambda_{\bm{\tau}}$. Then the normal vector to $\lambda_{\bm{\tau}}$ is $\mathbf{u}$. Then there exists a $c$ such that $\hat{\mathbf{Z}}=c\mathbf{u}$. Let $c=\frac{-r}{\sqrt{\mathbf{u}'\mathbf{u}}}$, then $\hat{\mathbf{Z}}'\mathbf{u}=-r_\tau=d_{\tau}$. Thus $\sqrt{r^2_\tau} = |\alpha_{\tau} + \beta_{\tau \mathbf{x}}\mathbf{X}|$.

2) Let $\tilde{\mathbf{Z}}$ represent a point on $T_\tau$ and $\tilde{\mathbf{u}}=\frac{\tilde{\mathbf{Z}}}{\sqrt{\tilde{\mathbf{Z}}'\tilde{\mathbf{Z}}}}$. Then $\tilde{\mathbf{u}}'\tilde{\mathbf{u}}=1$ implying $\tilde{\mathbf{u}}\in\mathcal{S}^{k-1}$. Note the normal of $\lambda_{\tilde{\tau}}$ is $\tilde{\mathbf{u}}$ which is a scalar multiple of $\tilde{\mathbf{Z}}$. Thus there exists a $\mathbf{u}$ such that $\lambda_{\bm{\tau}}$ is tangent to $T_\tau$ at every point on $T_\tau$.

3) Let $\mathbf{u}\in\mathcal{S}^{k-1}$ then the normal of $\lambda_{\bm{\tau}}$ is $\mathbf{u}$. Let $\mathbf{Z} = d_\tau \mathbf{u}$, which is normal to $\lambda_{\bm{\tau}}$. Further $\mathbf{Z}'\mathbf{Z} = d_\tau^2\mathbf{u}'\mathbf{u} = d_\tau^2$ is a point on $T_\tau$. Thus there is a point on $\lambda_{\bm{\tau}}$  that is tangent to $T_\tau$ for every $\mathbf{u}\in\mathcal{S}^{k-1}$.



\section{Non-zero centered prior: second approach}

The second approach is to investigate the implicit prior in the untransformed 
response space of $Y_2$ against $Y_1$, $\mathbf{X}$ and 
an 
intercept. Denote $\mathbf{\Gamma}_\mathbf{u}= [u^\perp_1,u^\perp_2]'$.  Note 
that $
\mathbf{Y}_{\mathbf{u}i} = \beta_{{\bm{\tau}}\mathbf{y}}\mathbf{Y}_{\mathbf{u}
	i}^\perp
+
\beta_{{\bm{\tau}}\mathbf{x}}'\mathbf{X}_i+\alpha_{\bm{\tau}}$ can be rewritten 
as
\begin{align*}
	Y_{2i} &= \frac{1}{u_2 - \beta_{{\bm{\tau}}\mathbf{y}}u_2^\perp}
	\left((\beta_{{\bm{\tau}}
		\mathbf{y}}u_1^\perp - u_1) Y_{1i} + 
	\beta_{{\bm{\tau}}\mathbf{x}}'\mathbf{X}_i+\alpha_{\bm{\tau}}\right)\\
	&=\phi_{{\bm{\tau}} y} Y_{1i} + \phi_{{\bm{\tau}}\mathbf{x}}'\mathbf{X}_i+
	\phi_{{\bm{\tau}} 1}
\end{align*}

The interpretation of $
\phi_{\bm{\tau}}$ is fairly straight forward since the equation is in slope-intercept form. It can be verified that $
\phi_{{\bm{\tau}} y} = \phi_{{\bm{\tau}} y}(\beta_{{\bm{\tau}}\mathbf{y}}) = 
\frac{\beta_{{\bm{\tau}}\mathbf{y}}u_1^\perp - 
	u_1}{u_2 - \beta_{{\bm{\tau}}\mathbf{y}}u_2^\perp} = \frac{1}{u_1(u_2^\perp
	\beta_{{\bm{\tau}}
		\mathbf{y}} -u_2)}+\frac{u_2}{u_1}$ for $\beta_{{\bm{\tau}} y} \neq 
\frac{u_2}
{u_2^\perp}$ and $u_1 \neq 0$.  Suppose prior $\theta_{\bm{\tau}} = 
[\beta_{{\bm{\tau}}\mathbf{y}},\beta_{{\bm{\tau}}\mathbf{x}}',
\alpha_{\bm{\tau}}]'\sim 
F_{\theta_
	{\bm{\tau}}}(\theta_{\bm{\tau}})$ with support $
\Theta_{{\bm{\tau}}}$. If $F_{\beta
	\mathbf{y}}
$ is a continuous distribution, the density of $\phi_{\bm{\tau}}$ is 
\[f_{\phi_{{\bm{\tau}} y}}=f_{\beta_{{\bm{\tau}}\mathbf{y}}} 
(\phi_{{\bm{\tau}} y}^{-1}(\beta_{{\bm{\tau}}\mathbf{y}})) \left|
\frac{d}{d\beta_{{\bm{\tau}}
		\mathbf{y}}}
\phi_{{\bm{\tau}} y}^{-1}(\beta_{{\bm{\tau}}\mathbf{y}})\right| = 
f_{\beta_{{\bm{\tau}}
		\mathbf{y}}}
\left(\frac{1}{u_2^\perp}\left(\frac{1}{u_1\phi_{{\bm{\tau}} y} - u_2}
+u_2\right)\right)
\left| \frac{u_1}{u_2^\perp ( u_1 \phi_{{\bm{\tau}} y} - u_2)^2}\right|
\] with 
support not containing $\left\{-\frac{u_1^\perp}{u_2^\perp}\right\}$, for 
$u_2^\perp\neq 0$.

If $\beta_{{\bm{\tau}}\mathbf{y}}\sim N(\underline{\mu}_
{{\bm{\tau}} y},\underline{\sigma}_{{\bm{\tau}} y}^2),$ then the density 
of $
\phi_{{\bm{\tau}} y}$ is 
a 
shifted reciprocal Gaussian with density \[f_{\phi_{{\bm{\tau}} y}}
(\phi| 
\underline{a},\underline{b}^2) = 
\frac{1}
{\sqrt{2\pi 
		\underline{b}_{\bm{\tau}}^2}(\phi-u_2/u_2^\perp)^2}exp\left(-\frac{1}
{2\underline{b}_
	\tau^2}	\left(\frac{1}{\phi - u_2/u_2^\perp} - \underline{a}
\right)^2\right).\]  
The 
parameters are $\underline{a} = \underline{\mu}_{\bm{\tau}} u_1 u_2^
\perp - 
u_1u_2$ and $\underline{b} = u_1u_2^\perp \underline{\sigma}_{\bm{\tau}}
$.  
The 
moments of $\phi_{{\bm{\tau}} y}$ do not exist \citep{robert91}.  The 
density  is bimodal with 
modes at

\[m_1 = \frac{-\underline{a} + \sqrt{\underline{a}^2 + 8 
		\underline{b}^2}}{4 
	\underline{b}^2} + \frac{u_2}{u_2^\perp} \text{ and } m_2 = \frac{-
	\underline{a} - 
	\sqrt{\underline{a}^2 + 8 
		\underline{b}^2}}{4 
	\underline{b}^2} + \frac{u_2}{u_2^\perp} .\]

Elicitation can be tricky since moments do not exist. However, elicitation can rely on the 
modes 
and their relative heights

\[ \frac{f_{\phi_{{\bm{\tau}} y}}(m_1| 
	\underline{a},\underline{b}^2)}{f_{\phi_{{\bm{\tau}} y}}(m_2| 
	\underline{a},\underline{b}^2)} = \frac{\underline{a}^2+ \underline{a}
	\sqrt{\underline{a}^2+8\underline{b}^2} + 4\underline{b}^2}{\underline{a}^2- 
	\underline{a}\sqrt{\underline{a}^2+8\underline{b}^2}+ 4\underline{b}^2} exp
\left(\frac{\underline{a}\sqrt{\underline{a}^2 + 8\underline{b}^2}}
{\underline{b}^4}\right) \]

Plots of the reciprocal Gaussian are shown in Figure \ref{fig:prior3}. The left plot presents densities of the reciprocal Gaussian for several hyper-parameter values. The right plot shows contours of the log relative heights of the modes over the set $(\underline{a}, \underline{b}^2)\in [-5,5]\times [ 10,
100]$.

\begin{figure}
	\centering
	\includegraphics[width=1\linewidth]{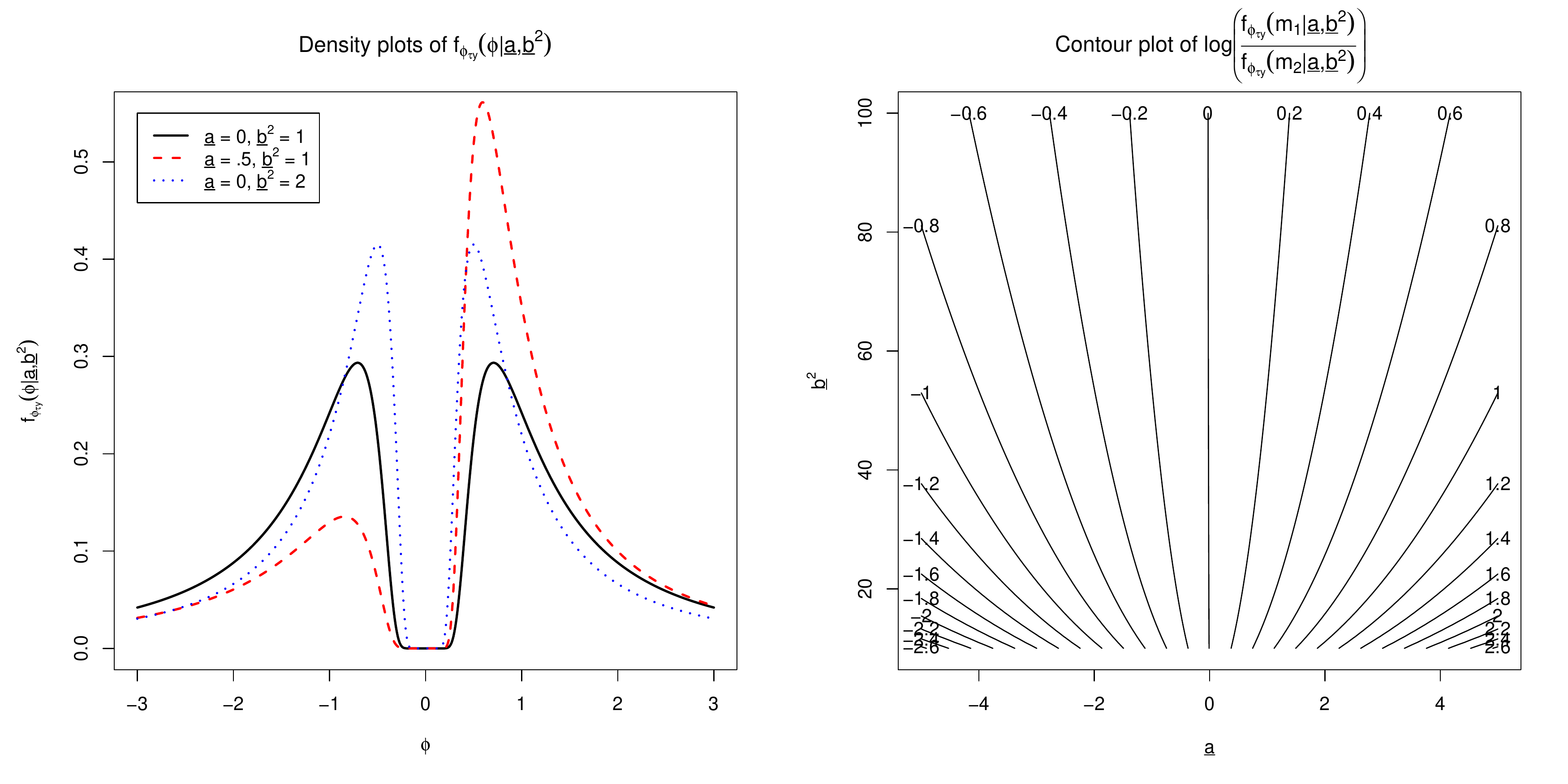}
	\caption{(left) Density of $f_{\phi_{{\bm{\tau}} y}}(\phi| 
		\underline{a},\underline{b}^2)$ for hyper parameters $\underline{a}=0$, 
		$\underline{b}^2=1$ (solid black), $\underline{a}=0.5$, $\underline{b}
		^2=1$ 
		(dash red), $\underline{a}=0$, $\underline{b}^2=2$ (dotted blue).  
		(right) A 
		contour plot showing the log relative heights of the modes at $m_1$ 
		over $m_2$ 
		over the set $(\underline{a}, \underline{b}^2)\in [-5,5]\times [ 10,
		100]$. }
	\label{fig:prior3}
\end{figure}

The distribution of $\phi_{{\bm{\tau}} \mathbf{x}}$ and $\phi_{{\bm{\tau}} 1}$ 
are ratio 
normals. The implied prior on $\phi_{{\bm{\tau}}1}$ is discussed.  The 
distribution 
of $\phi_{{\bm{\tau}} \mathbf{x}}$ will follow by analogy.  The implied 
intercept $
\phi_{{\bm{\tau}} 1} = \frac{\alpha_{\bm{\tau}}}{u_2 - \beta_{{\bm{\tau}}
		\mathbf{y}}u_2^\perp}$ is a 
ratio of normals distribution.  The ratio of normals
distributions can always be expressed as a location scale shift of $R = 
\frac{Z_1+
	a}{Z_2+b}$ where $Z_i\overset{iid}{\sim}N(0,1)$ for $i\in\{1,2\}$. That 
is, there exist 
constants $c$ and $d$ such that $\phi_{{\bm{\tau}} 1} = cR + 
d$ \citep{hinkley69,hinkley70,marsaglia65,marsaglia06}.\footnote{Proof: let $W_i\sim N(\theta_i,
	\sigma^2_i)$ for $i\in\{1,2\}$ with $corr(W_1,W_2)=\rho$.  Then $\frac{W_1}
	{W_2}
	= \frac{\sigma_1}{\sigma_2}\sqrt{1-\rho^2}\left(\frac{\frac{\theta_1}
		{\sigma_1}
		+Z_1}{\frac{\theta_2}{\sigma_2}+Z_2} + \frac{\rho}{\sqrt{1-\rho^2}}
	\right)$ 
	where 
	$Z_i\sim N(0,
	1)$ for $i\in\{1,2\}$ with $corr(Z_1,Z_2)=0$.  Thus $a = \frac{\theta_1}
	{\sigma_1}$, $b = \frac{\theta_2}{\sigma_2}$, $c = 
	\frac{\sigma_1}
	{\sigma_2}\sqrt{1-\rho^2}$ and $d=c\frac{\rho}{\sqrt{1-\rho^2}}$ where $
	\theta_1 
	= \underline{a}_{{\bm{\tau}} 1}$, $\theta_2 = u_2-\underline{a}_{{\bm{\tau}} 
		y}
	u_2^\perp $, $
	\sigma_1 = \underline{b}_{{\bm{\tau}} 1}$ and $\sigma_2 = \underline{b}
	_{{\bm{\tau}} y} 
	u_2^\perp$.} The density of $\phi_{{\bm{\tau}} 1}$ is
\[f_{\phi_{{\bm{\tau}} 1}}(\phi|\underline{a},
\underline{b}) = \frac{e^{-\frac{1}{2}(\underline{a}^2 + \underline{b}^2)}}
{\pi(1+\phi^2)}\left[1 + c e^{\frac{1}{2} c^2} \int_{0}^{c} e^{-\frac{1}{2}t^2} 
dt \right], \text{ where } c = \frac{\underline{b} + \underline{a} \phi}
{\sqrt{1+\phi^2}} .\]  Note, when $\underline{a} = \underline{b} = 0$, 
then the distribution reduces to the standard Cauchy distribution.  The 
ratio of normals distribution, like the 
reciprocal Gaussian distribution, has no moments and can be bimodal. Focusing on the positive quadrant of $(\underline{a},\underline{b})$, if $\underline{a}\leq 1$ and $ \underline{b}\geq 0$ then ratio of normals distribution is unimodal. If 
$\underline{a} \gtrsim 2.256058904$ then the ratio of normals distribution is bimodal.  There is a curve that separates the 
unimodal and bimodal regions.\footnote{The curve is approximately $\underline{b} = 
	\frac{18.621 - 
		63.411\underline{a}^2 - 54.668\underline{a}^3 + 17.716\underline{a}^4 - 
		2.2986\underline{a}^5}{2.256058904 - \underline{a}}$ for $\underline{a}\leq 2.256...
	$.} Figure 
\ref{fig:priorintmod} shows three plots for the density of the ratio of normals distribution and the bottom right plot shows the regions where the density is unimodal and bimodal. The unimodal region is to the left of the presented curve and the bimodal region is to the right of the presented curve.  If 
the ratio of normals distribution is bimodal, one mode will be to the left of $-\underline{b}/
\underline{a}$ and the other to the right of $-\underline{b}/
\underline{a}$.  The left mode tends to be much lower 
than the right mode for positive $(\underline{a},\underline{b})$.  Unlike the 
reciprocal Gaussian, closed form solutions for the 
modes do not exist.  The distribution is approximately elliptical with 
central tendency $\mu=\frac{\underline{a}}{1.01 \underline{b}-0.2713}$ and squared dispersion $\sigma^2 =\frac{\underline{a}^2+1}{\underline{b}^2 + 
	0.108\underline{b} - 3.795} - \mu^2$ when $\underline{a}<2.256$ and $4<
\underline{b}$ \citep{marsaglia06}. 

\begin{figure}
	\centering
	\includegraphics[width=0.9\linewidth]{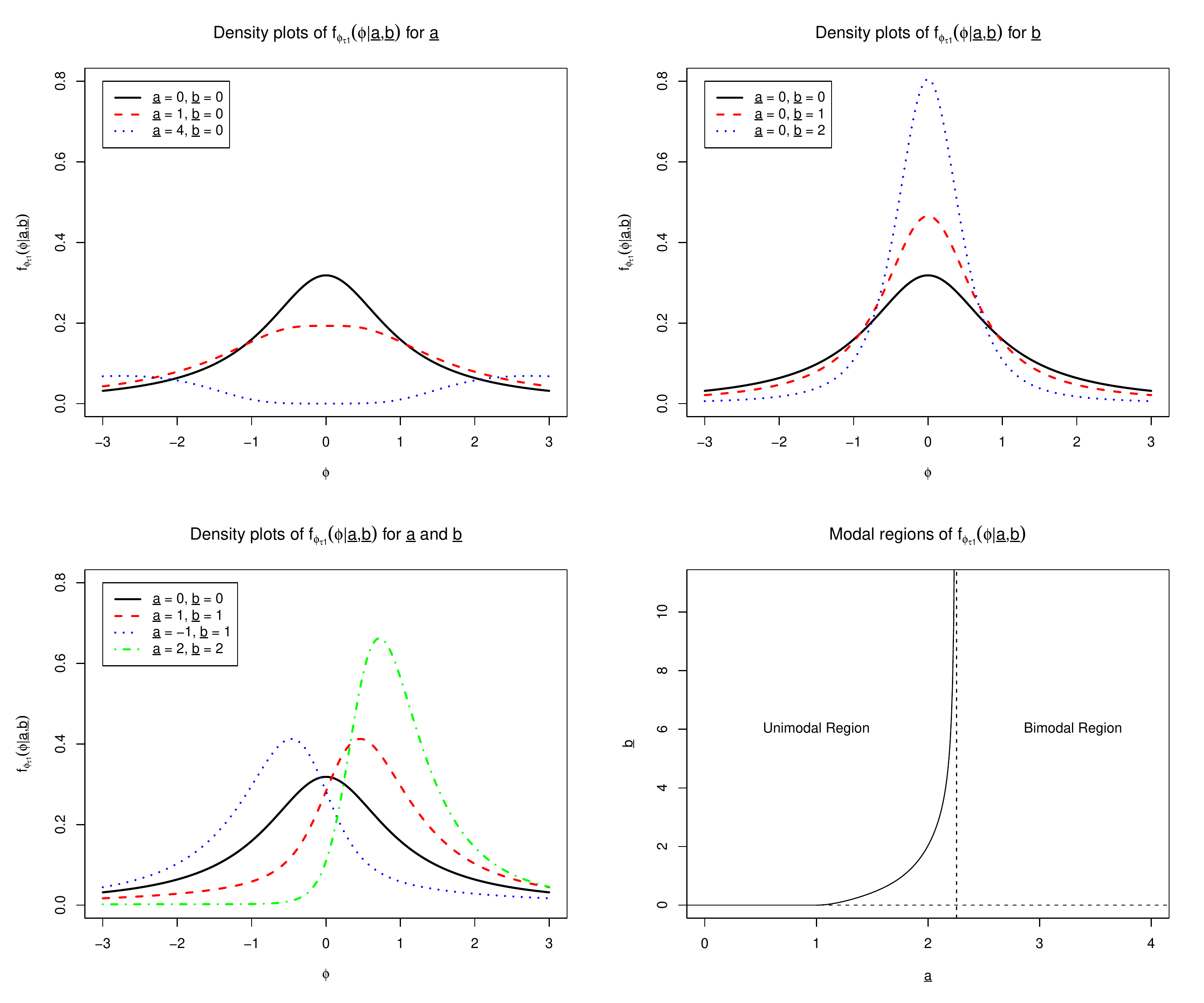}
	\caption{The top two plots and the bottom left plot show the density of the 
		ratio normal distribution with parameters $(\underline{a},\underline{b})
		$.  The 
		top left plot shows the density for different values of $\underline{a}$ 
		with $
		\underline{b}$ fixed at zero.  The parameters $(\underline{a},
		\underline{b}) = 
		(1,0)$ and $(4,0)$ result in the same density as $(\underline{a},
		\underline{b}) 
		= (-1,0)$ and $(-4,0)$. The 
		top right plot shows the density for different values of $\underline{b}$ 
		with $
		\underline{a}$ fixed at zero.  The parameters $(\underline{a},
		\underline{b}) = 
		(0,1)$ and $(0,2)$ result in the same density as $(\underline{a},
		\underline{b}) 
		= (0,-1)$ and $(0,-2)$.  The 
		bottom left plot shows the density for different values of $
		\underline{a}$ and $
		\underline{b}$.  The parameters $(\underline{a},\underline{b}) = 
		(1,1)$, $(-1,1)$and $(2,2)$ result in the same density as $
		(\underline{a},
		\underline{b}) 
		= (-1,-1)$, $(1,-1)$and $(-2,-2)$.  The bottom right graph shows the 
		regions of 
		the positive quadrant of the parameter space where the density is either 
		bimodal 
		or unimodal. }
	\label{fig:priorintmod}
\end{figure}

\section{Simulation} \label{app:simsubgrad}
\subsection{Convergence of subgradient conditions}
This section verified convergence of subgradient conditions (\ref{eq:subgrad1}) and (\ref{eq:subgrad2}). For DGPs 1-4, 
$E[\mathbf{Y}_{\mathbf{u}}
^\perp]= \mathbf{0}_2$ and for DGP 4, $E[\mathbf{X}]=0$. Define $\hat{H}_{{\bm{\tau}}}^{-}$ to 
be the empirical lower halfspace where the parameters in 
(\ref{eq:lwhalfsp}) are replaced with their Bayesian estimates.  Convergence of the first 
subgradient 
condition (\ref{eq:subgrad1}) requires
\begin{equation}\label{eq:subgrad1sim}
\frac{1}{n}\sum_{i=1}^{n}1_{(\mathbf{Y}_{i}\in \hat{H}_{{\bm{\tau}}}^{-})} 
\rightarrow \tau.
\end{equation}
Computation of  
$1_{(\mathbf{Y}_{i}\in \hat{H}_{{\bm{\tau}}}^{-})}$ is simple. Convergence of the second subgradient condition (\ref{eq:subgrad2}) requires

\begin{align}
	\label{eq:subgrad2sim}\frac{1}{n}\sum_{i=1}^{n}\mathbf{Y}_{ 
		\mathbf{u}i}^
	\perp 1_{(\mathbf{Y}
		_{i}\in \hat{H}_{\bm{\tau}}^{-})} &\rightarrow \tau E[\mathbf{Y}_{ 
		\mathbf{u}}^\perp]\\
	\intertext{ and }
	\label{eq:subgrad3sim}	
	\frac{1}{n}\sum_{i=1}^{n}\mathbf{X}_{i}1_{(\mathbf{Y}
		_{i}\in \hat{H}_{\bm{\tau}}^{-})} &\rightarrow \tau E[\mathbf{X}].
\end{align}
Similar to the first subgradient condition, computation of $\mathbf{Y}
_{\mathbf{u}i}^\perp1_{(\mathbf{Y}_{i}\in \hat{H}_{{\bm{\tau}}}^{-})}$ and $
\mathbf{X}_{ i}
1_{(\mathbf{Y}_{i}\in \hat{H}_{\bm{\tau}}^{-})}$ is simple.

Tables \ref{tab:simres1}, \ref{tab:simres2} and \ref{tab:simres3} show the 
results from the simulation.  Tables  \ref{tab:simres1} and \ref{tab:simres2} 
show the Root Mean Square Error (RMSE) of (\ref{eq:subgrad1sim}) and 
(\ref{eq:subgrad2sim}). Table 
\ref{tab:simres1} is using directional vector $\mathbf{u} = (1/\sqrt{2},1/
\sqrt{2})$ and Table \ref{tab:simres2} is using directional vector $\mathbf{u} 
= 
(0,1)$.  For Tables  \ref{tab:simres1} and \ref{tab:simres2} the first three rows show the RMSE for the first 
subgradient condition (\ref{eq:subgrad1sim}). The last three rows show the RMSE 
for 
the second subgradient condition (\ref{eq:subgrad2sim}).  The second column, $n$, 
is 
the sample size. The next five columns are the DGPs previously described.  It is clear that as sample size increases the RMSEs are decreasing, 
showing the convergence of the subgradient conditions.

\begin{table}[htb]
	\centering
	\begin{tabular}{rr|cccc}
		& & \multicolumn{4}{c}{Data Generating Process}\\ \hline
		& $n$ & 1 & 2 & 3 & 4 \\ 
		\hline
		& $10^2$ & 4.47e-02 & 2.91e-02 & 1.52e-02 & 1.75e-02 \\ 
		Sub Grad 1 & $10^3$ & 5.44e-03 & 4.59e-03 & 2.48e-03 & 2.60e-03 \\ 
		& $10^4$ & 9.29e-04 & 8.66e-04 & 5.42e-04 & 5.12e-04 \\ 
		\hline
		& $10^2$ & 6.34e-03 & 1.43e-02 & 4.34e-02 & 7.06e-02 \\ 
		Sub Grad 2 & $10^3$ & 2.01e-03 & 3.29e-03 & 1.32e-02 & 2.05e-02 \\ 
		& $10^4$ & 5.82e-04 & 8.00e-04 & 3.59e-03 & 4.91e-03 \\ 
	\end{tabular}
	\caption{Unconditional model RMSE of subgradient conditions for $\mathbf{u}
		=(1/\sqrt{2},1/\sqrt{2})$} 
	\label{tab:simres1}
\end{table}

\begin{table}[htb]
	\centering
	\begin{tabular}{rr|cccc}
		& & \multicolumn{4}{c}{Data Generating Process}\\ \hline
		& $n$ & 1 & 2 & 3 & 4 \\ 
		\hline
		& $10^2$ & 2.02e-02 & 1.89e-02 & 1.16e-02 & 1.36e-02 \\ 
		Sub Grad 1 & $10^3$ & 3.38e-03 & 3.61e-03 & 1.96e-03 & 1.98e-03 \\ 
		& $10^4$ & 7.71e-04 & 9.32e-04 & 3.87e-04 & 4.68e-04 \\ 
		\hline
		& $10^2$ & 9.74e-03 & 1.35e-02 & 2.59e-02 & 2.29e-02 \\ 
		Sub Grad 2 & $10^3$ & 2.08e-03 & 3.24e-03 & 7.11e-03 & 6.51e-03 \\ 
		& $10^4$ & 6.15e-04 & 9.89e-04 & 2.01e-03 & 1.83e-03 \\ 
	\end{tabular}
	\caption{Unconditional model RMSE of subgradient conditions for $\mathbf{u}
		=(0,1)$} 
	\label{tab:simres2}
\end{table}

\begin{table}[htb]
	\centering
	\begin{tabular}{r|cc}
		& \multicolumn{2}{c}{Direction $\mathbf{u}$}\\ \hline
		$n$ & $(1/\sqrt{2},1/\sqrt{2})$ & $(0,1)$ \\ 
		\hline
		$10^2$ & 5.17e-02 & 5.17e-02 \\ 
		$10^3$ & 1.41e-02 & 1.41e-02 \\ 
		$10^4$ & 3.90e-03 & 3.90e-03 \\ 
		\hline
	\end{tabular}
	\caption{Unconditional model RMSE of covariate subgradient condition for DGP 4} 
	\label{tab:simres3}
\end{table}

Table \ref{tab:simres3} shows RMSE of (\ref{eq:subgrad3sim}) for the covariate subgradient condition of DGP 4.  The three rows show sample size and 
the two columns show direction.  It is clear that as sample size increases the 
RMSEs are decreasing, showing convergence of the subgradient conditions.

\section{Application}\label{app:star}
\subsection{Fixed-$\mathbf{u}$ hyperplanes}

Figure \ref{fig:Testfixu1} shows fixed-$\mathbf{u}$ $\lambda_{\bm{\tau}}$ hyperplanes for various $\tau$ 
along a fixed $\mathbf{u}$ 
direction with model (\ref{eq:appunc}). The values of $\tau$ are $
\{0.01, 0.05, 0.1, 0.2, 0.3, 0.4, 0.5, 0.6, 0.7, 0.8, 0.9, 0.95, 0.99\}$. Two directions are presented: $\mathbf{u} = (1/\sqrt{2},1/\sqrt{2})$ 
(left) and $
\mathbf{u} = (1,0)$ (right).  The direction vectors are represented by the 
orange arrows passing through the Tukey median (red dot).  The 
hyperplanes 
to the far left of either graph are for $\tau=0.01$. The hyperplanes along the direction of the arrow are for larger values of $
\tau$, ending with $\tau=0.99$ hyperplanes on the far right.  

\begin{figure}[H]
	\centering
	\includegraphics[width=.9\linewidth]{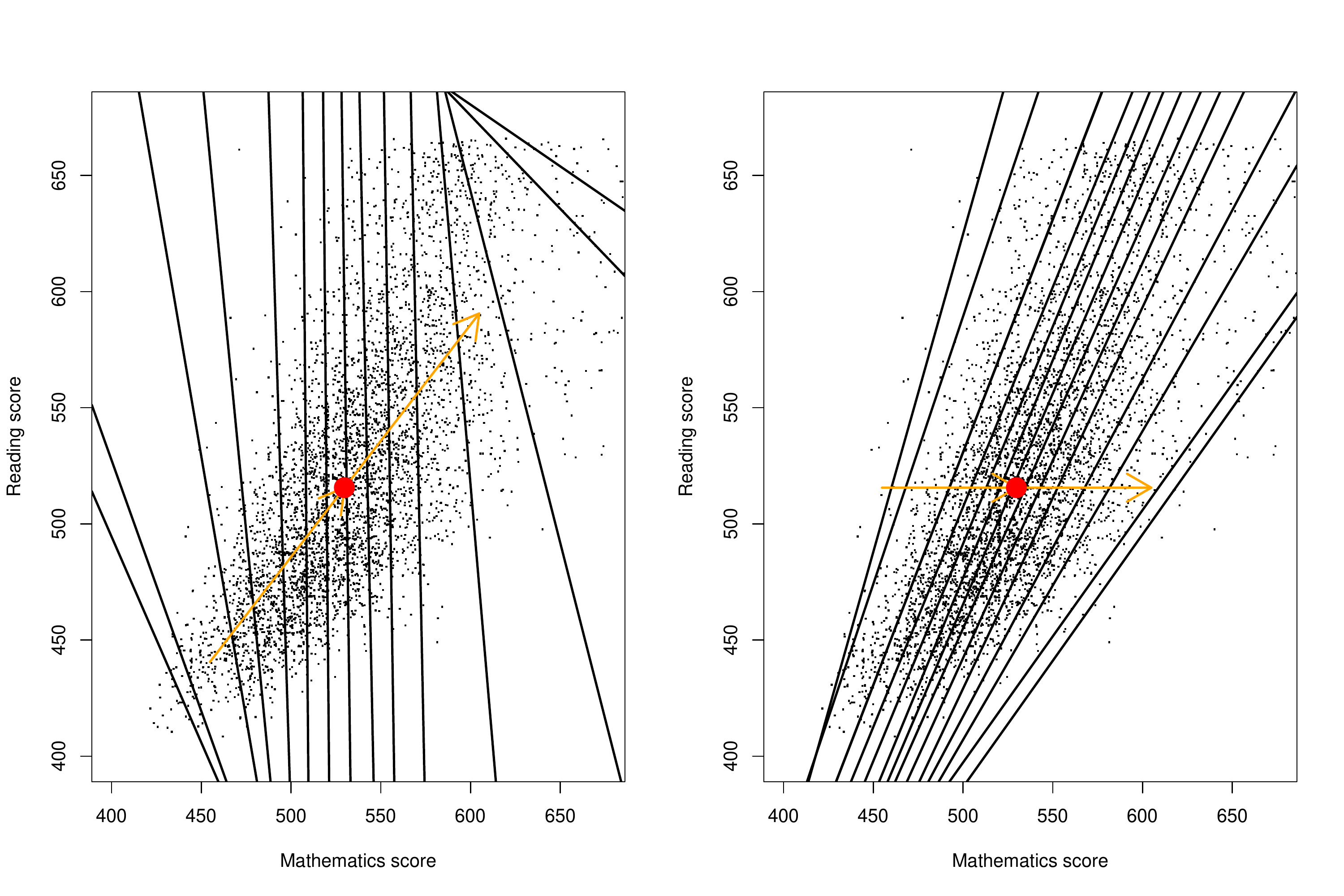}
	\caption[Fixed-u hyperplanes]{Left, fixed $\mathbf{u}=(1/\sqrt{2},1/\sqrt{2})$ 
		hyperplanes. 
		Right, fixed $\mathbf{u}=(1,0)$ hyperplanes.}
	\label{fig:Testfixu1}
\end{figure}

The left plot 
shows the hyperplanes initially tilt counter-clockwise for $\tau = 0.01$, tilt nearly 
vertical for $
\tau=0.5$ and then begin tilting counter-clockwise again for $\tau =0.99$.  The hyperplanes in the right plot are all almost parallel tilting slightly clockwise. To 
understand why this is happening, imagine traveling along the $\mathbf{u} = 
(1/\sqrt{2},1/
\sqrt{2})
$ vector through the Tukey median. Data can be thought of as a viscous liquid that the 
hyperplane must travel 
through.  
When the hyperplane hits a dense region of data, that part of the hyperplane is 
slowed down 
as it attempts to 
travel through it, resulting in the hyperplane tilting towards 
the 
region with less dense data. Since the density of the data changes as one 
travels through the $\mathbf{u} = (1/\sqrt{2},1/\sqrt{2})$ direction, the 
hyperplanes are tilting. However, the density of the data in the $
\mathbf{u} = (1,0)$ direction does not change much, so the tilt of the 
hyperplanes does not change.

\subsection{Sensitivity analysis}

\begin{figure}[H]
	\centering
	\includegraphics[width=0.7\linewidth]{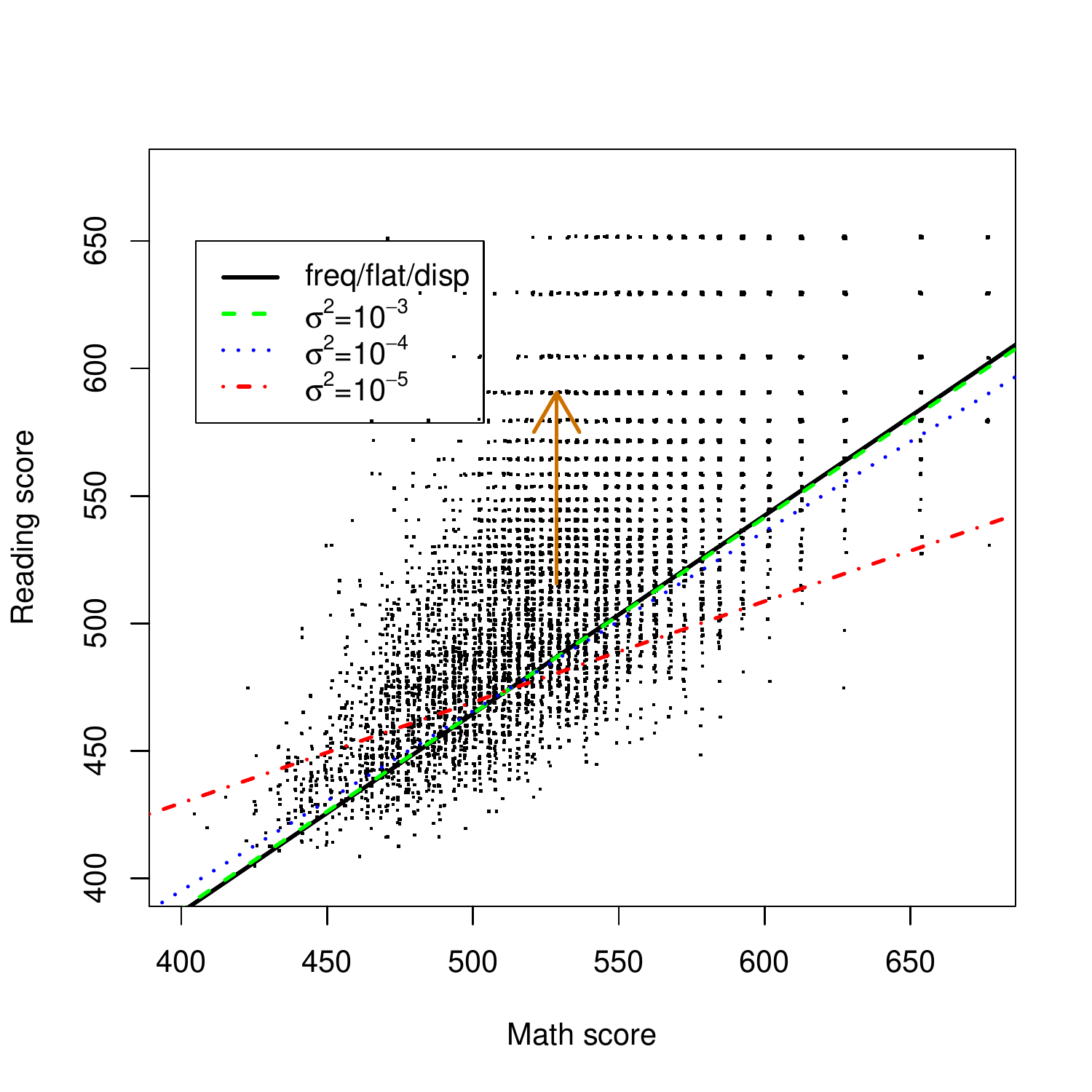}
	\caption[Prior influence ex-post]{Prior influence ex-post}
	\label{fig:testpr1}
\end{figure}

Figure \ref{fig:testpr1}  shows posterior sensitivity to different prior 
specifications of the location model with directional vector $\mathbf{u}=(0,1)$ 
pointing $90^\circ$ in 
the reading direction. The posteriors are compared against the frequentist estimate 
(solid 
black line). The first specification is the (improper) flat prior (i.e.\ 
Lebesgue 
measure) represented by the solid black line and cannot be visually 
differentiated from the frequentist estimate.  The rest of the specifications 
are proper priors with common mean, $\mu_{\theta_{\bm{\tau}}} = \mathbf{0}_2$.  
The dispersed prior has covariance $\Sigma_{\theta_{\bm{\tau}}} = 1000\mathbf{I}
_2$ and is represented by the solid black line and cannot be visually 
differentiated from the frequentist estimate or the estimate from the flat 
prior.  The next three priors have covariance matrices $
\Sigma_{\theta_{\bm{\tau}}} = diag(1000,\sigma^2)$ with $\sigma^2=10^{-3}$ 
(dashed green), $\sigma^2=10^{-4}$ (dotted blue) and $\sigma^2=10^{-5}$ (dash 
dotted red).  As the prior becomes more informative $\beta_{\bm{\tau}}$
converges to zero with resulting model $\hat{reading}_i = \alpha_{\bm{\tau}}$.

\bibliographystyle{chicago}

\bibliography{BayesmultquantbibRev1arXiv}

\end{document}